\pdfoutput=1
\documentclass[11pt]{article}
\usepackage[letterpaper,margin=1in]{geometry}
\usepackage[T1]{fontenc}
\usepackage{lmodern}
\usepackage{microtype}
\usepackage{amsmath,amssymb,amsthm,mathtools}
\usepackage{booktabs,array,multirow}
\usepackage{float}
\usepackage{graphicx}
\usepackage[numbers,sort&compress]{natbib}
\usepackage[dvipsnames]{xcolor}
\usepackage{url}
\usepackage{xspace}
\usepackage{algorithm}
\usepackage[noend]{algpseudocode}
\usepackage{longtable}
\usepackage{authblk}
\usepackage{hyperref}
\hypersetup{
  colorlinks=true,
  linkcolor=MidnightBlue,
  citecolor=MidnightBlue,
  urlcolor=MidnightBlue,
  pdftitle={Vertex Ordering for Exact Triangle Listing: Concave Geometry and Certified Rounding},
  pdfauthor={Kourosh Mirsohi, Melika Gorgi, Erfan Hajihashemi},
  pdfsubject={Vertex ordering for exact triangle listing},
  pdfkeywords={triangle listing, vertex ordering, concave minimization, orientation polytope, graphical zonotope, Min UnCut, configuration model, iterated local search}
}
\newtheorem{theorem}{Theorem}
\newtheorem{proposition}[theorem]{Proposition}
\newtheorem{corollary}[theorem]{Corollary}
\newtheorem{lemma}[theorem]{Lemma}
\theoremstyle{remark}

\newcommand{\Lam}{\Lambda}
\newcommand{\Lopt}[1]{\Lambda^{*}_{#1}}
\newcommand{\dout}{d^{+}}
\newcommand{\Nout}{N^{+}}
\newcommand{\smin}{\operatorname{smin}}
\newcommand{\Cpp}{C^{++}}
\newcommand{\Cpm}{C^{+-}}
\newcommand{\ZG}{Z_G}
\newcommand{\TILS}{T+ILS\xspace}
\newcommand{\DOVE}{DOVE\xspace}
\newcommand{\relint}{\operatorname{relint}}
\newcommand{\repo}{\url{https://github.com/Mirsohi/triangle-ordering}}
\newcommand{\repopath}[1]{\nolinkurl{#1}}
\makeatletter\g@addto@macro\UrlBreaks{\do\_}\makeatother

\newif\ifresultspending
\newcommand{\pend}{\textcolor{red}{\textbf{??}}}
\resultspendingfalse

\newcommand{\SuiteDoveMean}{4.36}
\newcommand{\SuiteTMean}{4.43}
\newcommand{\SuiteTIKMean}{5.70}

\newcommand{\SuiteTIKITR}{32/0/0}

\newcommand{\SuiteTIKvsDove}{31/1/0}

\newcommand{\HeldCount}{9 of 10}
\newcommand{\HeldDoveMean}{10.61}
\newcommand{\HeldTMean}{7.71}
\newcommand{\HeldTIKMean}{14.83}

\newcommand{\HeldTIKITR}{9/0/0}

\newcommand{\HeldTIKvsDove}{9/0/0}

\newcommand{\SuiteTIKCI}{[3.95, 7.67]\%}
\newcommand{\SuiteTIKvsDoveP}{<0.001}
\newcommand{\SuiteTIKvsT}{31/1/0}
\newcommand{\HeldTIKCI}{[9.53, 20.93]\%}
\newcommand{\HeldTIKvsDoveP}{0.0039}
\newcommand{\HeldTIKvsT}{9/0/0}

\providecommand{\PHAllCILSTBestMean}{}\renewcommand{\PHAllCILSTBestMean}{7.62}

\providecommand{\PHAllCILSTBestvsTILS}{}\renewcommand{\PHAllCILSTBestvsTILS}{16/1/24}

\providecommand{\PHAllColBestMean}{}\renewcommand{\PHAllColBestMean}{$-$13.53}

\providecommand{\PHAllCount}{}\renewcommand{\PHAllCount}{41}

\providecommand{\PHAllILSTMean}{}\renewcommand{\PHAllILSTMean}{7.70}

\providecommand{\PHAllILSTvsTILS}{}\renewcommand{\PHAllILSTvsTILS}{19/1/21}

\providecommand{\PHAllTILSMean}{}\renewcommand{\PHAllTILSMean}{7.71}

\providecommand{\PHAllTimeRatioTILS}{}\renewcommand{\PHAllTimeRatioTILS}{0.9905}

\providecommand{\PHAllWorkRatioTILS}{}\renewcommand{\PHAllWorkRatioTILS}{0.9197}

\providecommand{\PHHeldTimeRatioTILS}{}\renewcommand{\PHHeldTimeRatioTILS}{0.9615}

\providecommand{\NSCheck}{\pend}
\providecommand{\NSFrozenDoveMean}{\pend}
\providecommand{\NSFrozenTILSCI}{\pend}
\providecommand{\NSFrozenTILSITR}{\pend}
\providecommand{\NSFrozenTILSMean}{\pend}

\providecommand{\NSFrozenTMean}{\pend}
\providecommand{\NSGapMax}{\pend}

\providecommand{\NSGapMean}{\pend}
\providecommand{\NSHeldCount}{\pend}

\providecommand{\NSNeighSeeds}{\pend}

\providecommand{\NSReDoveMean}{\pend}
\providecommand{\NSReMaxOver}{\pend}
\providecommand{\NSReOverRuns}{\pend}
\providecommand{\NSReReused}{\pend}
\providecommand{\NSReRuns}{\pend}
\providecommand{\NSReTILSCI}{\pend}
\providecommand{\NSReTILSCappedCI}{\pend}
\providecommand{\NSReTILSCappedMean}{\pend}
\providecommand{\NSReTILSITR}{\pend}
\providecommand{\NSReTILSMean}{\pend}

\providecommand{\NSReTILSNoFinalMean}{\pend}
\providecommand{\NSReTILSvsDove}{\pend}
\providecommand{\NSReTILSvsDoveP}{\pend}
\providecommand{\NSReTILSvsT}{\pend}

\providecommand{\NSReTMean}{\pend}
\providecommand{\NSRefChanged}{\pend}

\providecommand{\NSTimeFrozenCI}{\pend}
\providecommand{\NSTimeFrozenGeo}{\pend}
\providecommand{\NSTimeFrozenWorkGeo}{\pend}

\providecommand{\NSTimeReCI}{\pend}
\providecommand{\NSTimeReGeo}{\pend}
\providecommand{\NSTimeReWorkGeo}{\pend}

\IfFileExists{supp_macros.tex}{
\providecommand{\NSHeldCount}{}\renewcommand{\NSHeldCount}{9}

\providecommand{\NSCheck}{}\renewcommand{\NSCheck}{passed for all 130 supplementary orders}

\providecommand{\NSNeighSeeds}{}\renewcommand{\NSNeighSeeds}{45 of 45}
\providecommand{\NSRefChanged}{}\renewcommand{\NSRefChanged}{8}

\providecommand{\NSGapMean}{}\renewcommand{\NSGapMean}{8.33}
\providecommand{\NSGapMax}{}\renewcommand{\NSGapMax}{18.82}

\providecommand{\NSFrozenTMean}{}\renewcommand{\NSFrozenTMean}{$-$1.08}
\providecommand{\NSFrozenDoveMean}{}\renewcommand{\NSFrozenDoveMean}{2.40}
\providecommand{\NSFrozenTILSMean}{}\renewcommand{\NSFrozenTILSMean}{7.22}

\providecommand{\NSFrozenTILSITR}{}\renewcommand{\NSFrozenTILSITR}{9/0/0}
\providecommand{\NSFrozenTILSCI}{}\renewcommand{\NSFrozenTILSCI}{[3.86, 11.09]\%}
\providecommand{\NSReTMean}{}\renewcommand{\NSReTMean}{4.66}
\providecommand{\NSReDoveMean}{}\renewcommand{\NSReDoveMean}{4.32}
\providecommand{\NSReTILSMean}{}\renewcommand{\NSReTILSMean}{7.50}

\providecommand{\NSReTILSITR}{}\renewcommand{\NSReTILSITR}{9/0/0}
\providecommand{\NSReTILSCI}{}\renewcommand{\NSReTILSCI}{[3.95, 11.53]\%}

\providecommand{\NSReTILSCappedMean}{}\renewcommand{\NSReTILSCappedMean}{7.43}
\providecommand{\NSReTILSCappedCI}{}\renewcommand{\NSReTILSCappedCI}{[3.91, 11.42]\%}
\providecommand{\NSReTILSNoFinalMean}{}\renewcommand{\NSReTILSNoFinalMean}{7.44}
\providecommand{\NSReRuns}{}\renewcommand{\NSReRuns}{27}
\providecommand{\NSReOverRuns}{}\renewcommand{\NSReOverRuns}{7}

\providecommand{\NSReReused}{}\renewcommand{\NSReReused}{1}
\providecommand{\NSReTILSvsDove}{}\renewcommand{\NSReTILSvsDove}{9/0/0}
\providecommand{\NSReTILSvsDoveP}{}\renewcommand{\NSReTILSvsDoveP}{0.0039}
\providecommand{\NSReTILSvsT}{}\renewcommand{\NSReTILSvsT}{9/0/0}

\providecommand{\NSReMaxOver}{}\renewcommand{\NSReMaxOver}{3.7\%}
\providecommand{\NSTimeFrozenGeo}{}\renewcommand{\NSTimeFrozenGeo}{1.0138}
\providecommand{\NSTimeFrozenCI}{}\renewcommand{\NSTimeFrozenCI}{[0.9751, 1.0561]}

\providecommand{\NSTimeFrozenWorkGeo}{}\renewcommand{\NSTimeFrozenWorkGeo}{0.8466}
\providecommand{\NSTimeReGeo}{}\renewcommand{\NSTimeReGeo}{1.0102}
\providecommand{\NSTimeReCI}{}\renewcommand{\NSTimeReCI}{[0.9954, 1.0245]}

\providecommand{\NSTimeReWorkGeo}{}\renewcommand{\NSTimeReWorkGeo}{0.9229}
}{}
\providecommand{\RConfigs}{}\renewcommand{\RConfigs}{40}
\providecommand{\RCheapTMin}{}\renewcommand{\RCheapTMin}{4.5}
\providecommand{\RCheapTMax}{}\renewcommand{\RCheapTMax}{12.0}
\providecommand{\RCheapTMean}{}\renewcommand{\RCheapTMean}{8.4}
\providecommand{\RChangedLowS}{}\renewcommand{\RChangedLowS}{0.2}
\providecommand{\RCheapTLowS}{}\renewcommand{\RCheapTLowS}{8.4}

\providecommand{\RReuseTLowS}{}\renewcommand{\RReuseTLowS}{10.1}

\providecommand{\RChangedSmallS}{}\renewcommand{\RChangedSmallS}{2.0}

\providecommand{\RChangedMidS}{}\renewcommand{\RChangedMidS}{18.0}

\providecommand{\RChangedOneS}{}\renewcommand{\RChangedOneS}{83.1}

\providecommand{\RChangedHighS}{}\renewcommand{\RChangedHighS}{94.8}
\providecommand{\RCheapTHighS}{}\renewcommand{\RCheapTHighS}{8.4}

\providecommand{\RReuseTHighS}{}\renewcommand{\RReuseTHighS}{4.0}

\providecommand{\RNeighHighS}{}\renewcommand{\RNeighHighS}{1.7}
\providecommand{\RNeighTHighS}{}\renewcommand{\RNeighTHighS}{8.6}

\providecommand{\RHTwoWinsLowS}{}\renewcommand{\RHTwoWinsLowS}{8/8}
\providecommand{\RHTwoPLowS}{}\renewcommand{\RHTwoPLowS}{0.0039}
\providecommand{\RHTwoWinsSmallS}{}\renewcommand{\RHTwoWinsSmallS}{1/8}

\providecommand{\RHThreeGapMax}{}\renewcommand{\RHThreeGapMax}{0.7}
\providecommand{\RHThreeTimeFactorMin}{}\renewcommand{\RHThreeTimeFactorMin}{81}
\providecommand{\RHFourRatioMin}{}\renewcommand{\RHFourRatioMin}{0.07}
\providecommand{\RHFourRatioMax}{}\renewcommand{\RHFourRatioMax}{0.21}

\providecommand{\RHFiveRatioMin}{}\renewcommand{\RHFiveRatioMin}{1.50}
\providecommand{\RHFiveRatioMax}{}\renewcommand{\RHFiveRatioMax}{9.75}
\providecommand{\RHFiveRatioLowSMin}{}\renewcommand{\RHFiveRatioLowSMin}{1.50}
\providecommand{\RHFiveRatioLowSMax}{}\renewcommand{\RHFiveRatioLowSMax}{2.55}
\providecommand{\RGammaStarMin}{}\renewcommand{\RGammaStarMin}{13}
\providecommand{\RGammaStarMax}{}\renewcommand{\RGammaStarMax}{422}
\providecommand{\RGammaStarLowSMin}{}\renewcommand{\RGammaStarLowSMin}{13}
\providecommand{\RGammaStarLowSMax}{}\renewcommand{\RGammaStarLowSMax}{69}
\providecommand{\RGammaStarCheapTMin}{}\renewcommand{\RGammaStarCheapTMin}{546}
\providecommand{\RGammaStarCheapTMax}{}\renewcommand{\RGammaStarCheapTMax}{3109}
\providecommand{\RGammaMarkerMax}{}\renewcommand{\RGammaMarkerMax}{6.0}
\providecommand{\RGammaHashMin}{}\renewcommand{\RGammaHashMin}{4.3}
\providecommand{\RGammaHashMax}{}\renewcommand{\RGammaHashMax}{8.5}
\providecommand{\RIncListLowS}{}\renewcommand{\RIncListLowS}{0.09}

\providecommand{\RIncListSmallS}{}\renewcommand{\RIncListSmallS}{0.71}

\providecommand{\RIncListMidS}{}\renewcommand{\RIncListMidS}{4.28}

\providecommand{\RIncListOneS}{}\renewcommand{\RIncListOneS}{38.72}

\providecommand{\RIncListHighS}{}\renewcommand{\RIncListHighS}{382.93}

\providecommand{\RPTwo}{}\renewcommand{\RPTwo}{40/0/0}
\providecommand{\RStabViol}{}\renewcommand{\RStabViol}{0}
\providecommand{\RStabChecks}{}\renewcommand{\RStabChecks}{1,600}
\providecommand{\RStabLocSlack}{}\renewcommand{\RStabLocSlack}{44}
\providecommand{\RStabUniSlack}{}\renewcommand{\RStabUniSlack}{872}
\providecommand{\RFrameZeroRedMin}{}\renewcommand{\RFrameZeroRedMin}{5.0}
\providecommand{\RFrameZeroRedMax}{}\renewcommand{\RFrameZeroRedMax}{13.3}

\providecommand{\RPHGainLowS}{}\renewcommand{\RPHGainLowS}{1.7}

\providecommand{\RPHGainSmallS}{}\renewcommand{\RPHGainSmallS}{0.5}

\providecommand{\RPHGainMidS}{}\renewcommand{\RPHGainMidS}{0.0}

\providecommand{\CubicTILSn}{}\renewcommand{\CubicTILSn}{0.28}



\title{Vertex Ordering for Exact Triangle Listing:\\ Concave Geometry and Certified Rounding}
\author[1]{Kourosh Mirsohi}
\author[2]{Melika Gorgi}
\author[3]{Erfan Hajihashemi}
\affil[1]{Department of Computer Science, University of California, Irvine}
\affil[2]{Center for Complex Biological Systems, University of California, Irvine}
\affil[3]{Department of Electrical Engineering and Computer Science, University of California, Irvine}
\date{}

\begin{document}
\maketitle

\begin{abstract}
Exact forward triangle listers orient each edge of a graph by a vertex order. The order does not change the output, but it changes the work. For the adaptive lister, which scans the smaller of the two forward neighborhoods, the number of candidate tests is a known sum of minima of forward degrees, but how to optimize it had not been characterized. We show that this work is a concave function of the forward degrees, so choosing an order is concave minimization over the orientation polytope, whose vertices are the forward-degree vectors of orders. A single sort therefore rounds any point of the polytope, including the out-degree vector of any orientation, to an order with no more work, with a computable lower bound on the improvement. Minimizing the work is nevertheless NP-hard and APX-hard on triangle-free graphs of maximum degree three, and under the Unique Games Conjecture it has no polynomial-time constant-factor approximation. We also analyze logistic score relaxations and the large-graph limit of degree-based orders; on random cubic graphs, the expected work of every degree-based order is asymptotically at least $2.43$ times the expected optimum. Tangent refinement followed by iterated local search reduces the work by \HeldTIKMean\% on average on nine held-out networks fixed by a pre-registered protocol, and by \NSFrozenTILSMean\% in a supplementary rescoring, specified in advance, against a stronger reference that adds the published Neigh local search. The gains are in candidate tests; we observe no reproducible wall-clock speedup.
\end{abstract}

\section{Introduction}
Listing all triangles of a graph is a basic primitive in graph mining \citep{chiba1985arboricity,schank2005triangles,latapy2008main,benson2016higher} and in graph learning, where subgraph counts, which message-passing networks cannot in general capture \citep{chen2020substructures}, enrich node and graph features \citep{bouritsas2023gsn}. In network biology, triangles define clustering coefficients \citep{watts1998smallworld}, form the skeletons of closed three-node motifs such as the feed-forward loop \citep{milo2002motifs,shenorr2002motifs,alon2007motifs}, and enter graphlet analyses of brain and protein networks \citep{sporns2004motifs,przulj2007graphlets}. In molecular simulations whose three-body terms require all three pairwise distances within a cutoff \citep{pinzon2025threebody}, the admissible triplets are exactly the triangles of the cutoff graph; our guarantees apply to this enumeration, but we do not evaluate simulation runtime.

Motif analysis repeats triangle listing on many randomized copies of a network \citep{milo2002motifs}, so the cost of listing is incurred many times. Even a single listing of a large graph can be costly, motivating parallel, distributed, and external-memory listers \citep{suri2011counting,chu2011triangle,hu2013massive,shun2015multicore}. Fast exact listers orient every edge from its earlier to its later endpoint under a vertex order and intersect the resulting forward neighborhoods \citep{chiba1985arboricity,ortmann2014triangle}. The output does not depend on the order, but the work does. On Wiki-Talk (2.4M vertices), the adaptive lister studied here (defined below) performs 4.6 times more candidate tests under the reverse of the degeneracy order than under ascending degree (Appendix~\ref{app:prereg}).

\paragraph{Problem.} Let $G=(V,E)$ be a simple undirected graph with $n$ vertices and $m$ edges, and let $\pi$ be a strict total order of $V$. The \emph{forward degree} $\dout_\pi(v)$ is the number of neighbors of $v$ that come after $v$ in $\pi$, and these neighbors form its \emph{forward neighborhood} $N^+_\pi(v)$. The adaptive forward lister \citep{yu2020aot} handles each edge by scanning the smaller of the forward neighborhoods of its two endpoints. Each scanned vertex is tested for membership in the other neighborhood, and we call each such test a \emph{candidate test}. The total number of candidate tests is
\begin{equation}
 \Lam(\pi)=\sum_{\{u,v\}\in E}\min\{\dout_\pi(u),\dout_\pi(v)\},
 \label{eq:work}
\end{equation}
which we call the \emph{work} of $\pi$ (Proposition~\ref{prop:exact}). The sum itself was given by \citet{yu2020aot}; what had not been characterized is its dependence on the order. We study the \emph{ordering problem} $\Lopt{G}=\min_\pi\Lam(\pi)$ and write $\Lam_G(\pi)$ when the graph matters. For listers that always scan the forward neighborhood of one fixed endpoint, two related costs are NP-hard to minimize \citep{lecuyer2023tailored}.

\paragraph{Prior work and open questions.} In practice, vertices are ordered by degree or by degeneracy \citep{chiba1985arboricity,matula1983smallest,ortmann2014triangle}. \citet{lecuyer2023tailored} designed tailored orders, Split and Check, and a local search, Neigh, that can beat degeneracy orders, but for two non-adaptive costs that differ from $\Lam$. Differentiable relaxations of permutations \citep{mena2018gumbel,cuturi2019sorting,blondel2020sorting} offer another route: optimize a smooth surrogate and round. What has been missing is an understanding of the objective itself: its structure, how to round without increasing the work, its hardness, why gradient-based optimization can fail, and how much degree information alone can capture.

\paragraph{Our results.} We show that $\Lam$ depends on the order only through the vector of forward degrees and is a concave function of it. The fractional forward-degree vectors, which split one unit per edge between its two endpoints, form the \emph{orientation polytope} of $G$. Its vertices are exactly the forward-degree vectors of orders, so the ordering problem is concave minimization over this polytope (Section~\ref{sec:geometry}). Our main contributions are:
\begin{enumerate}\setlength{\itemsep}{0pt}
\item \textbf{Geometry and rounding.} A single sort, the \emph{tangent step}, rounds any point of the polytope, including the out-degree vector of any orientation, to an order whose work is no larger, with a computable lower bound on the improvement; iterating it gives \emph{tangent refinement}. The same geometry gives an exact mixed-integer program without acyclicity constraints, displayed in Eq.~\eqref{eq:mip} (Theorem~\ref{thm:rounding}, Corollary~\ref{cor:free}, Propositions~\ref{prop:rules} and~\ref{prop:bounds}).
\item \textbf{Complexity.} Through an exact correspondence with Max-Cut on subdivided graphs, we show that minimizing $\Lam$ is NP-hard and APX-hard even on triangle-free graphs of maximum degree three, and under the Unique Games Conjecture (UGC) it has no polynomial-time constant-factor approximation (Theorem~\ref{thm:hard}). A Min UnCut connection with the same UGC consequence was known for the non-adaptive cost $\Cpm$ \citep[\S 5.3.3]{lecuyer2023thesis}; ours is for the adaptive cost and rests on an exact identity for the optimum on 2-subdivisions.
\item \textbf{Relaxation landscape.} Pairwise-logistic (Bradley--Terry) relaxations orient each edge softly by a logistic function of unconstrained vertex-score differences at a fixed temperature. Every finite stationary point of an objective that is concave in the resulting soft forward degrees is a global maximizer. Our smoothed work is strictly concave in them on the affine hull of the polytope, so it has no finite local minimum, yet its infimum is hard to approximate, and its minimizing sequences saturate the score differences, where small gradients certify nothing about the score-sorted order (Theorem~\ref{thm:landscape}).
\item \textbf{Random graphs.} On configuration-model random graphs with bounded degrees, the expected work per vertex of every \emph{degree-based} order, which places vertices in consecutive blocks determined only by their degrees, has an explicit large-graph limit; in this limit the best layouts we found are V-shaped and do 8.7--21.4\% less work than ascending degree, depending on the degree law. Yet on large random cubic graphs, the expected work of every degree-based order is at least $2.438-o(1)$ times the expected optimum (Theorem~\ref{thm:mf}, Proposition~\ref{prop:cubic}).
\item \textbf{Algorithms and evaluation.} Iterated local search reaches the optimum on all 27 small instances on which a solver certifies the optimum of the program. Our method, tangent refinement followed by iterated local search (\TILS), reduces work on 32 public development graphs by \SuiteTIKMean\% on average relative to the best of standard and tailored reference orders in both directions, improving on every graph. On nine held-out networks whose selection was pre-registered, it reduces work by \HeldTIKMean\% on average relative to the pre-registered reference. In a supplementary rescoring that was specified in advance but is not part of the original pre-registration, the same orders reduce work by \NSFrozenTILSMean\% against a stronger reference that adds the published local search Neigh, improving on all nine; rerun from this stronger reference, \TILS needs fewer candidate tests than the gradient-based relaxation \DOVE (Section~\ref{sec:relax}) on all nine (Section~\ref{sec:experiments}). The ablations attribute most of the gain to the plateau moves and kicks of the local search, not to the tangent prefix (Sections~\ref{sec:ablation} and~\ref{sec:posthoc}).
\end{enumerate}
All gains are in candidate tests. The geometric-mean listing time changes by at most a few percent, not reproducibly across machines, and ordering costs more than it saves, even when orders are reused across related graphs (Section~\ref{sec:reuse}), so we claim no end-to-end speedup. Table~\ref{tab:novelty} in Appendix~\ref{app:extra} separates prior ingredients from new results.

\section{Related Work}
\paragraph{Ordered listing.}
Triangle finding and listing date back to \citet{itai1978minimum}, with matrix-multiplication and output-sensitive bounds \citep{alon1997cycles,bjorklund2014listing} and approximate and parallel counting methods \citep{tsourakakis2008fast,arifuzzaman2013patric}. Arboricity bounds sparse listing \citep{chiba1985arboricity}, degeneracy orders are computed efficiently \citep{matula1983smallest,batagelj2003core}, and forward and node-iterator variants are compared by \citet{schank2005triangles} and \citet{ortmann2014triangle}. \citet{yu2020aot} make the adaptive intersection explicit and express its running time through the sum in Eq.~\eqref{eq:work}, but give no structural or hardness result for choosing the order. \citet{lecuyer2023tailored} formalize two non-adaptive costs (Proposition~\ref{prop:rules}), prove their minimization NP-hard by reductions from Set Cover and NAE-3SAT, and propose three tailored orders: Split places vertices by decreasing degree alternately at the two ends, Check places them greedily at the front or back, and Neigh is a local search that moves each vertex to its best position among its neighbors. \citet[\S 5.3.3]{lecuyer2023thesis} further discusses relaxations, bounds and approximations for these costs, including the minimization of $\Cpp$ over arbitrary orientations and, attributed there to Makarychev, a connection between $\Cpm$ and Min UnCut that gives an $O(\sqrt{\log n})$ approximation and, under the UGC, rules out a constant-factor approximation for $\Cpm$. To our knowledge, the adaptive cost $\Lam$ is not studied in these works. Data structures such as 2-3 cuckoo filters \citep{eppstein2017cuckoo} reduce the cost per candidate test, whereas the order changes the number of tests. For $k$-cliques, kClist and color-based orders are standard \citep{danisch2018kcliques,li2020kclique}, and memory-locality orders optimize a different objective \citep{wei2016gorder}.
\paragraph{Orientations and concave minimization.}
The set of fractional out-degree vectors is the graphical zonotope, whose vertices correspond to acyclic orientations \citep{stanley1980decompositions,greene1983interpretation,postnikov2009permutohedra} and whose integral points are the out-degree vectors of orientations \citep{hakimi1965degrees}. It is the base polytope of a submodular function, so linear functions over it are minimized by sorting \citep{edmonds1970submodular}, and with this linear minimization oracle, successive linearization, a classical method for concave minimization \citep{pardalos1986concave,horst1996global,mangasarian1997concave}, coincides with a Frank--Wolfe step of full step size \citep{frank1956algorithm,jaggi2013revisiting}. Convex separable costs of out-degrees are minimized over orientations by path reversal \citep{harvey2006semimatchings,borradaile2017egalitarian}, and minimum-entropy orientation is a related concave degree objective \citep{cardinal2008entropy}. What is new is that $\Lam$ is concave over this polytope, which makes the tangent step an exact rounding to an order with certified non-increasing work and removes the need for acyclicity constraints in our exact program.
\paragraph{Continuous relaxations and evaluation.}
Differentiable sorting, ranking and permutation operators relax global permutation structure \citep{mena2018gumbel,grover2019neuralsort,cuturi2019sorting,blondel2020sorting,prillo2020softsort,petersen2021sortnet,struski2025lapsum}, and gradient-based methods have been applied to algorithmic decisions such as join ordering \citep{schwabe2025join}. Closest in spirit, \citet{nerem2024birkhoff} construct differentiable extensions with rounding guarantees over permutations; our rounding instead follows from the concavity of $\Lam$ and applies to any point of the orientation polytope. Pairwise logistic comparisons are classical \citep{bradley1952rank,hunter2004mm,burges2005ranknet}. \citet{deng2023notears} optimize topological orders by discrete swaps, and learned orderings for other exact algorithms \citep{zhao2020graphordering,wang2022rlqvo,langedal2025coloring,min2025clique} optimize different objectives. Because weak baselines inflate reported gains \citep{armstrong2009improvements,melis2018state,lipton2019troubling} and evaluation references can themselves be unreliable \citep{gorgi2026auditing}, we evaluate published orders in both directions, compare against the strongest published local search available to us, and pre-register the held-out analysis \citep{nosek2018preregistration}.

\section{The Work of Ordered Listing}\label{sec:geometry}
We write $d(v)$ for the degree of $v$ and $\mathcal T(G)$ for the set of triangles of $G$; an order orients each edge from its earlier to its later endpoint.

\begin{proposition}[Exact output and work]\label{prop:exact}
For every $\pi$, adaptive forward listing outputs each triangle exactly once and performs exactly $\Lam(\pi)$ candidate tests; hence $\Lam(\pi)\ge|\mathcal T(G)|$.
\end{proposition}

\begin{figure}[t]
\centering
\includegraphics[width=0.55\textwidth,trim=0 0 286 16,clip]{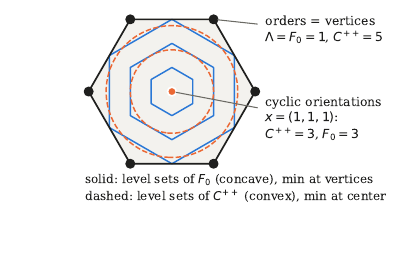}
\caption{Acyclicity is free for concave costs. The orientation polytope $Z_{K_3}$ lies in the plane $x_1+x_2+x_3=3$; its vertices are the forward-degree vectors of the six orders, and its center $(1,1,1)$ is the out-degree vector of both cyclic orientations. Concave costs such as $\Lambda=F_0$ (solid level sets) attain their minimum at vertices (Corollary~\ref{cor:free}); the convex $C^{++}$ (dashed level sets) attains it at the center (Proposition~\ref{prop:rules}).}
\label{fig:geometry}
\end{figure}

\subsection{Degree-space geometry}
Write $e_v$ for the unit vectors of $\mathbb R^V$. The \emph{orientation polytope} is the Minkowski sum
\[
\ZG=\sum_{\{u,v\}\in E}[e_u,e_v].
\]
Each point splits one unit per edge between its endpoints, and the integral points of $\ZG$ are exactly the out-degree vectors of orientations of $G$, acyclic or not \citep{hakimi1965degrees}. The vertices of $\ZG$ are exactly the out-degree vectors of acyclic orientations, i.e., of orders \citep{stanley1980decompositions,greene1983interpretation}, as Figure~\ref{fig:geometry} illustrates for $K_3$. Define
\[
F_0(x)=\sum_{\{u,v\}\in E}\min(x_u,x_v),
\]
a concave function on $\mathbb R^V$ with $\Lam(\pi)=F_0(\dout_\pi)$.

\begin{theorem}[Certified tangent rounding]\label{thm:rounding}
Let $F$ be concave on an open convex set containing $\ZG$, let $x\in\ZG$, and let $a$ be a supergradient of $F$ at $x$. Let $\pi_a$ sort $V$ by ascending $a_v$, breaking ties by vertex label, and set
\[
 \delta(x,a)=a^\top x-\textstyle\sum_{\{u,v\}\in E}\min(a_u,a_v).
\]
Then $\delta(x,a)\ge0$ and $F(\dout_{\pi_a})\le F(x)-\delta(x,a)$. For $F_0$, $\Cpm$ (Proposition~\ref{prop:rules}) and $F_\epsilon$ (Section~\ref{sec:relax}), computing $a$ and $\delta$ takes $O(m+n)$ time and computing $\pi_a$ takes $O(n\log n)$ time.
\end{theorem}
Each edge contributes to $a^\top x$ at least $\min(a_u,a_v)$, which $\pi_a$ attains, and the supergradient inequality at $\dout_{\pi_a}$ gives the bound (Appendix~\ref{app:proofs}). The step is one Frank--Wolfe iteration with full step size, with the linear minimization over $\ZG$ solved by sorting \citep{edmonds1970submodular}, and $\delta(x,a)$ is its Frank--Wolfe gap \citep{jaggi2013revisiting}. Because $F$ is concave, the gap lower-bounds the improvement rather than the suboptimality.

\begin{corollary}[Acyclicity is free for concave costs]\label{cor:free}
If $F$ is concave on an open convex set containing $\ZG$, then
\[
\min_\pi F(\dout_\pi)=\min_{x\in\ZG}F(x)=\min_{\text{orientations }D}F(\dout_D),
\]
and any orientation or point of $\ZG$ is converted to an order of no larger cost, in $O(m+n\log n)$ time for $F_0$.
\end{corollary}

For $F_0$ at an integral $x$, a supergradient is $a_v=\sum_{e\ni v}\omega_{e,v}$, where each edge gives weight one to its endpoint with the smaller value of $x$, zero to the other, and one half to each at ties. Iterating $x\leftarrow\dout_{\pi_a(x)}$ while $\Lam$ strictly decreases defines \emph{tangent refinement}~(T), a successive-linearization method \citep{mangasarian1997concave} with an exact integer objective. It terminates because $\Lam$ is a nonnegative integer, but it need not reach a global optimum (Appendix~\ref{app:proofs} gives a $K_{3,3}$ example).

\subsection{Three scan rules, three costs}\label{sec:costs}
The forward kernel handles an oriented edge $u\to v$ by intersecting $\Nout(u)$ with $\Nout(v)$ (we drop the subscript $\pi$); its three natural scan rules give the three cost models in the literature.

\begin{proposition}[Cost models and curvature]\label{prop:rules}
Scanning the smaller set costs $\Lam$; always scanning $\Nout(v)$ costs $\Cpm=\sum_v\dout(v)(d(v)-\dout(v))$; always scanning $\Nout(u)$ costs $\Cpp=\sum_v\dout(v)^2$. As functions of $x=\dout$, $\Lam$ and $\Cpm$ are concave and $\Cpp$ is strictly convex. Consequently Corollary~\ref{cor:free} applies to $\Lam$ and $\Cpm$, while for $\Cpp$ the minimum over all orientations can be strictly smaller than over orders (on $K_3$, 3 versus 5). That orientation minimum is computable in polynomial time by path reversal \citep{harvey2006semimatchings,borradaile2017egalitarian} and is a certified lower bound for orders.
\end{proposition}

\subsection{Lower bounds, complexity and an exact program}\label{sec:complexity}
\begin{proposition}[Certificates]\label{prop:bounds}
(a) $\Lam(\pi)\ge\Lopt{G}\ge\max\{|\mathcal T(G)|,\beta(G)\}$ for every $\pi$, where $\beta(G)$ is the minimum number of edges whose deletion makes $G$ bipartite; $\Lopt{G}=0$ if and only if $G$ is bipartite.

(b) For integral $x\in\ZG$, $F_0(x)=\sum_{k\ge1}|E(U_k)|$ with $U_k=\{v:x_v\ge k\}$. Hence $\Lopt{G}$ equals the optimum of the mixed-integer program~\eqref{eq:mip}, which has no acyclicity constraints.
\end{proposition}

To state the program, name the endpoints of each edge as $e=\{u_e,v_e\}$ in a fixed but arbitrary way and let $\mu_e=\min\{d(u_e),d(v_e)\}$. A continuous split $p_e$ assigns the share $p_e$ of edge $e$ to $u_e$ and $1-p_e$ to $v_e$, binary variables $b_{v,k}$ encode $x_v$ in unary, and $w_{e,k}$ indicates that both endpoints of $e$ lie in the level set $U_k$:
\begin{equation}\label{eq:mip}
\begin{aligned}
\min\quad & \sum_{e\in E}\sum_{k=1}^{\mu_e}w_{e,k}\\
\text{s.t.}\quad & \sum_{e:\,u_e=v}p_e+\sum_{e:\,v_e=v}(1-p_e)=\sum_{k=1}^{d(v)}b_{v,k} && v\in V,\\
& b_{v,k}\ge b_{v,k+1} && v\in V,\ 1\le k<d(v),\\
& w_{e,k}\ge b_{u_e,k}+b_{v_e,k}-1 && e\in E,\ 1\le k\le\mu_e,\\
& p_e\in[0,1],\quad b_{v,k}\in\{0,1\},\quad w_{e,k}\ge0.
\end{aligned}
\end{equation}
For a graph without isolated vertices it has $m$ continuous splits, $2m$ binary variables, $\sum_e\mu_e$ continuous indicators, $n$ equations, $2m-n$ monotonicity constraints and $\sum_e\mu_e$ linking constraints. Because the objective coefficients are positive, every optimal solution sets $w_{e,k}=\max\{0,b_{u_e,k}+b_{v_e,k}-1\}$, and the monotone binaries force $b_{v,k}=\mathbf 1[x_v\ge k]$, which gives the layer identity; the integral point $x$ need not come from an acyclic orientation, and Corollary~\ref{cor:free} shows that this costs nothing. From any incumbent we read off $x_v=\sum_kb_{v,k}$, round it to an order by one tangent step, and recompute the work of that order exactly, so every reported upper bound is the exact integer work of an order. Appendix~\ref{app:proofs} gives the proof and valid triangle and odd-cycle cuts.

To show hardness, we use the \emph{2-subdivision} $S_2(H)$ of a graph $H$, which replaces every edge $uv$ by a path $u\,a\,b\,v$ through two new vertices. It has no triangles, so all of its work is excess work beyond the output. For any graph $H$, $\beta(H)=|E(H)|-\mathrm{MaxCut}(H)$, the minimum number of edges that a cut of $H$ leaves uncut (Min UnCut). Recall that $\Lam_G(\pi)$ is the work of $\pi$ on $G$ and $\Lopt{G}$ its minimum.
\begin{theorem}[Complexity]\label{thm:hard}
For every graph $H$, $\Lopt{S_2(H)}=\beta(H)$, and linear-time maps convert orders of $S_2(H)$ into cuts of $H$ and back without increasing the objective value. Consequently, minimizing $\Lam$ is NP-hard and APX-hard even on triangle-free graphs of maximum degree three, and under the UGC no polynomial-time algorithm approximates $\Lopt{G}$ within any constant factor, even on triangle-free graphs.
\end{theorem}

\begin{proof}[Proof sketch]
The lower bound follows from Proposition~\ref{prop:bounds}a together with $\beta(S_2(H))=\beta(H)$. For the upper bound, a cut of $H$ colors $S_2(H)$ so that one color class $Y$ is independent and each uncut edge of $H$ leaves exactly one monochromatic edge; placing $Y$ last, with original vertices first, makes exactly these edges cost one. The cut-to-order map preserves the objective exactly, and the order-to-cut map does not increase it. Consequently the optima coincide, and approximation guarantees transfer in the required direction, so the hardness of Max-Cut on cubic graphs \citep{alimonti2000cubic} and of Min UnCut under the UGC \citep{khot2002unique,khot2007optimal} carries over (Appendix~\ref{app:complexity}).
\end{proof}

The NP-hardness reductions of \citet{lecuyer2023tailored} for the non-adaptive costs, from Set Cover and NAE-3SAT, do not give inapproximability. For $\Cpm$, the Min UnCut connection discussed by \citet[\S 5.3.3]{lecuyer2023thesis} does, with the same UGC consequence as Theorem~\ref{thm:hard}. Our contribution is the corresponding result for the adaptive cost $\Lam$, through a different construction: the exact identity $\Lopt{S_2(H)}=\beta(H)$, which also gives APX-hardness on subcubic triangle-free graphs and shows that the bound $\beta$ of Proposition~\ref{prop:bounds}a is attained on 2-subdivisions. Because $\beta(H)=|E(H)|-\mathrm{MaxCut}(H)$ is additive, Max-Cut approximation ratios do not transfer.

\section{Optimizing Orders}\label{sec:opt}
\subsection{Smooth relaxation and its landscape}\label{sec:relax}

Scores $s\in\mathbb R^V$ and a temperature $\tau>0$ define soft edge directions
$p_{uv}=\sigma((s_v-s_u)/\tau)$, where $\sigma$ is the logistic function.
As $\tau\to0$, $p_{uv}$ tends to one when $u$ precedes $v$ in the order
$\pi_s$ that sorts $V$ by ascending score. The soft forward degrees are
$x_u(s)=\sum_{v\in N(u)}p_{uv}$, where $N(u)$ is the set of neighbors of $u$.
To smooth the minimum, we use $\smin_\epsilon(a,b)=\tfrac12\big(a+b-\sqrt{(a-b)^2+\epsilon^2}+\epsilon\big)$, which satisfies $0\le\smin_\epsilon(a,b)-\min(a,b)\le\epsilon/2$. With $F_\epsilon(x)=\sum_{\{u,v\}\in E}\smin_\epsilon(x_u,x_v)$, the relaxed work is $J(s)=F_\epsilon(x(s))$. Both $J$ and its exact gradient
can be evaluated in $O(m+n)$ time, and $J(s)\to\Lam(\pi_s)$ as
$\tau,\epsilon\to0$ for scores that are distinct on every edge
(Appendix~\ref{app:proofs}).

Our gradient-based baseline, \DOVE (differentiable ordering via edges),
minimizes $J$ with Adam \citep{kingma2015adam}, anneals the temperature,
and standardizes the scores after each step. It starts from a given order
and returns the best exact work of $\pi_s$ seen during the run.

\begin{theorem}[Score landscape]\label{thm:landscape}
Let $G$ have at least one edge, fix $\tau>0$, and let $L(s)$ be the weighted Laplacian with edge weights $p_{uv}(1-p_{uv})$.
(i) The map $s\mapsto x(s)$ is onto the relative interior $\relint\ZG$ and is injective modulo adding a constant on each connected component.
(ii) If $F$ is concave and differentiable, every finite stationary point of $J(s)=F(x(s))$ maps to a global \emph{maximizer} of $F$ on $\ZG$.
(iii) With $a=\nabla F(x(s))$ and $\lambda_*(s)$ the smallest nonzero eigenvalue of $L(s)$,
\[
 \delta(x,a)\le\frac{\tau\,\|x(s)-\dout_{\pi_a}\|_2}{\lambda_*(s)}\,\|\nabla J(s)\|_2 .
\]
\end{theorem}
Part (i) is the exponential-family mean map for edge orientations \citep{wainwright2008graphical,aldous2018scores}. Part (ii) follows from $\nabla J=-L\,a/\tau$: stationarity makes $a$ constant on components, and concavity then gives a global maximum. It holds for any differentiable objective that is concave in the soft degrees, for example $F_\epsilon$, $\Cpm$ and minimum-entropy orientation \citep{cardinal2008entropy} (Appendix~\ref{app:entropy}).

As $F_\epsilon$ is strictly concave on the affine hull of $\ZG$, $J$ has no finite local minimum. Minimizing sequences drive score gaps to infinity, where $\lambda_*\to0$ and the bound (iii) on the gap of the tangent order $\pi_a$ becomes vacuous. A small gradient also says nothing about the score-sorted order $\pi_s$: take scores that are distinct on every edge and sort to an arbitrary, possibly poor order, and scale them by $c\to\infty$. The weights $p_{uv}(1-p_{uv})$ tend to zero, so $\nabla J\to0$, while $\pi_{cs}$ stays the same order and $J(cs)\to F_\epsilon(\dout_{\pi_s})$. The infimum is nevertheless hard to reach: one tangent step turns any scores $s$ into an order with $\Lam\le J(s)$, so for $\epsilon\le1/m$ finding scores within $1/2$ of $\inf J$ is NP-hard, and under the UGC no constant-factor approximation of $\inf J$ is possible on non-bipartite graphs (Corollary~\ref{cor:relaxhard}, Appendix~\ref{app:complexity}). These statements assume that soft degrees are evaluated and compared exactly (a real-RAM model of the relaxation); a finite-precision version would need a representation of scores and a robust rounding rule. Absence of spurious local minima thus does not make the problem tractable; the difficulty lies at infinity.

The theorem fixes $\tau$ and leaves scores unnormalized, whereas \DOVE anneals $\tau$ and standardizes scores, so its weaker empirical performance is a separate finding (Appendices~\ref{app:fixedtau}, \ref{app:mirror} and~\ref{app:budget} give a fixed-temperature diagnostic, a monotone mirror step and budget variations).

\subsection{Tangent refinement with iterated local search}\label{sec:tils}
Our practical method, \TILS (Algorithm~\ref{alg:tils} in Appendix~\ref{app:protocol}), runs tangent refinement, then iterated local search (ILS) \citep{lourenco2003ils} over insertion moves, and finally one additional tangent pass. A move picks a random edge $\{u,v\}$ and reinserts $u$ immediately before or after $v$; re-evaluating only the edges at vertices whose forward degree changes gives the exact change $\Delta$ in work. Plateau moves ($\Delta=0$) are accepted with probability $\rho$, and after $S$ chunks of 256 proposals without improvement the search restarts from the incumbent with $\kappa$ random moves (a \emph{kick}); $(\rho,\kappa,S)=(1,10,50)$ was selected on six synthetic development graphs. The initial tangent phase uses at most 10\% of the time budget, and the final pass is reserved within the budget (Appendix~\ref{app:protocol}). The exact incumbent ensures that \TILS is never worse than its starting order. The geometry supplies the guaranteed rounding and refinement steps, while the ablations show that the plateau moves and kicks account for most of the search improvement (Section~\ref{sec:ablation}). The initial tangent phase turns out to be neutral (Section~\ref{sec:posthoc}); we keep it because it is part of the pre-registered configuration, and the final pass is cheap and never increases work.

\section{Random-Graph Limits}\label{sec:random}
How much can degree information alone achieve, and does optimizing beyond degree help on large graphs?

Let $\mathrm{CM}_n(p)$ be the configuration model \citep{bollobas1980configuration,molloy1995critical,newman2001random,vanderhofstad2016random} with i.i.d.\ degrees $D_1,\dots,D_n$ of law $p$ on $\{1,\dots,K\}$ and mean $\bar d$. Let $q_k=kp_k/\bar d$ be the size-biased degree law seen along a random edge. A \emph{degree-based order} places degree classes (the sets of vertices of equal degree), or fixed fractions of them, as consecutive blocks with uniformly random order inside each block. For example, Split, which places vertices alternately at the two ends, splits each degree class between the two ends. The layout (the sequence of blocks and their fractions) is fixed as $n$ grows, and the assignment of vertices to blocks depends only on degrees and vertex labels, not on the pairing (Appendix~\ref{app:mf} handles loops, multi-edges and odd degree sums). In the limit, the normalized position (rank divided by $n$) of a degree-$k$ vertex is a random variable $\xi_k\in[0,1]$ determined by the layout, and a neighbor of a vertex at position $t$ is later with probability
\[
h(t)=\sum_k q_k\Pr(\xi_k>t).
\]

\begin{theorem}[Mean-field work]\label{thm:mf}
Let $(k,t)$ and $(k',t')$ be i.i.d., with $k\sim q$ and $t\sim\xi_k$ conditional on $k$. Let
$Q_{k,t}(j)=\Pr[\mathrm{Bin}(k-1,h(t))\ge j]$. For every degree-based order sequence $\pi_n$ as above,
\[
\frac{\mathbb E\Lam(\pi_n)}{n}\;\longrightarrow\;
\bar d\sum_{j\ge1}\mathbb E\!\left[
\mathbf 1\{t<t'\}\,
Q_{k,t}(j-1)\,
Q_{k',t'}(j)
\right].
\]
\end{theorem}

The formula describes a uniformly random edge whose earlier endpoint has forward degree
$1+\mathrm{Bin}(k-1,h(t))$, where the $1$ counts the other endpoint, and whose later endpoint has forward degree $\mathrm{Bin}(k'-1,h(t'))$. Local tree-likeness makes the two asymptotically independent (proof in Appendix~\ref{app:mf}). The limit depends only on the block layout and takes milliseconds to evaluate, so we can search over layouts numerically.

\begin{table}[t]
\centering\small\setlength{\tabcolsep}{3pt}
\caption{Mean-field limit of $\Lambda/n$ for degree-based orders. PL$_\gamma$: $p_k\propto k^{-\gamma}$ on $[2,100]$; Poisson$(c)$: conditioned on $k\ge1$, truncated at $c+12\sqrt{c}+20$. Best: lowest value among Split and the layouts found by local search (not certified optimal); last column: its reduction relative to ascending degree.}
\label{tab:mf}
\begin{tabular}{lrrrrr}
\toprule
Law & Ascending & Descending & Split & Best & vs.\ asc.\\
\midrule
Poisson(4) & 2.415 & 2.615 & 2.071 & 2.071 & 14.3\%\\
Poisson(8) & 9.899 & 10.400 & 8.807 & 8.807 & 11.0\%\\
Poisson(16) & 40.980 & 41.956 & 37.395 & 37.395 & 8.7\%\\
PL2.1 & 9.380 & 14.682 & 8.626 & 8.080 & 13.9\%\\
PL2.5 & 3.679 & 4.928 & 3.039 & 2.975 & 19.1\%\\
PL3.0 & 1.498 & 1.684 & 1.177 & 1.177 & 21.4\%\\
\bottomrule
\end{tabular}
\end{table}

\begin{figure}[t]
\centering
\includegraphics[width=0.85\textwidth,trim=200 0 0 16.5,clip]{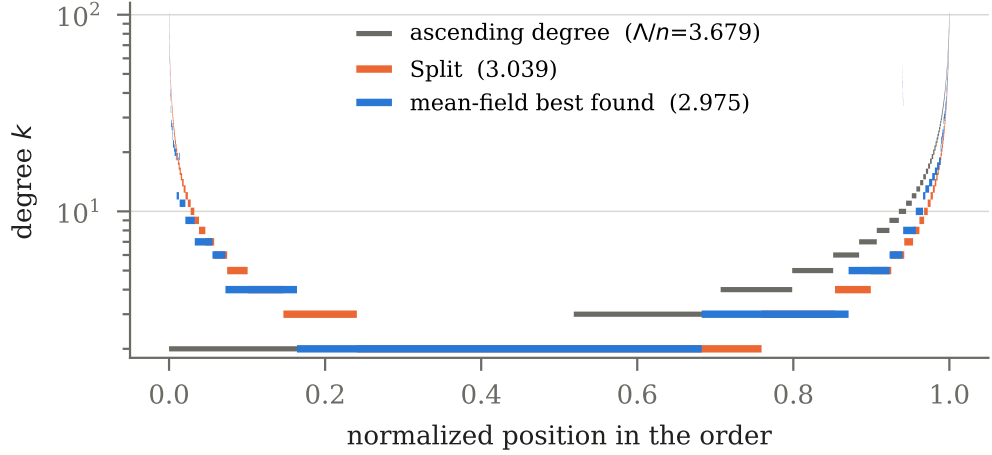}
\caption{Degree-based layouts in the configuration model with $p_k\propto k^{-2.5}$ on $[2,100]$ (PL2.5): degree of the vertex at each normalized position in the order, with mean-field limits of $\Lambda/n$ (Theorem~\ref{thm:mf}) in the legend. The best layout found is V-shaped, like Split.}
\label{fig:layouts}
\end{figure}

\paragraph{Consequences.} Table~\ref{tab:mf} evaluates the limit. Sampled graphs with $n=10^5$, averaged over graph seeds, agree with it to within 1\%. (1) Ascending degree, the common default, is not optimal: the best layouts found have 8.7--21.4\% less work. (2) The best layouts found are \emph{V-shaped} (Figure~\ref{fig:layouts}), with large degrees at both ends and small degrees in the middle. Split \citep{lecuyer2023tailored} is such a V. It matches the best layout found for the Poisson and PL3.0 laws but is 2--7\% worse for heavier tails. (3) Optimizing beyond degree retains a large gain at the sizes we tested. From the best-found degree-based order at $n=10^5$, T removes another 19.8--28.8\% of the work, \DOVE 22.4--38.0\%, and \TILS, with about five times \DOVE's time, 27.8--57.9\%. The \TILS values change by $-0.2$ to $+0.9$ percentage points from $n=10^4$ (Table~\ref{tab:scaling-full}, Appendix~\ref{app:extra}).

\paragraph{Degree is not enough.} On cubic graphs, where every degree is three, a degree-based order behaves like a uniformly random order, and Theorem~\ref{thm:mf} gives the limit $59/60$ per vertex. The next result makes this exact for finite cubic graphs and bounds the optimum, where $\alpha(G)$ denotes the independence number.

\begin{proposition}[Degree is not enough on cubic graphs]\label{prop:cubic}
Let $G$ be a cubic graph with $n$ vertices and $t$ triangles, and let $\pi$ be a uniformly random order. Then $|\mathbb E\Lam(\pi)-\tfrac{59}{60}n|\le9t$ and $\Lopt{G}\le3n-6\alpha(G)$. Hence, if $G$ is triangle-free, $\mathbb E\Lam(\pi)=\tfrac{59}{60}n$ while $\Lopt{G}\le\tfrac67n$. On the uniformly random cubic graph, every order chosen independently of the edges, in particular every degree-based order, has $\mathbb E\Lam/n\to\tfrac{59}{60}$, whereas $\Lopt{G}\le(15-36\ln\tfrac32+o(1))\,n\approx0.4033\,n$ asymptotically almost surely. Hence $\mathbb E\Lam(\pi)/\mathbb E\Lopt{G}\ge2.438-o(1)$.
\end{proposition}

\emph{Proof idea.} For an edge in no triangle, the forward degrees of its endpoints depend only on the relative order of six distinct vertices, and averaging over the 720 patterns gives $59/90$ per edge. For the bound, place a maximum independent set last and orient the paths and cycles induced by the other vertices alternately, so that each remaining edge costs at most two; the bounds on $\alpha$ are due to \citet{staton1979independence} and \citet{frieze1994independence} (Appendix~\ref{app:complexity}). In the limit, the optimum has over 58\% less work than any degree-based order in expectation, a constant-factor gap. At $n=10^5$, \TILS reaches $\Lam\approx\CubicTILSn\,n$ (Table~\ref{tab:cubic}, Appendix~\ref{app:complexity}).

\section{Experiments}\label{sec:experiments}
\paragraph{Protocol.}
All work values are recomputed from exported orders, and triangle counts are checked to be order-invariant; these checks passed for all 192 development and 54 held-out saved orders. The \emph{reference} $R^\star$ of a graph is the lowest-$\Lam$ order among up to eight candidates, namely Natural, a seeded Random order, ascending Degree, core ordering, smallest-last degeneracy \citep{matula1983smallest}, Split, Check and Neigh. On the four public development graphs used for ablations (Bitcoin-Alpha, P2P-Gnutella05, Loc-Brightkite and GitHub), Neigh is run with five seeds; on the 32 development graphs, we use the stored best of five seeds. Every candidate is evaluated \emph{in both directions}, because Split, Check and Neigh target the reversal-invariant $\Cpm$, while their $\Lam$ values can differ under reversal. Reversal lowers the best forward-direction reference by up to 1.85\%. Neigh minimizes $\Cpm$, not $\Lam$, so the comparisons below show that optimizing $\Lam$ directly improves on the evaluated orders; they do not compare diversified search with an equally strong deterministic local search for $\Lam$ itself.

In the pre-registered held-out run, the Neigh binary was unavailable on the evaluation machine, so the held-out $R^\star$ excludes Neigh. A supplementary comparison declared in advance adds Neigh, giving the stronger reference $R^\star_N$ (Section~\ref{sec:posthoc}). Every optimization method starts from the applicable reference, keeps an exact incumbent, and is charged all of its time. \DOVE runs its default 200-step schedule, and \TILS receives the same wall-clock time. The \emph{reduction} of an order relative to a reference $R$ is $100(\Lam_R-\Lam)/\Lam_R$ in \%. Means over graphs are unweighted means of these per-graph reductions, not reductions of the total work of a pooled workload, so we also report geometric-mean work ratios $\Lam/\Lam_R$. Differences between reductions are reported in percentage points (Appendix~\ref{app:protocol}). Table~\ref{tab:notation} summarizes the references and metrics.

\begin{table}[t]
\centering\small
\caption{Methods, references and metrics used in the experiments.}
\label{tab:notation}
\begin{tabular}{@{}>{\raggedright\arraybackslash}p{0.2\textwidth}>{\raggedright\arraybackslash}p{0.75\textwidth}@{}}
\toprule
Symbol & Meaning\\
\midrule
$\Lam_G(\pi)$, $\Lopt{G}$ & Work (candidate tests) of $\pi$ on $G$, Eq.~\eqref{eq:work}, and its minimum over orders. $\mathcal T(G)$ is the set of triangles.\\
T & Tangent refinement: repeated tangent steps while the work strictly decreases (Section~\ref{sec:geometry}).\\
ILS, \TILS, ILS-T & Iterated local search over insertion moves; \TILS runs T, then ILS, then a final T pass (Section~\ref{sec:tils}); ILS-T omits the initial T phase.\\
\DOVE & Gradient-based relaxation baseline (Section~\ref{sec:relax}).\\
$R^\star$ & Lowest-work order among Natural, Random, Degree, Core, smallest-last, Split, Check and, on the 32 development graphs, stored Neigh orders, each in both directions. Used for the 32 development graphs and, without Neigh, for the pre-registered held-out graphs.\\
$R^\star_N$ & Held-out $R^\star$ plus five Neigh runs in both directions; supplementary analysis (Appendix~\ref{app:neighsupp}).\\
$R^\star_{\rm cheap}$ & $R^\star$ without Neigh, recomputed on each sample of a null-model chain (Section~\ref{sec:reuse}).\\
Reduction & $100(\Lam_R-\Lam)/\Lam_R$ in \% relative to the stated reference $R$; differences in percentage points; means unweighted over graphs.\\
Work ratio & $\Lam/\Lam_R$, aggregated by the geometric mean over graphs.\\
Listing time & Time of the lister given an order (kernel time); ordering time is reported separately and is charged to the method.\\
Gap & Relative bound width of the program, $(\mathrm{UB}-\lceil\mathrm{LB}\rceil)/\mathrm{UB}$ (Section~\ref{sec:exact}).\\
\bottomrule
\end{tabular}
\end{table}

\subsection{Distance to optimality}\label{sec:exact}
We solve the program~\eqref{eq:mip} with HiGHS~1.15.1 \citep{huangfu2018highs}, without warm starts and with a time limit of 60\,s for $n\le32$ and 120\,s for $n>32$, on 49 instances of seven random-graph families with $n\in\{16,\dots,64\}$. The solver runs on one thread with a relative gap tolerance of zero and otherwise default tolerances, and the model includes all the valid cuts of Appendix~\ref{app:proofs}. HiGHS reports a proven optimum on 27 of the 49 instances. We call these optima \emph{solver-certified}: the certificate rests on the solver's floating-point feasibility, integrality and optimality-gap tolerances and was not checked independently. As a check of the formulation, the program matches brute force on all graphs with at least one edge and at most seven vertices (Appendix~\ref{app:extra}).

On the 27 solver-certified instances (Table~\ref{tab:exact}, Appendix~\ref{app:extra}), with 0.2\,s per heuristic run and three seeds, tangent refinement alone is 21\% above the optimum on average, strict insertion (insertion moves accepting only improvements) 10\%, and orientation-space search (edge flips and path reversals, Corollary~\ref{cor:free}) 2\%. Iterated local search reaches the optimum on all 27. The certified instances are the easier ones. On the 22 others, the relative bound width $(\mathrm{UB}-\lceil\mathrm{LB}\rceil)/\mathrm{UB}$, where UB is the work of the best order found and LB the solver's lower bound, is 8--56\% (median 26\%; Table~\ref{tab:uncertified}). It upper-bounds the shortfall from optimality as a fraction of the best order's work, $(\mathrm{UB}-\Lopt{G})/\mathrm{UB}$; it does not estimate the actual shortfall.

\subsection{Public networks}

\paragraph{32 development graphs.}
These SNAP graphs \citep{leskovec2014snap} have up to 1.1M vertices and 2.0M edges and form three groups used at successive stages of development: development, validation and audit (a final check during development). All 32 informed the method; only the graphs below are held out. Table~\ref{tab:dev} in Appendix~\ref{app:extra} lists every graph. On 31 of them, the reference is a stored Neigh order (reversed on 12). \TILS reduces work by \SuiteTIKMean\% on average (95\% bootstrap interval \mbox{\SuiteTIKCI}; improved/tied/regressed \SuiteTIKITR; geometric-mean work ratio 0.941), against \SuiteDoveMean\% for \DOVE and \SuiteTMean\% for T. Against \DOVE, \TILS has wins/ties/losses \SuiteTIKvsDove\ (two-sided sign test, $p\SuiteTIKvsDoveP$), and against T, \SuiteTIKvsT.

\paragraph{Pre-registered held-out graphs.}
Because the graphs above informed development, we fixed ten unused SNAP graphs, together with the methods, budgets, seeds and metrics, before downloading any of them; two share data sources with development graphs (Appendix~\ref{app:prereg}). Facebook page--page failed to download and was reported as missing, not replaced. On the \HeldCount{} graphs completed, \TILS reduces work by \HeldTIKMean\% on average (95\% bootstrap interval \mbox{\HeldTIKCI}; improved/tied/regressed \HeldTIKITR; geometric-mean work ratio 0.847), against \HeldDoveMean\% for \DOVE and \HeldTMean\% for T. Against \DOVE, \TILS has wins/ties/losses \HeldTIKvsDove\ (two-sided sign test, $p=\HeldTIKvsDoveP$), and against T, \HeldTIKvsT\ (Table~\ref{tab:heldout}). Some runs ended slightly after their budget, and our runs did not log timestamped incumbents. A conservative analysis that scores each run by the last saved order known to exist at the deadline retains a mean reduction of 14.17\%, with 8/0/1 wins/ties/losses against \DOVE; this is a rescoring of saved orders, not a rerun with an enforced deadline (Appendix~\ref{app:protocol}).

The pre-registered held-out reference lacks Neigh. In the supplementary rescoring against the stronger reference $R^\star_N$, which was specified before any held-out Neigh order existed but is not part of the original pre-registration, the same \TILS orders reduce work by \NSFrozenTILSMean\% on average (geometric-mean work ratio 0.926) and still improve on all nine networks. When rerun from $R^\star_N$ on a second machine, \TILS reaches \NSReTILSMean\% (geometric-mean work ratio 0.923) and needs fewer candidate tests than \DOVE and T on all nine networks, also when scored conservatively by the orders known to exist at the deadline (\NSReTILSCappedMean\%). One of these graphs reuses the original run, so these paired comparisons mix machines and do not establish hardware-independent superiority (Appendix~\ref{app:neighsupp}).

\begin{table}[t]
\centering\small\setlength{\tabcolsep}{5pt}
\caption{Pre-registered held-out graphs: reference work and reduction (\%) relative to $R^\star$; \TILS: mean$\pm$SD over three seeds. Means are unweighted over graphs.}
\label{tab:heldout}
\begin{tabular}{lrrrrrr}
\toprule
Graph & $n$ & $m$ & $\Lam_{R^\star}$ & T & \DOVE & \TILS\\
\midrule
Email-Eu-core & 986 & 16.1k & 230.1k & 5.19 & 5.81 & 9.52$\pm$0.04 \\
CollegeMsg & 1.9k & 13.8k & 69.0k & 8.17 & 21.05 & 21.29$\pm$0.01 \\
LastFM-Asia & 7.6k & 27.8k & 96.0k & 5.44 & 6.07 & 8.45$\pm$0.05 \\
Deezer-Europe & 28.3k & 92.8k & 225.5k & 5.82 & 7.04 & 9.34$\pm$0.00 \\
\begin{tabular}[t]{@{}l@{}}Facebook\\page--page\end{tabular} & \multicolumn{6}{c}{download failed (excluded)} \\
Amazon0302 & 262.1k & 899.8k & 1.376M & 3.17 & 4.67 & 8.73$\pm$0.00 \\
Wiki-Talk & 2.39M & 4.66M & 57.386M & 11.48 & 10.97 & 13.04$\pm$0.02 \\
com-Youtube & 1.13M & 2.99M & 21.995M & 19.59 & 21.22 & 24.43$\pm$0.00 \\
web-Google & 875.7k & 4.32M & 16.783M & 1.85 & 2.22 & 5.16$\pm$0.00 \\
roadNet-CA & 1.97M & 2.77M & 861.6k & 8.66 & 16.39 & 33.50$\pm$0.09 \\
\midrule
Mean & & & & 7.71 & 10.61 & 14.83 \\
 
\bottomrule
\end{tabular}
\end{table}

\subsection{What drives the gains}\label{sec:ablation}
Table~\ref{tab:ablation} (Appendix~\ref{app:extra}) isolates the components at fixed budgets. Diversification is the main ingredient: plateau moves and kicks add 5.7, 2.4 and 0.5 percentage points over strict insertion on the synthetic development graphs (fresh instances of the tuning families), the synthetic graphs generated after tuning, and the four public development graphs. In this ablation (16 non-bipartite graphs, budgets of 1 and 10\,s), orientation-space search loses to iterated local search on 15 graphs and wins on one by 0.05 points. A final tangent pass cannot increase work and lowers it in 108 of the 123 public-graph runs, by 0.10 percentage points on average.

\subsection{Reusing orders along null-model chains}\label{sec:reuse}
Consecutive samples of the double-edge-swap chains used in motif analysis \citep{milo2002motifs} are related graphs. In a pre-registered study (Appendix~\ref{app:reuse}), we ran such chains from the degree sequences of eight graphs at five spacings that change \RChangedLowS--\RChangedHighS\% of the edges between samples (40 settings), against the reference without Neigh, $R^\star_{\rm cheap}$, recomputed on each sample. Tangent refinement of $R^\star_{\rm cheap}$ (Cheap+T) removes \RCheapTMin--\RCheapTMax\% of its work in all 40 settings. Repairing the previous sample's optimized order by tangent refinement (Reuse+T) beats Cheap+T only at \RChangedLowS\% change per sample (\RReuseTLowS\% vs.\ \RCheapTLowS\%, within \RHThreeGapMax\ percentage points of fresh \TILS at $\ge\RHThreeTimeFactorMin\times$ less ordering time); at \RChangedHighS\% it is worse (\RReuseTHighS\%), so the pre-registered reuse hypothesis failed. Reuse+T never lowers the time per sample, which is \RHFiveRatioMin--\RHFiveRatioMax\ times that of ascending degree: it would pay off only above \RGammaStarMin--\RGammaStarMax\,ns per candidate test, whereas we measure at most \RGammaHashMax\,ns, and at low drift, updating the triangle count incrementally \citep{makkar2017dynamic} is cheaper still when counts rather than listings are needed.

\subsection{Secondary analyses (not pre-registered)}\label{sec:posthoc}
The first three analyses use all \PHAllCount{} public graphs (Appendix~\ref{app:posthoc}), with wins/ties/losses of each variant against \TILS. \emph{Greedy-coloring orders} are far weaker than $R^\star$ (\PHAllColBestMean\% for the best of eight), ILS started from them does not help (\PHAllCILSTBestMean\% vs.\ \PHAllTILSMean\% for \TILS; \PHAllCILSTBestvsTILS), and \emph{ILS without the tangent prefix} reaches \PHAllILSTMean\% (\PHAllILSTvsTILS), so the prefix is neutral.

\emph{Listing time.} With a Numba port of the adaptive lister, the geometric-mean listing-time ratio of \TILS orders to $R^\star$ over the 41 graphs is \PHAllTimeRatioTILS{} (95\% interval \mbox{[0.982, 0.999]}) against a work ratio of \PHAllWorkRatioTILS; the held-out listing-time ratio of \PHHeldTimeRatioTILS{} (\mbox{[0.938, 0.984]}) was not reproduced on a second machine (\NSTimeFrozenGeo), and against $R^\star_N$ it is 1.0102 (\mbox{[0.9954, 1.0245]}); ordering takes seconds, listing milliseconds (Table~\ref{tab:ph-timing}).

\emph{Neigh in the held-out reference.} The supplementary run declared in advance (Appendix~\ref{app:neighsupp}) lowers the held-out reference on \NSRefChanged{} of \NSHeldCount{} graphs, by \NSGapMean\% on average and at most \NSGapMax\%. Rerun from this stronger reference $R^\star_N$ on a second machine, T, \DOVE and \TILS reach \NSReTMean\%, \NSReDoveMean\% and \NSReTILSMean\% (Table~\ref{tab:neighsupp}).

\emph{Other kernels.} For $k$-cliques, tangent refinement speeds up an adaptive lister by 20--27\% on Loc-Brightkite for $k=6$--$8$ (two implementations), but kClist \citep{danisch2018kcliques} remains faster (Appendix~\ref{app:kclique}), so the objective must match the lister's cost model.

\section{Conclusion}
Vertex ordering for adaptive triangle listing is concave minimization over the orientation polytope. This view gives certified rounding, an exact mixed-integer program, and a reason why unconstrained score relaxations at a fixed temperature saturate. Separately, our random-graph analysis shows that degree information alone can lose a constant factor. Since minimizing the work is NP-hard and hard to approximate, we pair certified rounding with diversified discrete search, which reduces work on every public graph tested, including against the stronger Neigh-inclusive reference. In practice, orders designed for reversal-invariant costs should be evaluated in both directions, and discrete heuristics should finish with a tangent pass. For minimizing candidate tests, diversified discrete search outperformed score-based gradient optimization in all our comparisons; for one-off wall-clock listing, the simple degree and degeneracy orders remain the practical choice with our implementations.

\paragraph{Limitations.} The gains are in candidate tests: with our implementations, listing time changes little and not reproducibly across machines, and ordering, even with reuse, costs more than it saves. An optimized order pays for its extra ordering time only after $N=(t^{\rm opt}_{\rm order}-t^{\rm base}_{\rm order})/(t^{\rm base}_{\rm list}-t^{\rm opt}_{\rm list})$ listings, and only if the denominator is positive; on the held-out graphs, the evaluation-machine timings of Table~\ref{tab:ph-timing} give no break-even on Email-Eu-core and roughly $10^4$ to $10^5$ listings on the others, and this listing-time gain was not reproduced on a second machine. Optimality is solver-certified only on small instances, the hardness results are worst-case, Theorem~\ref{thm:mf} assumes bounded degrees and concerns expectations, and the degree gap is proved for cubic graphs only. Generalization rests on nine held-out networks; against the Neigh-inclusive reference, the held-out gain is \NSFrozenTILSMean\% (\NSReTILSMean\% in reruns).

\paragraph{Future directions.} Open questions include approximation algorithms for $\min\Lam$ (Min UnCut admits an $O(\sqrt{\log n})$ approximation \citep{agarwal2005uncut}), concentration and the optimality of V-shaped layouts, degree gaps on heavy-tailed and geometric graphs, such as the adjacency graphs of the Voronoi tessellations formed by bacterial communities \citep{gorgi2026geometric}, annealed relaxations, $k$-cliques and other concave objectives (Appendices~\ref{app:kclique} and~\ref{app:entropy}), and molecular cutoff graphs \citep{pinzon2025threebody}.

\section*{Code and Data Availability}
The code, protocols, run logs and per-graph results of the pre-registered evaluation, the secondary analyses and the Neigh-inclusive held-out run are available at \repo. All graphs are public SNAP datasets \citep{leskovec2014snap}; the repository downloads and preprocesses them rather than redistributing them.

\section*{AI Use Statement}
Generative AI tools were used to assist with manuscript preparation and code debugging. The authors reviewed, verified, and take full responsibility for all content presented in this work.

\bibliographystyle{plainnat}
\bibliography{references}

\begin{thebibliography}{89}
\providecommand{\natexlab}[1]{#1}
\providecommand{\url}[1]{\texttt{#1}}
\expandafter\ifx\csname urlstyle\endcsname\relax
  \providecommand{\doi}[1]{doi: #1}\else
  \providecommand{\doi}{doi: \begingroup \urlstyle{rm}\Url}\fi

\bibitem[Agarwal et~al.(2005)Agarwal, Charikar, Makarychev, and
  Makarychev]{agarwal2005uncut}
Amit Agarwal, Moses Charikar, Konstantin Makarychev, and Yury Makarychev.
\newblock {$O(\sqrt{\log n})$} approximation algorithms for {Min UnCut}, {Min
  2CNF Deletion}, and directed cut problems.
\newblock In \emph{Proceedings of the 37th Annual ACM Symposium on Theory of
  Computing (STOC)}, pages 573--581, 2005.
\newblock \doi{10.1145/1060590.1060675}.

\bibitem[Aldous and Kolesnik(2022)]{aldous2018scores}
David~J. Aldous and Brett Kolesnik.
\newblock To stay discovered: On tournament mean score sequences and the
  {Bradley--Terry} model.
\newblock \emph{Stochastic Processes and their Applications}, 150:\penalty0
  844--852, 2022.
\newblock \doi{10.1016/j.spa.2019.10.011}.

\bibitem[Alimonti and Kann(2000)]{alimonti2000cubic}
Paola Alimonti and Viggo Kann.
\newblock Some {APX}-completeness results for cubic graphs.
\newblock \emph{Theoretical Computer Science}, 237\penalty0 (1--2):\penalty0
  123--134, 2000.
\newblock \doi{10.1016/S0304-3975(98)00158-3}.

\bibitem[Alon et~al.(1997)Alon, Yuster, and Zwick]{alon1997cycles}
Noga Alon, Raphael Yuster, and Uri Zwick.
\newblock Finding and counting given length cycles.
\newblock \emph{Algorithmica}, 17\penalty0 (3):\penalty0 209--223, 1997.
\newblock \doi{10.1007/BF02523189}.

\bibitem[Alon(2007)]{alon2007motifs}
Uri Alon.
\newblock Network motifs: Theory and experimental approaches.
\newblock \emph{Nature Reviews Genetics}, 8\penalty0 (6):\penalty0 450--461,
  2007.
\newblock \doi{10.1038/nrg2102}.

\bibitem[Arifuzzaman et~al.(2013)Arifuzzaman, Khan, and
  Marathe]{arifuzzaman2013patric}
Shaikh Arifuzzaman, Maleq Khan, and Madhav Marathe.
\newblock {PATRIC}: A parallel algorithm for counting triangles in massive
  networks.
\newblock In \emph{Proceedings of the 22nd ACM International Conference on
  Information and Knowledge Management}, pages 529--538, 2013.
\newblock \doi{10.1145/2505515.2505545}.

\bibitem[Armstrong et~al.(2009)Armstrong, Moffat, Webber, and
  Zobel]{armstrong2009improvements}
Timothy~G. Armstrong, Alistair Moffat, William Webber, and Justin Zobel.
\newblock Improvements that don't add up: Ad-hoc retrieval results since 1998.
\newblock In \emph{Proceedings of the 18th ACM Conference on Information and
  Knowledge Management}, pages 601--610, 2009.
\newblock \doi{10.1145/1645953.1646031}.

\bibitem[Batagelj and Zaver{\v{s}}nik(2003)]{batagelj2003core}
Vladimir Batagelj and Matja{\v{z}} Zaver{\v{s}}nik.
\newblock An {$O(m)$} algorithm for cores decomposition of networks.
\newblock arXiv preprint cs/0310049, 2003.

\bibitem[Beck and Teboulle(2003)]{beck2003mirror}
Amir Beck and Marc Teboulle.
\newblock Mirror descent and nonlinear projected subgradient methods for convex
  optimization.
\newblock \emph{Operations Research Letters}, 31\penalty0 (3):\penalty0
  167--175, 2003.
\newblock \doi{10.1016/S0167-6377(02)00231-6}.

\bibitem[Benson et~al.(2016)Benson, Gleich, and Leskovec]{benson2016higher}
Austin~R. Benson, David~F. Gleich, and Jure Leskovec.
\newblock Higher-order organization of complex networks.
\newblock \emph{Science}, 353\penalty0 (6295):\penalty0 163--166, 2016.
\newblock \doi{10.1126/science.aad9029}.

\bibitem[Bj{\"o}rklund et~al.(2014)Bj{\"o}rklund, Pagh, Vassilevska~Williams,
  and Zwick]{bjorklund2014listing}
Andreas Bj{\"o}rklund, Rasmus Pagh, Virginia Vassilevska~Williams, and Uri
  Zwick.
\newblock Listing triangles.
\newblock In \emph{Automata, Languages, and Programming (ICALP)}, volume 8572
  of \emph{Lecture Notes in Computer Science}, pages 223--234. Springer, 2014.
\newblock \doi{10.1007/978-3-662-43948-7_19}.

\bibitem[Blondel et~al.(2020)Blondel, Teboul, Berthet, and
  Djolonga]{blondel2020sorting}
Mathieu Blondel, Olivier Teboul, Quentin Berthet, and Josip Djolonga.
\newblock Fast differentiable sorting and ranking.
\newblock In \emph{Proceedings of the 37th International Conference on Machine
  Learning}, volume 119 of \emph{Proceedings of Machine Learning Research},
  pages 950--959. PMLR, 2020.
\newblock URL \url{https://proceedings.mlr.press/v119/blondel20a.html}.

\bibitem[Bogu{\~n}{\'a} et~al.(2004)Bogu{\~n}{\'a}, Pastor-Satorras, and
  Vespignani]{boguna2004cutoffs}
Mari{\'a}n Bogu{\~n}{\'a}, Romualdo Pastor-Satorras, and Alessandro Vespignani.
\newblock Cut-offs and finite size effects in scale-free networks.
\newblock \emph{The European Physical Journal B}, 38\penalty0 (2):\penalty0
  205--209, 2004.
\newblock \doi{10.1140/epjb/e2004-00038-8}.

\bibitem[Bollob{\'a}s(1980)]{bollobas1980configuration}
B{\'e}la Bollob{\'a}s.
\newblock A probabilistic proof of an asymptotic formula for the number of
  labelled regular graphs.
\newblock \emph{European Journal of Combinatorics}, 1\penalty0 (4):\penalty0
  311--316, 1980.
\newblock \doi{10.1016/S0195-6698(80)80030-8}.

\bibitem[Borradaile et~al.(2017)Borradaile, Iglesias, Migler, Ochoa, Wilfong,
  and Zhang]{borradaile2017egalitarian}
Glencora Borradaile, Jennifer Iglesias, Theresa Migler, Antonio Ochoa, Gordon
  Wilfong, and Lisa Zhang.
\newblock Egalitarian graph orientations.
\newblock \emph{Journal of Graph Algorithms and Applications}, 21\penalty0
  (4):\penalty0 687--708, 2017.
\newblock \doi{10.7155/jgaa.00435}.

\bibitem[Bouritsas et~al.(2023)Bouritsas, Frasca, Zafeiriou, and
  Bronstein]{bouritsas2023gsn}
Giorgos Bouritsas, Fabrizio Frasca, Stefanos Zafeiriou, and Michael~M.
  Bronstein.
\newblock Improving graph neural network expressivity via subgraph isomorphism
  counting.
\newblock \emph{IEEE Transactions on Pattern Analysis and Machine
  Intelligence}, 45\penalty0 (1):\penalty0 657--668, 2023.
\newblock \doi{10.1109/TPAMI.2022.3154319}.

\bibitem[Bradley and Terry(1952)]{bradley1952rank}
Ralph~Allan Bradley and Milton~E. Terry.
\newblock Rank analysis of incomplete block designs: {I}. {T}he method of
  paired comparisons.
\newblock \emph{Biometrika}, 39\penalty0 (3/4):\penalty0 324--345, 1952.
\newblock \doi{10.2307/2334029}.

\bibitem[Burges et~al.(2005)Burges, Shaked, Renshaw, Lazier, Deeds, Hamilton,
  and Hullender]{burges2005ranknet}
Chris Burges, Tal Shaked, Erin Renshaw, Ari Lazier, Matt Deeds, Nicole
  Hamilton, and Greg Hullender.
\newblock Learning to rank using gradient descent.
\newblock In \emph{Proceedings of the 22nd International Conference on Machine
  Learning}, pages 89--96, 2005.
\newblock \doi{10.1145/1102351.1102363}.

\bibitem[Cardinal et~al.(2008)Cardinal, Fiorini, and
  Joret]{cardinal2008entropy}
Jean Cardinal, Samuel Fiorini, and Gwena{\"e}l Joret.
\newblock Minimum entropy orientations.
\newblock \emph{Operations Research Letters}, 36\penalty0 (6):\penalty0
  680--683, 2008.
\newblock \doi{10.1016/j.orl.2008.06.010}.

\bibitem[Chen et~al.(2020)Chen, Chen, Villar, and Bruna]{chen2020substructures}
Zhengdao Chen, Lei Chen, Soledad Villar, and Joan Bruna.
\newblock Can graph neural networks count substructures?
\newblock In \emph{Advances in Neural Information Processing Systems},
  volume~33, pages 10383--10395, 2020.
\newblock URL
  \url{https://proceedings.neurips.cc/paper/2020/hash/75877cb75154206c4e65e76b88a12712-Abstract.html}.

\bibitem[Chiba and Nishizeki(1985)]{chiba1985arboricity}
Norishige Chiba and Takao Nishizeki.
\newblock Arboricity and subgraph listing algorithms.
\newblock \emph{SIAM Journal on Computing}, 14\penalty0 (1):\penalty0 210--223,
  1985.
\newblock \doi{10.1137/0214017}.

\bibitem[Chu and Cheng(2011)]{chu2011triangle}
Shumo Chu and James Cheng.
\newblock Triangle listing in massive networks and its applications.
\newblock In \emph{Proceedings of the 17th ACM SIGKDD International Conference
  on Knowledge Discovery and Data Mining}, pages 672--680, 2011.
\newblock \doi{10.1145/2020408.2020513}.

\bibitem[Cuturi et~al.(2019)Cuturi, Teboul, and Vert]{cuturi2019sorting}
Marco Cuturi, Olivier Teboul, and Jean-Philippe Vert.
\newblock Differentiable ranking and sorting using optimal transport.
\newblock In \emph{Advances in Neural Information Processing Systems},
  volume~32, pages 6861--6871, 2019.

\bibitem[Danisch et~al.(2018)Danisch, Balalau, and Sozio]{danisch2018kcliques}
Maximilien Danisch, Oana Balalau, and Mauro Sozio.
\newblock Listing $k$-cliques in sparse real-world graphs.
\newblock In \emph{Proceedings of the 2018 World Wide Web Conference (WWW)},
  pages 589--598, 2018.
\newblock \doi{10.1145/3178876.3186125}.

\bibitem[Deng et~al.(2023)Deng, Bello, Aragam, and Ravikumar]{deng2023notears}
Chang Deng, Kevin Bello, Bryon Aragam, and Pradeep~Kumar Ravikumar.
\newblock Optimizing {NOTEARS} objectives via topological swaps.
\newblock In \emph{Proceedings of the 40th International Conference on Machine
  Learning}, volume 202 of \emph{Proceedings of Machine Learning Research},
  pages 7563--7595. PMLR, 2023.
\newblock URL \url{https://proceedings.mlr.press/v202/deng23a.html}.

\bibitem[Edmonds(1970)]{edmonds1970submodular}
Jack Edmonds.
\newblock Submodular functions, matroids, and certain polyhedra.
\newblock In \emph{Combinatorial Structures and Their Applications}, pages
  69--87. Gordon and Breach, 1970.

\bibitem[Eppstein et~al.(2017)Eppstein, Goodrich, Mitzenmacher, and
  Torres]{eppstein2017cuckoo}
David Eppstein, Michael~T. Goodrich, Michael Mitzenmacher, and Manuel~R.
  Torres.
\newblock 2-3 cuckoo filters for faster triangle listing and set intersection.
\newblock In \emph{Proceedings of the 36th ACM SIGMOD-SIGACT-SIGAI Symposium on
  Principles of Database Systems (PODS)}, pages 247--260, 2017.
\newblock \doi{10.1145/3034786.3056115}.

\bibitem[Frank and Wolfe(1956)]{frank1956algorithm}
Marguerite Frank and Philip Wolfe.
\newblock An algorithm for quadratic programming.
\newblock \emph{Naval Research Logistics Quarterly}, 3\penalty0
  (1--2):\penalty0 95--110, 1956.
\newblock \doi{10.1002/nav.3800030109}.

\bibitem[Frieze and Suen(1994)]{frieze1994independence}
Alan~M. Frieze and Stephen Suen.
\newblock On the independence number of random cubic graphs.
\newblock \emph{Random Structures \& Algorithms}, 5\penalty0 (5):\penalty0
  649--664, 1994.
\newblock \doi{10.1002/rsa.3240050504}.

\bibitem[Gorgi and Mirsohi(2026)]{gorgi2026auditing}
Melika Gorgi and Kourosh Mirsohi.
\newblock Auditing {Bayesian} graph alignment: Diagnostic comparisons and
  reference failure.
\newblock arXiv preprint arXiv:2609.23232, 2026.

\bibitem[Gorgi et~al.(2026)Gorgi, Kasallis, Trinh, Ortiz~de Ora, Wiles, and
  Siryaporn]{gorgi2026geometric}
Melika Gorgi, Summer~J. Kasallis, Calvin Trinh, Lizett Ortiz~de Ora, Travis~J.
  Wiles, and Albert Siryaporn.
\newblock Geometric ordering in bacterial communities.
\newblock \emph{Proceedings of the National Academy of Sciences}, 123\penalty0
  (20):\penalty0 e2526643123, 2026.
\newblock \doi{10.1073/pnas.2526643123}.

\bibitem[Greene and Zaslavsky(1983)]{greene1983interpretation}
Curtis Greene and Thomas Zaslavsky.
\newblock On the interpretation of {W}hitney numbers through arrangements of
  hyperplanes, zonotopes, non-{R}adon partitions, and orientations of graphs.
\newblock \emph{Transactions of the American Mathematical Society},
  280\penalty0 (1):\penalty0 97--126, 1983.
\newblock \doi{10.1090/S0002-9947-1983-0712251-1}.

\bibitem[Grover et~al.(2019)Grover, Wang, Zweig, and
  Ermon]{grover2019neuralsort}
Aditya Grover, Eric Wang, Aaron Zweig, and Stefano Ermon.
\newblock Stochastic optimization of sorting networks via continuous
  relaxations.
\newblock In \emph{International Conference on Learning Representations}, 2019.
\newblock URL \url{https://arxiv.org/abs/1903.08850}.

\bibitem[Hagberg et~al.(2008)Hagberg, Schult, and Swart]{hagberg2008networkx}
Aric~A. Hagberg, Daniel~A. Schult, and Pieter~J. Swart.
\newblock Exploring network structure, dynamics, and function using {NetworkX}.
\newblock In Ga{\"e}l Varoquaux, Travis Vaught, and Jarrod Millman, editors,
  \emph{Proceedings of the 7th Python in Science Conference (SciPy 2008)},
  pages 11--15, Pasadena, CA, USA, 2008.

\bibitem[Hakimi(1965)]{hakimi1965degrees}
S.~L. Hakimi.
\newblock On the degrees of the vertices of a directed graph.
\newblock \emph{Journal of the Franklin Institute}, 279\penalty0 (4):\penalty0
  290--308, 1965.
\newblock \doi{10.1016/0016-0032(65)90340-6}.

\bibitem[Harvey et~al.(2006)Harvey, Ladner, Lov{\'a}sz, and
  Tamir]{harvey2006semimatchings}
Nicholas J.~A. Harvey, Richard~E. Ladner, L{\'a}szl{\'o} Lov{\'a}sz, and Tami
  Tamir.
\newblock Semi-matchings for bipartite graphs and load balancing.
\newblock \emph{Journal of Algorithms}, 59\penalty0 (1):\penalty0 53--78, 2006.
\newblock \doi{10.1016/j.jalgor.2005.01.003}.

\bibitem[Heckman and Thomas(2001)]{heckman2001independence}
Christopher~Carl Heckman and Robin Thomas.
\newblock A new proof of the independence ratio of triangle-free cubic graphs.
\newblock \emph{Discrete Mathematics}, 233\penalty0 (1--3):\penalty0 233--237,
  2001.
\newblock \doi{10.1016/S0012-365X(00)00242-9}.

\bibitem[Horst and Tuy(1996)]{horst1996global}
Reiner Horst and Hoang Tuy.
\newblock \emph{Global Optimization: Deterministic Approaches}.
\newblock Springer, Berlin, third edition, 1996.
\newblock \doi{10.1007/978-3-662-03199-5}.

\bibitem[Hu et~al.(2013)Hu, Tao, and Chung]{hu2013massive}
Xiaocheng Hu, Yufei Tao, and Chin-Wan Chung.
\newblock Massive graph triangulation.
\newblock In \emph{Proceedings of the 2013 ACM SIGMOD International Conference
  on Management of Data}, pages 325--336, 2013.
\newblock \doi{10.1145/2463676.2463704}.

\bibitem[Huangfu and Hall(2018)]{huangfu2018highs}
Qi~Huangfu and J.~A.~Julian Hall.
\newblock Parallelizing the dual revised simplex method.
\newblock \emph{Mathematical Programming Computation}, 10\penalty0
  (1):\penalty0 119--142, 2018.
\newblock \doi{10.1007/s12532-017-0130-5}.

\bibitem[Hunter(2004)]{hunter2004mm}
David~R. Hunter.
\newblock {MM} algorithms for generalized {Bradley--Terry} models.
\newblock \emph{The Annals of Statistics}, 32\penalty0 (1):\penalty0 384--406,
  2004.
\newblock \doi{10.1214/aos/1079120141}.

\bibitem[Itai and Rodeh(1978)]{itai1978minimum}
Alon Itai and Michael Rodeh.
\newblock Finding a minimum circuit in a graph.
\newblock \emph{SIAM Journal on Computing}, 7\penalty0 (4):\penalty0 413--423,
  1978.
\newblock \doi{10.1137/0207033}.

\bibitem[Jaggi(2013)]{jaggi2013revisiting}
Martin Jaggi.
\newblock Revisiting {F}rank--{W}olfe: Projection-free sparse convex
  optimization.
\newblock In \emph{Proceedings of the 30th International Conference on Machine
  Learning}, volume~28 of \emph{Proceedings of Machine Learning Research},
  pages 427--435. PMLR, 2013.

\bibitem[Khot(2002)]{khot2002unique}
Subhash Khot.
\newblock On the power of unique 2-prover 1-round games.
\newblock In \emph{Proceedings of the 34th Annual ACM Symposium on Theory of
  Computing (STOC)}, pages 767--775, 2002.
\newblock \doi{10.1145/509907.510017}.

\bibitem[Khot et~al.(2007)Khot, Kindler, Mossel, and
  O'Donnell]{khot2007optimal}
Subhash Khot, Guy Kindler, Elchanan Mossel, and Ryan O'Donnell.
\newblock Optimal inapproximability results for {MAX-CUT} and other 2-variable
  {CSPs}?
\newblock \emph{SIAM Journal on Computing}, 37\penalty0 (1):\penalty0 319--357,
  2007.
\newblock \doi{10.1137/S0097539705447372}.

\bibitem[Kingma and Ba(2015)]{kingma2015adam}
Diederik~P. Kingma and Jimmy Ba.
\newblock Adam: A method for stochastic optimization.
\newblock In \emph{International Conference on Learning Representations}, 2015.

\bibitem[Langedal and Manne(2025)]{langedal2025coloring}
Kenneth Langedal and Fredrik Manne.
\newblock Graph neural networks as ordering heuristics for parallel graph
  coloring.
\newblock In \emph{Proceedings of the Symposium on Algorithm Engineering and
  Experiments (ALENEX)}, pages 56--67. SIAM, 2025.
\newblock \doi{10.1137/1.9781611978339.5}.

\bibitem[Latapy(2008)]{latapy2008main}
Matthieu Latapy.
\newblock Main-memory triangle computations for very large (sparse (power-law))
  graphs.
\newblock \emph{Theoretical Computer Science}, 407\penalty0 (1--3):\penalty0
  458--473, 2008.
\newblock \doi{10.1016/j.tcs.2008.07.017}.

\bibitem[L{\'e}cuyer(2023)]{lecuyer2023thesis}
Fabrice L{\'e}cuyer.
\newblock \emph{Ordering nodes to scale to large real-world networks}.
\newblock PhD thesis, Sorbonne Universit{\'e}, 2023.
\newblock URL
  \url{https://fabrice.lecuyer.me/public/pdf/Lecuyer_2023_Thesis.pdf}.

\bibitem[L{\'e}cuyer et~al.(2023)L{\'e}cuyer, Jachiet, Magnien, and
  Tabourier]{lecuyer2023tailored}
Fabrice L{\'e}cuyer, Louis Jachiet, Cl{\'e}mence Magnien, and Lionel Tabourier.
\newblock Tailored vertex ordering for faster triangle listing in large graphs.
\newblock In \emph{Proceedings of the Symposium on Algorithm Engineering and
  Experiments (ALENEX)}, pages 77--85. SIAM, 2023.
\newblock \doi{10.1137/1.9781611977561.ch7}.

\bibitem[Leskovec and Krevl(2014)]{leskovec2014snap}
Jure Leskovec and Andrej Krevl.
\newblock {SNAP} datasets: {Stanford} large network dataset collection.
\newblock \url{https://snap.stanford.edu/data}, 2014.

\bibitem[Li et~al.(2020)Li, Gao, Qin, Wang, Yang, and Yu]{li2020kclique}
Rong-Hua Li, Sen Gao, Lu~Qin, Guoren Wang, Weihua Yang, and Jeffrey~Xu Yu.
\newblock Ordering heuristics for $k$-clique listing.
\newblock \emph{Proceedings of the VLDB Endowment}, 13\penalty0 (11):\penalty0
  2536--2548, 2020.
\newblock \doi{10.14778/3407790.3407843}.

\bibitem[Lipton and Steinhardt(2019)]{lipton2019troubling}
Zachary~C. Lipton and Jacob Steinhardt.
\newblock Troubling trends in machine learning scholarship.
\newblock \emph{Queue}, 17\penalty0 (1):\penalty0 45--77, 2019.
\newblock \doi{10.1145/3317287.3328534}.

\bibitem[Louren{\c{c}}o et~al.(2003)Louren{\c{c}}o, Martin, and
  St{\"u}tzle]{lourenco2003ils}
Helena~R. Louren{\c{c}}o, Olivier~C. Martin, and Thomas St{\"u}tzle.
\newblock Iterated local search.
\newblock In Fred Glover and Gary~A. Kochenberger, editors, \emph{Handbook of
  Metaheuristics}, volume~57 of \emph{International Series in Operations
  Research \& Management Science}, pages 320--353. Kluwer Academic Publishers,
  2003.
\newblock \doi{10.1007/0-306-48056-5_11}.

\bibitem[Makkar et~al.(2017)Makkar, Bader, and Green]{makkar2017dynamic}
Devavret Makkar, David~A. Bader, and Oded Green.
\newblock Exact and parallel triangle counting in dynamic graphs.
\newblock In \emph{Proceedings of the 24th IEEE International Conference on
  High Performance Computing (HiPC)}, pages 2--12. IEEE, 2017.
\newblock \doi{10.1109/HiPC.2017.00011}.

\bibitem[Mangasarian(1997)]{mangasarian1997concave}
O.~L. Mangasarian.
\newblock Solution of general linear complementarity problems via
  nondifferentiable concave minimization.
\newblock \emph{Acta Mathematica Vietnamica}, 22\penalty0 (1):\penalty0
  199--205, 1997.
\newblock URL \url{https://math.ac.vn/public/uploads/files/9701199.pdf}.

\bibitem[Matula and Beck(1983)]{matula1983smallest}
David~W. Matula and Leland~L. Beck.
\newblock Smallest-last ordering and clustering and graph coloring algorithms.
\newblock \emph{Journal of the ACM}, 30\penalty0 (3):\penalty0 417--427, 1983.
\newblock \doi{10.1145/2402.322385}.

\bibitem[Melis et~al.(2018)Melis, Dyer, and Blunsom]{melis2018state}
G{\'a}bor Melis, Chris Dyer, and Phil Blunsom.
\newblock On the state of the art of evaluation in neural language models.
\newblock In \emph{International Conference on Learning Representations}, 2018.

\bibitem[Mena et~al.(2018)Mena, Belanger, Linderman, and Snoek]{mena2018gumbel}
Gonzalo Mena, David Belanger, Scott Linderman, and Jasper Snoek.
\newblock Learning latent permutations with {Gumbel-Sinkhorn} networks.
\newblock In \emph{International Conference on Learning Representations}, 2018.

\bibitem[Milo et~al.(2002)Milo, Shen-Orr, Itzkovitz, Kashtan, Chklovskii, and
  Alon]{milo2002motifs}
Ron Milo, Shai Shen-Orr, Shalev Itzkovitz, Nadav Kashtan, Dmitri Chklovskii,
  and Uri Alon.
\newblock Network motifs: Simple building blocks of complex networks.
\newblock \emph{Science}, 298\penalty0 (5594):\penalty0 824--827, 2002.
\newblock \doi{10.1126/science.298.5594.824}.

\bibitem[Min and Gomes(2025)]{min2025clique}
Yimeng Min and Carla~P. Gomes.
\newblock Unsupervised ordering for maximum clique.
\newblock arXiv preprint arXiv:2503.21814, 2025.

\bibitem[Molloy and Reed(1995)]{molloy1995critical}
Michael Molloy and Bruce Reed.
\newblock A critical point for random graphs with a given degree sequence.
\newblock \emph{Random Structures \& Algorithms}, 6\penalty0 (2--3):\penalty0
  161--180, 1995.
\newblock \doi{10.1002/rsa.3240060204}.

\bibitem[Nerem et~al.(2025)Nerem, Luo, Rafiey, and Wang]{nerem2024birkhoff}
Robert~R. Nerem, Zhishang Luo, Akbar Rafiey, and Yusu Wang.
\newblock Differentiable extensions with rounding guarantees for combinatorial
  optimization over permutations.
\newblock In \emph{Advances in Neural Information Processing Systems},
  volume~38, pages 117299--117338, 2025.
\newblock \doi{10.52202/085713-3535}.

\bibitem[Newman et~al.(2001)Newman, Strogatz, and Watts]{newman2001random}
Mark E.~J. Newman, Steven~H. Strogatz, and Duncan~J. Watts.
\newblock Random graphs with arbitrary degree distributions and their
  applications.
\newblock \emph{Physical Review E}, 64\penalty0 (2):\penalty0 026118, 2001.
\newblock \doi{10.1103/PhysRevE.64.026118}.

\bibitem[Nosek et~al.(2018)Nosek, Ebersole, DeHaven, and
  Mellor]{nosek2018preregistration}
Brian~A. Nosek, Charles~R. Ebersole, Alexander~C. DeHaven, and David~T. Mellor.
\newblock The preregistration revolution.
\newblock \emph{Proceedings of the National Academy of Sciences}, 115\penalty0
  (11):\penalty0 2600--2606, 2018.
\newblock \doi{10.1073/pnas.1708274114}.

\bibitem[Ortmann and Brandes(2014)]{ortmann2014triangle}
Mark Ortmann and Ulrik Brandes.
\newblock Triangle listing algorithms: Back from the diversion.
\newblock In \emph{Proceedings of the Sixteenth Workshop on Algorithm
  Engineering and Experiments (ALENEX)}, pages 1--8. SIAM, 2014.
\newblock \doi{10.1137/1.9781611973198.1}.

\bibitem[Pardalos and Rosen(1986)]{pardalos1986concave}
Panos~M. Pardalos and J.~Ben Rosen.
\newblock Methods for global concave minimization: A bibliographic survey.
\newblock \emph{SIAM Review}, 28\penalty0 (3):\penalty0 367--379, 1986.
\newblock \doi{10.1137/1028106}.

\bibitem[Petersen et~al.(2021)Petersen, Borgelt, Kuehne, and
  Deussen]{petersen2021sortnet}
Felix Petersen, Christian Borgelt, Hilde Kuehne, and Oliver Deussen.
\newblock Differentiable sorting networks for scalable sorting and ranking
  supervision.
\newblock In \emph{Proceedings of the 38th International Conference on Machine
  Learning}, volume 139 of \emph{Proceedings of Machine Learning Research},
  pages 8546--8555. PMLR, 2021.
\newblock URL \url{https://proceedings.mlr.press/v139/petersen21a.html}.

\bibitem[Pinz{\'o}n~Escobar et~al.(2025)Pinz{\'o}n~Escobar,
  M{\"u}hlh{\"a}u{\ss}er, Bungartz, and Neumann]{pinzon2025threebody}
Jos{\'e}~Alfonso Pinz{\'o}n~Escobar, Markus M{\"u}hlh{\"a}u{\ss}er,
  Hans-Joachim Bungartz, and Philipp Neumann.
\newblock Linked cell traversal algorithms for three-body interactions in
  molecular dynamics.
\newblock arXiv preprint arXiv:2510.21230, 2025.

\bibitem[Postnikov(2009)]{postnikov2009permutohedra}
Alexander Postnikov.
\newblock Permutohedra, associahedra, and beyond.
\newblock \emph{International Mathematics Research Notices}, 2009\penalty0
  (6):\penalty0 1026--1106, 2009.
\newblock \doi{10.1093/imrn/rnn153}.

\bibitem[Prillo and Eisenschlos(2020)]{prillo2020softsort}
Sebastian Prillo and Julian~Martin Eisenschlos.
\newblock {SoftSort}: A continuous relaxation for the argsort operator.
\newblock In \emph{Proceedings of the 37th International Conference on Machine
  Learning}, volume 119 of \emph{Proceedings of Machine Learning Research},
  pages 7793--7802. PMLR, 2020.
\newblock URL \url{https://proceedings.mlr.press/v119/prillo20a.html}.

\bibitem[Pr{\v{z}}ulj(2007)]{przulj2007graphlets}
Nata{\v{s}}a Pr{\v{z}}ulj.
\newblock Biological network comparison using graphlet degree distribution.
\newblock \emph{Bioinformatics}, 23\penalty0 (2):\penalty0 e177--e183, 2007.
\newblock \doi{10.1093/bioinformatics/btl301}.

\bibitem[Schank and Wagner(2005)]{schank2005triangles}
Thomas Schank and Dorothea Wagner.
\newblock Finding, counting and listing all triangles in large graphs, an
  experimental study.
\newblock In \emph{Experimental and Efficient Algorithms (WEA)}, volume 3503 of
  \emph{Lecture Notes in Computer Science}, pages 606--609. Springer, 2005.
\newblock \doi{10.1007/11427186_54}.

\bibitem[Schwabe and Acosta(2026)]{schwabe2025join}
Tim Schwabe and Maribel Acosta.
\newblock Gradient-based join ordering.
\newblock In \emph{Proceedings of the Thirty-Fifth International Joint
  Conference on Artificial Intelligence}, pages 2346--2354, 2026.
\newblock \doi{10.24963/ijcai.2026/261}.
\newblock URL \url{https://www.ijcai.org/proceedings/2026/261}.

\bibitem[Shen-Orr et~al.(2002)Shen-Orr, Milo, Mangan, and
  Alon]{shenorr2002motifs}
Shai~S. Shen-Orr, Ron Milo, Shmoolik Mangan, and Uri Alon.
\newblock Network motifs in the transcriptional regulation network of
  {\emph{escherichia coli}}.
\newblock \emph{Nature Genetics}, 31\penalty0 (1):\penalty0 64--68, 2002.
\newblock \doi{10.1038/ng881}.

\bibitem[Shun and Tangwongsan(2015)]{shun2015multicore}
Julian Shun and Kanat Tangwongsan.
\newblock Multicore triangle computations without tuning.
\newblock In \emph{Proceedings of the 31st IEEE International Conference on
  Data Engineering}, pages 149--160, 2015.
\newblock \doi{10.1109/ICDE.2015.7113280}.

\bibitem[Sporns and K{\"o}tter(2004)]{sporns2004motifs}
Olaf Sporns and Rolf K{\"o}tter.
\newblock Motifs in brain networks.
\newblock \emph{PLoS Biology}, 2\penalty0 (11):\penalty0 e369, 2004.
\newblock \doi{10.1371/journal.pbio.0020369}.

\bibitem[Stanley(1980)]{stanley1980decompositions}
Richard~P. Stanley.
\newblock Decompositions of rational convex polytopes.
\newblock In \emph{Combinatorial Mathematics, Optimal Designs and Their
  Applications}, volume~6 of \emph{Annals of Discrete Mathematics}, pages
  333--342. North-Holland, 1980.
\newblock \doi{10.1016/S0167-5060(08)70717-9}.

\bibitem[Staton(1979)]{staton1979independence}
William Staton.
\newblock Some {R}amsey-type numbers and the independence ratio.
\newblock \emph{Transactions of the American Mathematical Society},
  256:\penalty0 353--370, 1979.
\newblock \doi{10.1090/S0002-9947-1979-0546922-6}.

\bibitem[Struski et~al.(2025)Struski, Bednarczyk, Podolak, and
  Tabor]{struski2025lapsum}
{\L}ukasz Struski, Micha{\l}~B. Bednarczyk, Igor~T. Podolak, and Jacek Tabor.
\newblock {LapSum} -- one method to differentiate them all: Ranking, sorting
  and top-$k$ selection.
\newblock In \emph{Proceedings of the 42nd International Conference on Machine
  Learning}, volume 267 of \emph{Proceedings of Machine Learning Research},
  pages 56990--57007. PMLR, 2025.
\newblock URL \url{https://proceedings.mlr.press/v267/struski25a.html}.

\bibitem[Suri and Vassilvitskii(2011)]{suri2011counting}
Siddharth Suri and Sergei Vassilvitskii.
\newblock Counting triangles and the curse of the last reducer.
\newblock In \emph{Proceedings of the 20th International Conference on World
  Wide Web}, pages 607--614, 2011.
\newblock \doi{10.1145/1963405.1963491}.

\bibitem[Tsourakakis(2008)]{tsourakakis2008fast}
Charalampos~E. Tsourakakis.
\newblock Fast counting of triangles in large real networks without counting:
  Algorithms and laws.
\newblock In \emph{Proceedings of the 8th IEEE International Conference on Data
  Mining}, pages 608--617, 2008.
\newblock \doi{10.1109/ICDM.2008.72}.

\bibitem[van~der Hofstad(2016)]{vanderhofstad2016random}
Remco van~der Hofstad.
\newblock \emph{Random Graphs and Complex Networks, Volume 1}.
\newblock Number~43 in Cambridge Series in Statistical and Probabilistic
  Mathematics. Cambridge University Press, 2016.
\newblock \doi{10.1017/9781316779422}.

\bibitem[Wainwright and Jordan(2008)]{wainwright2008graphical}
Martin~J. Wainwright and Michael~I. Jordan.
\newblock Graphical models, exponential families, and variational inference.
\newblock \emph{Foundations and Trends in Machine Learning}, 1\penalty0
  (1--2):\penalty0 1--305, 2008.
\newblock \doi{10.1561/2200000001}.

\bibitem[Wang et~al.(2022)Wang, Zhang, Qin, Wang, Zhang, and
  Lin]{wang2022rlqvo}
Hanchen Wang, Ying Zhang, Lu~Qin, Wei Wang, Wenjie Zhang, and Xuemin Lin.
\newblock Reinforcement learning based query vertex ordering model for subgraph
  matching.
\newblock In \emph{Proceedings of the 38th IEEE International Conference on
  Data Engineering (ICDE)}, pages 245--258, 2022.
\newblock \doi{10.1109/ICDE53745.2022.00023}.

\bibitem[Watts and Strogatz(1998)]{watts1998smallworld}
Duncan~J. Watts and Steven~H. Strogatz.
\newblock Collective dynamics of `small-world' networks.
\newblock \emph{Nature}, 393\penalty0 (6684):\penalty0 440--442, 1998.
\newblock \doi{10.1038/30918}.

\bibitem[Wei et~al.(2016)Wei, Yu, Lu, and Lin]{wei2016gorder}
Hao Wei, Jeffrey~Xu Yu, Can Lu, and Xuemin Lin.
\newblock Speedup graph processing by graph ordering.
\newblock In \emph{Proceedings of the 2016 International Conference on
  Management of Data (SIGMOD)}, pages 1813--1828, 2016.
\newblock \doi{10.1145/2882903.2915220}.

\bibitem[Yu et~al.(2020)Yu, Qin, Zhang, Zhang, and Lin]{yu2020aot}
Michael Yu, Lu~Qin, Ying Zhang, Wenjie Zhang, and Xuemin Lin.
\newblock {AOT}: Pushing the efficiency boundary of main-memory triangle
  listing.
\newblock In \emph{Database Systems for Advanced Applications}, volume 12113 of
  \emph{Lecture Notes in Computer Science}, pages 516--533. Springer, 2020.
\newblock \doi{10.1007/978-3-030-59416-9_31}.

\bibitem[Zhao et~al.(2021)Zhao, Rong, Yu, Huang, Huang, and
  Zhang]{zhao2020graphordering}
Kangfei Zhao, Yu~Rong, Jeffrey~Xu Yu, Wenbing Huang, Junzhou Huang, and Hao
  Zhang.
\newblock Graph ordering: Towards the optimal by learning.
\newblock In \emph{Web Information Systems Engineering -- WISE 2021}, volume
  13080 of \emph{Lecture Notes in Computer Science}, pages 423--437. Springer,
  2021.
\newblock \doi{10.1007/978-3-030-90888-1_33}.

\end{thebibliography}

\clearpage
\appendix
\section*{Contents of the Appendices}
\begingroup\small
\begin{tabular}{@{}r@{\quad}l@{}}
\ref{app:proofs} & Proofs\\
\ref{app:complexity} & Complexity, cubic graphs and stability\\
\ref{app:mirror} & A monotone mirror step\\
\ref{app:mf} & Mean-field limit\\
\ref{app:kclique} & $k$-clique listing\\
\ref{app:protocol} & Methods and protocol, including Algorithm~\ref{alg:tils} and the budget accounting\\
\ref{app:prereg} & Pre-registered held-out evaluation\\
\ref{app:fixedtau} & Fixed-temperature landscape diagnostic\\
\ref{app:posthoc} & Secondary analyses on public graphs (not pre-registered)\\
\ref{app:neighsupp} & Held-out reference with Neigh (supplementary analysis)\\
\ref{app:budget} & Budget curve\\
\ref{app:entropy} & A second concave objective: minimum-entropy orientation\\
\ref{app:nullens} & Degree-preserving null ensembles\\
\ref{app:reuse} & Reusing orders along a null-model chain\\
\ref{app:extra} & Additional tables, including the 32 development graphs (Table~\ref{tab:dev})\\
\end{tabular}
\endgroup

\section{Proofs}\label{app:proofs}

\paragraph{Proposition~\ref{prop:exact}.}
Fix a triangle $\{a,b,c\}$ with $\pi(a)<\pi(b)<\pi(c)$. Then $b,c\in\Nout(a)$ and $c\in\Nout(b)$, so the intersection for edge $ab$ finds $c$. For edge $ac$ the third vertex $b$ is not in $\Nout(c)$, and for edge $bc$ the vertex $a$ is in neither forward set; hence each triangle is found exactly once. On edge $\{u,v\}$ the kernel scans the smaller forward set and tests each element for membership in the other, which costs $\min\{\dout(u),\dout(v)\}$ tests. Summing over edges gives Eq.~\eqref{eq:work}, and every output is a successful test.

\paragraph{Theorem~\ref{thm:rounding}.}
Write $x=\sum_{e=\{u,v\}}[p_ee_u+(1-p_e)e_v]$ with $p_e\in[0,1]$. Then $a^\top x=\sum_e[p_ea_u+(1-p_e)a_v]\ge\sum_e\min(a_u,a_v)$, so $\delta\ge0$. Under $\pi_a$ the earlier endpoint of each edge has the smaller coefficient (ties give equal coefficients), so $a^\top\dout_{\pi_a}=\sum_e\min(a_u,a_v)$. Concavity gives $F(y)\le F(x)+a^\top(y-x)$ for all $y$ in the domain. With $y=\dout_{\pi_a}$ this becomes $F(\dout_{\pi_a})\le F(x)-\delta(x,a)$. For $F_0$, a supergradient at $x$ is $a_v=\sum_{e\ni v}\omega_{e,v}$ with $\omega_{e,v}=1$ if $x_v<x_w$ for the other endpoint $w$, $0$ if $x_v>x_w$, and $\tfrac12$ at ties. For $F_\epsilon$ the gradient is $a_v=\sum_{e\ni v}\tfrac12(1-(x_v-x_w)/\sqrt{(x_v-x_w)^2+\epsilon^2})$. For $\Cpm$, $a_v=d(v)-2x_v$. Each costs one edge scan, and $\pi_a$ needs one sort.

\paragraph{Corollary~\ref{cor:free}.}
Order degree vectors are orientation out-degree vectors, which lie in $\ZG$. This gives $\min_{\ZG}F\le\min_DF\le\min_\pi F$. Conversely, Theorem~\ref{thm:rounding} maps every $x\in\ZG$ to an order with $F(\dout_{\pi_a})\le F(x)$. Apply it at a minimizer, which exists by compactness and continuity.

\paragraph{Proposition~\ref{prop:rules}.}
For $u\to v$ the third vertex of any triangle lies in $\Nout(u)\cap\Nout(v)$. Scanning $\Nout(v)$ costs $\dout(v)$ per incoming edge of $v$, hence $\sum_v d^-(v)\dout(v)=\sum_v\dout(v)(d(v)-\dout(v))$. Scanning $\Nout(u)$ costs $\dout(u)$ per outgoing edge of $u$, hence $\sum_u\dout(u)^2$. Taking the minimum per edge gives $\Lam$, which is concave as a sum of minima of linear functions. $\Cpm$ is a sum of concave quadratics, and $\Cpp$ is a sum of strictly convex quadratics. On $K_3$ the cyclic orientation has out-degrees $(1,1,1)$ and $\Cpp=3$, while every order gives $(2,1,0)$ and $\Cpp=5$. By \citet{borradaile2017egalitarian} (stated there for in-degrees) and \citet{harvey2006semimatchings}, an orientation minimizes $\sum_vf(\dout(v))$ for every strictly convex $f$ (the same $f$ at every vertex) if and only if no directed path leads from $u$ to $w$ with $\dout(u)\ge\dout(w)+2$. Path reversal finds such an orientation in polynomial time. Its $\Cpp$ value lower-bounds $\min_\pi\Cpp$ because orders are orientations.

\paragraph{Proposition~\ref{prop:bounds}.}
(a) The triangle bound is Proposition~\ref{prop:exact}. Let $S$ be the vertices with $\dout=0$. Two vertices of $S$ are never adjacent, because the earlier endpoint of an edge has positive forward degree. Every edge inside $C=V\setminus S$ has both forward degrees at least one and costs at least one. Deleting $E(C)$ leaves the bipartite graph $(S,C)$, so $\Lam\ge|E(C)|\ge\beta(G)$. If $G$ is bipartite, placing one side first gives every edge an endpoint with $\dout=0$ and zero work. Conversely, zero work forces $\beta(G)=0$.
(b) For integers $a,b\ge0$, $\min(a,b)=\sum_{k\ge1}\mathbf 1[a\ge k]\mathbf 1[b\ge k]$; summing over edges gives the layer identity. In the program~\eqref{eq:mip}, monotone binaries with sum $x_v$ force $b_{v,k}=\mathbf 1[x_v\ge k]$, and the optimal $w_{e,k}=\max(0,b_{u,k}+b_{v,k}-1)=\mathbf 1[x_u\ge k]\mathbf 1[x_v\ge k]$. The optimum is therefore $\min F_0$ over integral points of $\ZG$. This set contains every order degree vector, and $F_0\ge\Lopt{G}$ on all of $\ZG$ by Corollary~\ref{cor:free}; hence the optimum equals $\Lopt{G}$.
Valid cuts, which hold at every integral feasible point: $\sum_{e,k}w_{e,k}\ge|\mathcal T(G)|$; $b_{u,1}+b_{v,1}\ge1$ for every edge (as $x_u+x_v\ge1$); $\sum_{e\in C}w_{e,1}\ge1$ for every odd cycle $C$ (the vertices with $x=0$ are independent, so an odd cycle has an edge with both endpoints of degree at least one); and $\sum_ew_{e,1}\ge\beta_{\rm LB}$ for any certified lower bound on $\beta(G)$. The runs of Section~\ref{sec:exact} use all of them: the cycle cuts for every triangle and for each cycle of a greedy packing of edge-disjoint odd cycles (triangles first, then shortest odd cycles), and $\beta_{\rm LB}$ equal to the number of packed cycles, which is a lower bound on $\beta(G)$ because deleting edges to make $G$ bipartite removes an edge from each packed cycle and the packed cycles share no edges.
For maximum degree at most two, the minimum work equals the number of odd-cycle components: paths and even cycles are bipartite, and an odd cycle $(0,1,\dots,2k)$ with $0,2,\dots,2k$ first and the odd vertices last has work one.

\paragraph{Smooth relaxation: complexity, consistency and finite-temperature error.}
One edge scan computes all $p_{uv}$ and accumulates $x$; a second computes $\smin_\epsilon$ and the degree adjoints $a=\nabla F_\epsilon(x)$; a third applies $\partial J/\partial s_v=\tau^{-1}\sum_{w\in N(v)}p_{vw}(1-p_{vw})(a_w-a_v)$. This takes $O(m+n)$ time and memory. If scores are distinct on every edge, $p_{uv}\to\mathbf 1[s_u<s_v]$ as $\tau\to0$, and $\smin_\epsilon\to\min$ with $0\le\smin_\epsilon-\min\le\epsilon/2$. Let $\gamma=\min_{uv\in E}|s_u-s_v|>0$ and $\eta=(1+e^{\gamma/\tau})^{-1}$. Each soft direction is within $\eta$ of its hard value, so $|x_v-\dout_{\pi_s}(v)|\le d(v)\eta$. Since $\min$ is 1-Lipschitz in $\ell_\infty$,
$|J(s)-\Lam(\pi_s)|\le\eta\sum_{uv}\max\{d(u),d(v)\}+m\epsilon/2\le\eta\sum_vd(v)^2+m\epsilon/2$.

\paragraph{Theorem~\ref{thm:landscape}.}
Put $r=-s/\tau$ and $A(r)=\sum_{uv\in E}\log(e^{r_u}+e^{r_v})$. Then $x(s)=\nabla A(r)$ is the mean of the product distribution that orients each edge independently, and $\nabla^2A=L$ is a Laplacian with positive weights $p_{uv}(1-p_{uv})$. Let $H$ be the vectors with zero sum on every connected component.
(i) $L$ is positive definite on $H$, so $A$ is strictly convex there and $\nabla A$ is injective on $H$. Every finite $r$ gives positive probability to every orientation, so $x\in\relint\ZG$. For surjectivity, fix $x\in\relint\ZG$. A ball of radius $c>0$ in the affine hull around $x$ lies in $\ZG$, and $A(r)\ge\sum_{uv}\max(r_u,r_v)=h_{\ZG}(r)$, so $A(r)-r^\top x\ge c\|r\|$ on $H$. The minimizer exists and satisfies $P_H(\nabla A(r)-x)=0$. Both vectors have the same component sums, the edge counts, so $\nabla A(r)=x$.
(ii) By the chain rule $\nabla_sJ=-\tau^{-1}La$ with $a=\nabla F(x)$. If this vanishes, $a$ is constant on components. Every $y-x$ with $y\in\ZG$ lies in $H$, so $F(y)\le F(x)+a^\top(y-x)=F(x)$.
(iii) $x-\dout_{\pi_a}\in H$, so $\delta=a^\top(x-\dout_{\pi_a})=(P_Ha)^\top(x-\dout_{\pi_a})\le\|P_Ha\|_2\|x-\dout_{\pi_a}\|_2$. Moreover $\|P_Ha\|_2\le\|La\|_2/\lambda_*=\tau\|\nabla J\|_2/\lambda_*$.
For $F_\epsilon$ with $\epsilon>0$, $F_\epsilon$ is strictly concave on the affine hull of $\ZG$, so a finite stationary score, if one exists, is unique modulo adding a constant on each connected component and is a strict maximum modulo these shifts, which leave $J$ unchanged.

\paragraph{Counterexamples.}
On $K_{3,3}$ with parts $\{0,1,2\},\{3,4,5\}$, the order $(0,1,3,4,5,2)$ has forward degrees $(3,3,0,1,1,1)$ and work 6. Its doubled $F_0$ supergradient is $(0,0,6,4,4,4)$; sorting it returns the same order with $\delta=0$, although the optimum is 0. On the five-cycle $01,12,23,34,40$, the order $(0,1,4,2,3)$ has work 3 and every adjacent swap has zero change, whereas the optimum is 1. Neither a zero tangent gap nor adjacent-swap optimality certifies a global optimum.

\section{Complexity, Cubic Graphs and Stability}\label{app:complexity}

\paragraph{Theorem~\ref{thm:hard}.}
Let $G=S_2(H)$, where the path replacing $uv\in E(H)$ is $u\,a\,b\,v$ with $a$ adjacent to $u$.

\emph{Orders to cuts ($\beta(H)\le\min\Lam$).} Fix an order $\pi$ of $G$, let $S$ be the set of vertices with $\dout_\pi=0$ and $C=V(G)\setminus S$. As in Proposition~\ref{prop:bounds}a, $S$ is independent and $\Lam(\pi)\ge|E(C)|$, and $G-E(C)$ is bipartite with sides $S$ and $C$. Let $F\subseteq E(H)$ be the set of edges of $H$ whose path contains an edge of $E(C)$, so $|F|\le|E(C)|$. A cycle of length $L$ in $H-F$ maps to a cycle of length $3L$ in $G-E(C)$, so $3L$, and hence $L$, is even; thus $H-F$ is bipartite. A 2-coloring of $H-F$, computed by breadth-first search, is a cut of $H$ whose uncut edges all lie in $F$. Hence the cut leaves at most $|F|\le|E(C)|\le\Lam(\pi)$ edges uncut, in linear time.

\emph{Cuts to orders ($\min\Lam\le\beta(H)$).} Fix a cut $(A,B)$ of $H$ with $k$ uncut edges; color $A$ with 1 and $B$ with 0. For each path $u\,a\,b\,v$: if $uv$ is cut, set $c(a)=1-c(u)$ and $c(b)=1-c(v)$ (no monochromatic edge); if $c(u)=c(v)=0$, set $c(a)=c(b)=1$ (only $ab$ is monochromatic); if $c(u)=c(v)=1$, set $c(a)=1$, $c(b)=0$ (only $ua$ is monochromatic). Then the class $Y$ of color 0 is independent, and there are exactly $k$ monochromatic edges, all inside the class $X$ of color 1. Order the original vertices of $X$, then the subdivision vertices of $X$, then $Y$. Every vertex of $Y$ is a sink, so every edge meeting $Y$ costs 0. In each monochromatic edge the later endpoint is a subdivision vertex $s$ whose other neighbor lies in $Y$ (for $ab$ this holds for both endpoints; for $ua$, $s=a$ comes after the original vertex $u$), so $\dout(s)=1$; the earlier endpoint has $s$ after it, so the edge costs exactly one. Hence $\Lam=k$, and an optimal cut gives $\Lam=\beta(H)$. (For $ua$ the order matters: if $a$ preceded $u$, the edge could cost two.)

\emph{Consequences.} $S_2(H)$ has maximum degree $\max(\Delta(H),2)$ and girth $3\cdot\mathrm{girth}(H)\ge9$, so it is triangle-free and all of its work is excess work beyond the output. Max-Cut is APX-complete on cubic graphs \citep{alimonti2000cubic}, so $\min\Lam$ is NP-hard on triangle-free graphs of maximum degree three. The cut-to-order map preserves the objective exactly, and the order-to-cut map does not increase it; the latter can be strict (for $H=K_3$, $S_2(H)$ is the 9-cycle, which has orders with $\Lam=7$, whereas every cut of $K_3$ leaves at most three edges uncut). Hence the optima are equal, and an order with $\Lam\le c\cdot\Lopt{S_2(H)}$ yields a cut with at most $c\,\beta(H)$ uncut edges. Since $\beta(H)\le|E(H)|/2\le\mathrm{MaxCut}(H)$, a $(1+\varepsilon)$-approximation for $\beta$ yields a cut of size at least $|E|-(1+\varepsilon)\beta\ge(1-\varepsilon)\mathrm{MaxCut}(H)$; a polynomial-time approximation scheme for $\min\Lam$ on subcubic triangle-free graphs would therefore give one for Max-Cut on cubic graphs, so $\min\Lam$ is APX-hard. Under the Unique Games Conjecture it is hard, for every small $\varepsilon>0$, to distinguish $\mathrm{MaxCut}\ge(1-\varepsilon)|E|$ from $\mathrm{MaxCut}\le(1-\tfrac2\pi\sqrt\varepsilon)|E|$ \citep{khot2007optimal}, i.e.\ $\beta\le\varepsilon|E|$ from $\beta\ge\tfrac2\pi\sqrt\varepsilon|E|$; the ratio is unbounded as $\varepsilon\to0$, so Min UnCut, and with it $\min\Lam$, has no constant-factor approximation \citep{khot2002unique}. These statements are worst-case. Because the correspondence is additive ($\Lopt{S_2(H)}=|E(H)|-\mathrm{MaxCut}(H)$), approximation ratios for Max-Cut do not transfer, and whether $\Lopt{G}$ admits, for example, an $O(\sqrt{\log n})$ approximation as Min UnCut does \citep{agarwal2005uncut} is open. For the non-adaptive cost $\Cpm$, \citet[\S 5.3.3]{lecuyer2023thesis} discusses a Min UnCut connection, attributed there to Makarychev, with an $O(\sqrt{\log n})$ approximation and the same UGC consequence; Theorem~\ref{thm:hard} is the corresponding result for $\Lam$, obtained through the exact identity on 2-subdivisions.

\emph{Checks.} On 38 graphs $H$ (cycles $C_3$--$C_9$, $K_4$, $K_5$, $K_4-e$, paw, bowtie, $2K_3$, the wheels on 5 and 6 vertices, $K_{3,3}$, $K_{2,3}$, prism, octahedron, Petersen, Möbius--Kantor, three random cubic graphs and 14 random graphs with at most 7 vertices and 10 edges), the forward map gives $\Lam=\#\text{uncut}$ for every cut (all $2^{|V(H)|}$ cuts, or 500 random cuts when $|V(H)|>10$), and the back map satisfies $\#\text{uncut}\le|F|\le|E(C)|\le\Lam$ on 300 random orders per graph. On the 36 graphs where an exact method finished, $\min\Lam_{S_2(H)}=\beta(H)$ (values 0--4) by up to three independent methods: enumeration of all orders, enumeration of all orientations (Corollary~\ref{cor:free}), and an order-space integer program with transitivity constraints, for which the solver proves the optimum 3 for the Petersen graph ($|V(S_2(H))|=40$) in 205\,s.

\paragraph{Hardness of the relaxation.}
\begin{corollary}[Relaxation]\label{cor:relaxhard}
Fix $\tau>0$ and $0<\epsilon\le1/m$, and let $\mathrm{OPT}=\min_\pi\Lam(\pi)$. Then $\mathrm{OPT}\le\inf_sJ(s)\le\mathrm{OPT}+\tfrac12$, and one tangent step at $x(s)$ maps any scores $s$ to an order $\pi$ with $\Lam(\pi)\le J(s)$. Hence finding scores with $J(s)<\inf J+\tfrac12$ is NP-hard, and under the Unique Games Conjecture no polynomial-time algorithm achieves $J(s)\le c\inf J$ for a constant $c$ on all non-bipartite graphs.
\end{corollary}
\begin{proof}
By Theorem~\ref{thm:landscape}i, $x(s)$ ranges over $\relint\ZG$, and $F_\epsilon$ is continuous, so $\inf J=\min_{\ZG}F_\epsilon$. Write $\smin_\epsilon(a,b)=\min(a,b)+\tfrac12\big(\epsilon-\sqrt{(a-b)^2+\epsilon^2}+|a-b|\big)$; the bracket lies in $[0,\epsilon]$, so $F_0\le F_\epsilon\le F_0+m\epsilon/2$. By Corollary~\ref{cor:free}, $\min_{\ZG}F_0=\mathrm{OPT}$, which gives the first claim. Theorem~\ref{thm:rounding} for $F_0$ at $x=x(s)$ gives $\Lam(\pi_a)=F_0(\dout_{\pi_a})\le F_0(x(s))\le J(s)$. If $J(s)<\inf J+\tfrac12\le\mathrm{OPT}+1$, then $\Lam(\pi_a)=\mathrm{OPT}$, and Theorem~\ref{thm:hard} applies. If $J(s)\le c\inf J$ on a non-bipartite graph ($\mathrm{OPT}\ge1$; bipartiteness is tested in linear time), then $\Lam(\pi_a)\le c(\mathrm{OPT}+\tfrac12)\le\tfrac{3c}2\mathrm{OPT}$, contradicting Theorem~\ref{thm:hard} under the Unique Games Conjecture.
\end{proof}
The reduction evaluates $x(s)$ and compares soft degrees exactly (real-RAM model of the relaxation). A statement in the Turing model would also need a representation of the scores and a precision model for the logistic evaluations and the comparisons in the rounding step, so the corollary is the idealized, exact-arithmetic consequence of Theorem~\ref{thm:hard} for score optimization; the discrete NP-hardness of Theorem~\ref{thm:hard} does not depend on it. The corollary says that the absence of finite local minima (Theorem~\ref{thm:landscape}) does not make the relaxation tractable: the difficulty sits at infinity, where Theorem~\ref{thm:landscape}iii shows that gradients certify nothing.

\paragraph{Proposition~\ref{prop:cubic}.}
\emph{Random order.} Let $uv$ be an edge in no triangle. The two other neighbors of $u$ and the two other neighbors of $v$ are four distinct vertices, distinct from $u$ and $v$, so $(\dout(u),\dout(v))$ is a function of the relative order of six distinct vertices, which is uniform over $6!=720$ patterns. Enumerating them gives $\mathbb E\min\{\dout(u),\dout(v)\}=59/90$, and $\tfrac{59}{90}\cdot\tfrac{3n}2=\tfrac{59}{60}n$. At most $3t$ edges lie in triangles, and for each of them both the true and the assumed contribution lie in $[0,3]$, so the error is at most $9t$.
\emph{Upper bound.} Let $I$ be a maximum independent set and $C=V\setminus I$. Since $I$ is maximal, every vertex of $C$ has a neighbor in $I$, so $G[C]$ has maximum degree two and is a union of paths and cycles. Orient every $C$--$I$ edge towards $I$, and orient each path and cycle of $G[C]$ alternately, so that every vertex is a source or a sink within $C$; an odd cycle $c_0\cdots c_{2k}$, with even-indexed vertices as sources, keeps one defect edge $c_0\to c_{2k}$. The orientation is acyclic, so it is induced by an order (any topological order). Vertices of $I$ are sinks, so edges at $I$ cost 0. An edge of $G[C]$ whose endpoint $w$ is a sink within $C$ costs at most $\dout(w)=3-\deg_C(w)\le2$; the defect edge costs at most $\dout(c_{2k})=2$ (one edge to $c_{2k-1}$, one to $I$). Hence $\Lam\le2|E(C)|=2(\tfrac{3n}2-3|I|)=3n-6\alpha(G)$.
\emph{Consequences.} Triangle-free subcubic graphs have $\alpha\ge5n/14$ \citep{staton1979independence,heckman2001independence}, so $\Lopt{G}\le3n-\tfrac{30}{14}n=\tfrac67n$. The uniformly random cubic graph $G_n$ has a law invariant under relabeling, so for an order $\pi$ chosen independently of the edges, $\mathbb E\Lam_{G_n}(\pi)=\mathbb E_{G}\mathbb E_{\sigma\text{ uniform}}\Lam_G(\sigma)=\tfrac{59}{60}n+O(\mathbb Et)$, and $\mathbb Et\to4/3$ \citep{bollobas1980configuration}. Asymptotically almost surely $\alpha(G_n)\ge(6\ln\tfrac32-2-\varepsilon)n$ \citep{frieze1994independence}, so $\Lopt{G_n}\le(15-36\ln\tfrac32+o(1))n\approx0.4033n$ a.a.s. Since $\Lopt{G_n}\le3m=\tfrac92n$ always, also $\mathbb E\Lopt{G_n}\le(15-36\ln\tfrac32+o(1))n$, and therefore $\mathbb E\Lam_{G_n}(\pi)/\mathbb E\Lopt{G_n}\ge\tfrac{59/60}{15-36\ln(3/2)}-o(1)=2.438\ldots-o(1)$. For degree $d\ge4$ the upper-bound argument fails ($G[C]$ can have degree $d-1$); Table~\ref{tab:cubic} reports 4-regular graphs numerically only (random-order limit $\tfrac{137}{70}n$, from $\mathbb E\min=137/140$ per edge by the same enumeration over $8!$ patterns).

\emph{Checks.} Exact enumeration gives $59/90$ per edge for $d=3$ ($137/140$ for $d=4$); the mean of $\Lam$ over all $8!$ orders of the cube $Q_3$ is exactly $118/15=\tfrac{59}{60}\cdot8$. On 32 random cubic graphs with $n=8$--$16$, exact $\mathrm{OPT}$ (orientation enumeration) $\le\Lam$ of the construction $\le3n-6\alpha$ in all cases. Table~\ref{tab:cubic} gives large random cubic graphs.

\paragraph{Edge edits.}
Let $G$ and $H$ be graphs on the same vertex set, $U=G\cup H$, $k_+=|E(H)\setminus E(G)|$, $k_-=|E(G)\setminus E(H)|$, $\Delta=\Delta(U)$, and fix an order $\pi$. For $G\subseteq U$ let
\[
B_\pi(G\to U)=\sum_{f=uv\in E(U)\setminus E(G),\ u\text{ earlier}}\Big[d_G(u)+\min\{\dout_{U}(u),\dout_{U}(v)\}\Big].
\]
\begin{lemma}[Stability]\label{lem:stab}
(a) If $G\subseteq U$, then $\Lam_G(\pi)\le\Lam_U(\pi)\le\Lam_G(\pi)+B_\pi(G\to U)$ and $B_\pi(G\to U)\le(2\Delta-1)|E(U)\setminus E(G)|$.
(b) $|\Lam_G(\pi)-\Lam_H(\pi)|\le\max\{B_\pi(G\to U),B_\pi(H\to U)\}\le(2\Delta-1)\max(k_+,k_-)$.
(c) $|\min\Lam_G-\min\Lam_H|\le(2\Delta-1)\max(k_+,k_-)$, and the optimality gap satisfies $\mathrm{gap}_H(\pi)\le\mathrm{gap}_G(\pi)+B_\pi(G\to U)+(2\Delta-1)k_-\le\mathrm{gap}_G(\pi)+(2\Delta-1)(k_++k_-)$; tangent refinement on $H$ does not increase $\mathrm{gap}_H$.
\end{lemma}
\begin{proof}
Let $x=\dout_G$ and $x'=\dout_U$ under $\pi$; then $x'-x=r\ge0$, where $r_v$ counts the edges of $U\setminus G$ whose earlier endpoint is $v$. Write $\Lam_U-\Lam_G=\sum_{f\in U\setminus G}\min_U(f)+\sum_{e=vw\in G}[\min(x_v+r_v,x_w+r_w)-\min(x_v,x_w)]$. Each bracket lies in $[0,\max(r_v,r_w)]\subseteq[0,r_v+r_w]$, which gives monotonicity, and summing over $e\in G$ gives at most $\sum_vr_vd_G(v)=\sum_fd_G(u_f)$. With $\min_U(f)\le\dout_U(u_f)\le\Delta$ and $d_G(u_f)\le\Delta-1$ this gives (a). Both $\Lam_G$ and $\Lam_H$ lie in $[\Lam_U-\max B,\Lam_U]$, which gives (b). For (c), let $\sigma$ minimize $\Lam_H$; then $\min\Lam_G\le\Lam_G(\sigma)\le\Lam_U(\sigma)\le\Lam_H(\sigma)+(2\Delta-1)k_-$, and symmetrically; and $\Lam_H(\pi)\le\Lam_U(\pi)\le\Lam_G(\pi)+B_\pi(G\to U)$. Theorem~\ref{thm:rounding} gives the last claim.
\end{proof}
The bound in (b) sharpens $2\Delta(k_++k_-)$. $\Lam_H(\pi)$ itself is computable exactly in $O(m)$ time, so the lemma's use is on the optimum and gap side: an order that is near-optimal for $G$ stays near-optimal after few edits, with a cheap localized certificate. \emph{Checks.} No violation of (a)--(b) or of $2\Delta k$ over all 98,304 pairs of labeled 4-vertex graphs with all 24 orders, all 6.76M combinations of 5-vertex graph pairs differing in at most two edges with all 120 orders, and 2,000 random and hub-heavy graph pairs ($n=20$--$200$); no violation of (c) on 3,000 pairs of 6-vertex graphs with exact optima. The worst observed ratio of actual change to bound is $1/2$ for $2\Delta k$, $2/3$ for $(2\Delta-1)\max(k_+,k_-)$, and $1$ (attained) for the localized bound. Along the swap chains of Appendix~\ref{app:reuse}, none of the bounds is violated (\RStabViol\ violations in \RStabChecks\ frame transitions), but they are loose: the median ratio of bound to actual change is \RStabLocSlack\ for the localized bound and \RStabUniSlack\ for $(2\Delta-1)\max(k_+,k_-)$, so they are worst-case guarantees, not predictions.

\begin{table}[H]
\centering\small\setlength{\tabcolsep}{3pt}
\caption{Work per vertex $\Lam/n$ on random regular graphs (NetworkX \citep{hagberg2008networkx}, three graph seeds each; mean$\pm$SD). The construction row applies the proof of Proposition~\ref{prop:cubic} to a greedy, not a maximum, independent set $I$. Random-order limits (Proposition~\ref{prop:cubic}): $59/60\approx0.983$ for $d=3$ and $137/70\approx1.957$ for $d=4$; for $d=3$ the optimum is at most $0.4033n$ a.a.s. \TILS uses 100\,$\mu$s per edge (seed 1). For $d=4$ only numerics are available.}
\label{tab:cubic}
\begin{tabular}{@{}lrrr@{}}
\toprule
 & \multicolumn{2}{c}{cubic ($d=3$)} & $d=4$\\
\cmidrule(lr){2-3}\cmidrule(lr){4-4}
Order & $n=10^4$ & $n=10^5$ & $n=10^5$\\
\midrule
Uniformly random order (5 per graph) & 0.986{\tiny$\pm$0.005} & 0.984{\tiny$\pm$0.001} & 1.959{\tiny$\pm$0.002}\\
Greedy independent-set construction & 0.313{\tiny$\pm$0.005} & 0.312{\tiny$\pm$0.000} & --\\
\quad its bound $3-6|I|/n$ & 0.405{\tiny$\pm$0.005} & 0.402{\tiny$\pm$0.001} & --\\
T from the random order & 0.633{\tiny$\pm$0.010} & 0.635{\tiny$\pm$0.001} & 1.333{\tiny$\pm$0.002}\\
\TILS from the construction ($d=3$) / random order ($d=4$) & 0.275{\tiny$\pm$0.000} & 0.281{\tiny$\pm$0.000} & 0.798{\tiny$\pm$0.005}\\
\TILS from the random order & 0.281{\tiny$\pm$0.007} & 0.288{\tiny$\pm$0.001} & --\\
Limit / bound & 0.983 / 0.403 & 0.983 / 0.403 & 1.957 / --\\
\bottomrule
\end{tabular}
\end{table}

\section{A Monotone Mirror Step}\label{app:mirror}
With $r=-s/\tau$ and $A$ as in Appendix~\ref{app:proofs}, $x=\nabla A(r)$ and $a=\nabla F(x)$ for a differentiable concave $F$, set $r'=r-\alpha a$ for $\alpha>0$ and $x'=\nabla A(r')$. For an edge $e=(u,v)$ let $p_e=\sigma(r_u-r_v)$ and $p'_e=\sigma(r'_u-r'_v)$.
\begin{proposition}
$F(x')\le F(x)-\alpha^{-1}\sum_e[\mathrm{KL}(p'_e\|p_e)+\mathrm{KL}(p_e\|p'_e)]$, with strict decrease unless $a$ is constant on every component, in which case $x$ maximizes $F$. If $\alpha>\max_{a_u\ne a_v}|r_u-r_v|/|a_u-a_v|$, sorting $r-\alpha a$ in descending order yields an order with $F\le F(x)-\delta(x,a)$.
\end{proposition}
\begin{proof}
Summing the Bernoulli identity $\mathrm{KL}(p'\|p)+\mathrm{KL}(p\|p')=(p'-p)(\operatorname{logit}p'-\operatorname{logit}p)$ over the edges gives $\sum_e[\mathrm{KL}(p'_e\|p_e)+\mathrm{KL}(p_e\|p'_e)]=(r'-r)^\top(x'-x)=-\alpha a^\top(x'-x)$, and the supergradient inequality $F(x')\le F(x)+a^\top(x'-x)$ gives the first claim. For the second, an edge with $a_u<a_v$ has new logit gap $(r_u-r_v)+\alpha(a_v-a_u)>0$, so $u$ precedes $v$, and the edgewise linear minimum is attained. Concavity then applies as in Theorem~\ref{thm:rounding}.
\end{proof}
This is entropic mirror descent \citep{beck2003mirror} specialized to the edge marginals. In tuned, time-matched comparisons on 19 graphs, measured against a reference without reversed orders, standalone mirror steps trailed tuned Adam and insertion search in all three graph groups, and mirror plus insertion was within 0.4 points of tangent plus insertion; we therefore do not use it in the main experiments.

\section{Mean-Field Limit}\label{app:mf}
\paragraph{Proof of Theorem~\ref{thm:mf}.}
If $\sum_iD_i$ is odd, the degree of vertex $n$ is increased by one, as is standard \citep{vanderhofstad2016random}; this changes one degree by one, the bounds below then hold with $K+1$ in place of $K$, and the limit is unaffected. $\Lam$ is evaluated on the simple graph obtained by deleting loops and multi-edges. Let $K$ bound the degrees. Fix the layout, and let block membership and within-block ranks depend only on degrees, labels and independent randomness, not on the pairing (vertex-level Split with identifier ties is of this form and converges to two half-blocks per class). (1) In $\mathrm{CM}_n(p)$ the expected numbers of self-loops and multi-edges are $O(1)$. Removing one edge changes at most two forward degrees by one and $\Lam$ by $O(K)$, which is negligible after division by $n$. (2) Pick a uniformly random edge $\{u,w\}$. With probability $1-O(K^2/n)$ the neighborhoods of $u$ and $w$ at distance one are disjoint and contain no further edges among themselves. Their other $2(K-1)$ half-edges attach to distinct vertices whose degrees are, up to total variation $O(K^2/n)$, i.i.d.\ with law $q$, and $\deg u,\deg w$ are themselves i.i.d.\ $q$ \citep{vanderhofstad2016random}. (3) Given the degree sequence, the ranks inside each block form a uniformly random permutation (for label-based ranks this holds in distribution because configuration-model labels are exchangeable), and the block boundaries converge because the class frequencies converge to $p$. Hence the normalized positions of a bounded number of distinct vertices converge jointly to independent draws from their block laws, and every neighbor of $u$ other than $w$ is later than $u$ independently with limiting probability $h(t_u)$. (4) Conditioned on $t_u<t_w$, $\dout(u)=1+\mathrm{Bin}(k_u-1,h(t_u))$ and $\dout(w)=\mathrm{Bin}(k_w-1,h(t_w))$ in the limit, independently. Symmetrizing over which endpoint is earlier gives $\mathbb E\min=2\,\mathbb E[\mathbf 1\{t_u<t_w\}\min(\cdot,\cdot)]$. (5) Given the degree sequence $\mathbf D$, the number of edges $m=\tfrac12\sum_iD_i$ is fixed, and $\mathbb E[\Lam\mid\mathbf D]=m\,\mathbb E[\min\{\dout(u),\dout(w)\}\mid\mathbf D]$ for a uniformly random edge $\{u,w\}$, up to the correction of step (1). Steps (2)--(4) hold conditionally on $\mathbf D$ whenever its empirical degree law converges to $p$, which happens almost surely, and the minimum is bounded by $K$, so the conditional expectation of the minimum converges to the expectation under the limit law. Since $m/n\to\bar d/2$ almost surely and $(m/n)\min\le K^2/2$, bounded convergence applied to $\mathbb E\Lam/n=\mathbb E\big[(m/n)\,\mathbb E[\min\mid\mathbf D]\big]$ gives the limit. Finally, $\mathbb E\min(A,B)=\sum_{j\ge1}\Pr(A\ge j)\Pr(B\ge j)$ for independent nonnegative integers. For unbounded degrees with $\mathbb E D^2<\infty$ we expect the same limit via uniform integrability, since $\mathbb E\min\le\mathbb E D_{\rm size\text{-}biased}$, but we do not claim it.

\paragraph{Computation and search.}
Table~\ref{tab:mf} reports $\Lam/n$ per non-isolated vertex. Positions inside each block are integrated with $N_q$ midpoint nodes, and suffix sums give $O(PK)$ time for $P$ nodes in total and at most $K$ binomial terms. Values at $N_q=8$ differ from those at $N_q=128$ by at most 0.05\%. Block layouts are searched by first-improvement insertion over class sequences, with random and structured restarts. Splitting each class into two blocks changes the best value found by at most 0.3\%. The best layouts found are V-shaped; for PL3.0 roughly every other class lies on the descending front arm and the rest on the ascending back arm.

\paragraph{Validation on sampled graphs.}
At $n=10^5$ (two graph seeds), sampled $\Lam/n$ of ascending, descending and optimized degree orders lies within 0.1--0.7\% of the limit for ER8 and PL2.1/2.5/3.0. Split lies within 0.2--0.8\% of its predicted value (ER8 8.791 vs.\ 8.807 for Poisson(8); PL2.5 3.021 vs.\ 3.039). An independent reimplementation reproduced the limit values to four decimals and the sampled comparisons to within 0.85\%. Individual heavy-tailed graphs fluctuate by 2--3\%, so Table~\ref{tab:scaling} compares methods on the same graphs. Coloring orders are greedy colorings with vertices processed in reverse smallest-last order (SL) or by descending degree (DEG), with color classes as blocks in ascending or descending color order (four variants; Table~\ref{tab:scaling-colour}). At $n=10^5$ the best of the four is 4.8--42.5\% below the best-found degree-based order (PL2.1 lowest); SL-desc is best on ER8, PL2.5 and PL3.0, and SL-asc on PL2.1. Both descending variants are worse than that order on PL2.1.

\section{\texorpdfstring{$k$}{k}-Clique Listing}\label{app:kclique}
In kClist-style listing \citep{danisch2018kcliques}, a partial clique $Q$ carries the candidate set $C(Q)$ of common later neighbors. Extending $Q$ by $v\in C(Q)$ computes $C(Q)\cap\Nout(v)$. With adaptive intersection, the exact work is
\[
W_k(\pi)=\sum_{l=1}^{k-2}\ \sum_{(l+1)\text{-cliques }Q+v,\ v\text{ last}}\min\{|C(Q)|,\dout(v)\}.
\]
The level $l=1$ is exactly $\Lam$, so $W_3=\Lam$; the two one-sided rules give $\Cpm$ and $\Cpp$ at $k=3$, as in Proposition~\ref{prop:rules}. Because the sinks of an acyclic orientation are determined by its out-degrees, peeling shows that an acyclic orientation is determined by its out-degree vector. Every order functional is therefore formally a function of $\dout$, but $W_k$ has no local form in degrees for $k\ge4$. In two graphs whose triangle $(0,1,2)$ has forward degrees $(3,2,2)$ under the identity order, the level-2 term is 2 in one and 1 in the other. Hence the $O(m)$ tangent step of Theorem~\ref{thm:rounding} does not transfer to $W_k$.

\paragraph{End-to-end benchmark.}
All orders were computed in C++ with full cost charged, five interleaved repetitions, $k=5$--$8$, on the four public development graphs and a random geometric graph ($n=20{,}000$, radius 0.025). Clique counts matched NetworkX \citep{hagberg2008networkx} on 2{,}400 checks. Findings:
\begin{itemize}
\item For a kClist-style lister, smallest-last is fastest per call and end-to-end on every expensive cell; tangent and \TILS orders list 7--50\% slower.
\item For the adaptive lister, tangent refinement from smallest-last lists 22--27\% faster than smallest-last on Loc-Brightkite ($k=6$--$8$), at 0.09\,s extra ordering time, i.e., break-even within one call; an independent reimplementation measured 20--25\%.
\item The adaptive lister remains slower than the kClist-style lister with smallest-last.
\end{itemize}

\section{Methods and Protocol}\label{app:protocol}
\paragraph{Reference orders.} Natural is the input order after dense remapping; Random uses seed 7; Degree is ascending by degree with identifier ties; core ordering sorts by core number; smallest-last repeatedly removes a vertex of minimum remaining degree \citep{matula1983smallest}. These last two differ, and smallest-last is the stronger baseline for $\Cpp$ and $k$-clique listing. Split and Check follow \citet{lecuyer2023tailored}. Neigh uses the published VOLT implementation with threshold 0.01 and seeds 1--5. Every candidate is evaluated in both directions, and $R^\star$ is the minimum exact $\Lam$. The reduction of an order is $100(\Lam_{R^\star}-\Lam)/\Lam_{R^\star}$, and $0$ (a tie) when $\Lam_{R^\star}=0$; every certified optimum in Table~\ref{tab:exact} is positive.

\paragraph{\DOVE.} Scores start at $s_v=2i_v/(n-1)-1$, where $i_v\in\{0,\dots,n-1\}$ is the rank of $v$ in the reference order. Adam (learning rate $0.08$, moments $(0.9,0.999)$) minimizes $J/m$ for 200 steps with $\tau_t=0.5(0.04/0.5)^{t/200}$ and $\epsilon_t=\max(0.01,\tau_t^2)$. Scores are standardized after each step, and the exact work of the sorted scores is checked every five steps to keep the best incumbent.

\paragraph{Tangent refinement and \TILS.} T iterates Theorem~\ref{thm:rounding} with the $F_0$ supergradient and identifier tie-breaking until work stops strictly decreasing, for at most 100 rounds. ILS proposes a uniformly random edge, one of its endpoints $u$, and a side (before or after the other endpoint $v$), and moves $u$ to that position. Only edges between $u$ and vertices it passes change direction, and the work change is computed exactly from edges incident to vertices whose forward degree changes. Proposals come in chunks of 256. Moves with negative delta are accepted, zero-delta moves with probability $\rho$, and after $S$ consecutive chunks without improvement the state resets to the incumbent and applies $\kappa$ random moves. The grid $\rho\in\{0,0.3,1\}$, $\kappa\in\{0,10\}$ with $S=50$ was evaluated on six synthetic development graphs, giving mean gains of 9.8, 15.2, 10.8, 16.6, 10.9 and 17.0\% for $(\rho,\kappa)=(0,0),(0,10),(0.3,0),(0.3,10),(1,0),(1,10)$; $(1,10)$ was selected. \TILS spends at most 10\% of its budget on the initial T phase and ends with a charged T pass.

\begin{algorithm}[H]
\caption{Tangent refinement with iterated local search (\TILS)}
\label{alg:tils}
\small
\begin{algorithmic}[1]
\Require graph $G$, reference order $\pi_0$, time budget $B$, parameters $(\rho,\kappa,S)$
\State $\pi^\star\gets\pi_0$ \Comment{exact incumbent, never worse than $\pi_0$}
\Repeat\ $x\gets\dout_{\pi^\star}$; $a\gets\partial F_0(x)$; $\pi\gets\mathrm{sort}(a)$ \Comment{Thm.~\ref{thm:rounding}}
\State \textbf{if} $\Lam(\pi)<\Lam(\pi^\star)$ \textbf{then} $\pi^\star\gets\pi$
\Until{no strict decrease, or 10\% of $B$ used}
\State $\pi\gets\pi^\star$; $c\gets0$
\While{time $<B$ minus the reserved final pass}
 \State draw 256 moves ($u$ next to a neighbor $v$); apply if $\Delta<0$, or if $\Delta=0$ w.p.\ $\rho$
 \State \textbf{if} $\Lam(\pi)<\Lam(\pi^\star)$ \textbf{then} $\pi^\star\gets\pi$; $c\gets0$ \textbf{else} $c\gets c+1$
 \State \textbf{if} $c=S$ \textbf{then} $\pi\gets\pi^\star$ with $\kappa$ random moves; $c\gets0$
\EndWhile
\State final tangent pass on $\pi^\star$ (lines 2--3 until no strict decrease, at most 100 rounds, without the 10\% limit; time reserved in line 6, charged)
\State \Return $\pi^\star$
\end{algorithmic}
\end{algorithm}

\paragraph{Orientation-space search.} Using Corollary~\ref{cor:free}, we searched over orientations with edge flips and directed-path reversals of length up to four, which move one unit of out-degree between the path's endpoints. Exact $F_0$ deltas were verified on 13{,}937 random moves by one implementation and on 6{,}016 by a second, independent one. The search used plateau moves, kicks and periodic tangent rounding. In a separate tuning comparison of path variants on 18 graphs, adjoint-guided and longer paths were slightly worse than single flips, and even the best post-hoc path configuration per graph lost to ILS on every graph (by 0.65--14.9 points). This comparison differs from the ablation of Section~\ref{sec:ablation}, which runs the selected configuration on 16 non-bipartite graphs at budgets of 1 and 10\,s.

\paragraph{$\Cpp$ lower bound tightness.} On all 1{,}245 non-isomorphic graphs with two to seven vertices and at least one edge, the path-reversal bound equals $\min_\pi\Cpp$ on 72 graphs; the mean relative gap is 17.3\% and the maximum 40\% ($K_3$). On the four public development graphs, smallest-last is within 0.42--2.12\% of the bound; on non-regular synthetic graphs within 0.57--8.4\%; on 4-regular graphs within 10.6\%, where an Eulerian orientation attains the bound but no order can.

\paragraph{Licenses.} SNAP datasets are distributed for research use by their maintainers \citep{leskovec2014snap}. VOLT, which provides Neigh, is GPL-3.0 licensed; a small build patch leaves its ordering objective and update rule unchanged, and VOLT is not redistributed. Our code, protocols, run logs and per-graph results are at \repo; Table~\ref{tab:repo} maps its folders to the paper. Saved orders are not included because of their size; the drivers regenerate them from the seeds stored with each result.

\begin{table}[H]
\centering\small
\caption{Contents of the code repository \repo.}
\label{tab:repo}
\begin{tabular}{@{}>{\raggedright\arraybackslash}p{0.33\textwidth}>{\raggedright\arraybackslash}p{0.62\textwidth}@{}}
\toprule
Folder & Contents (sections of the paper)\\
\midrule
\repopath{0_dove_baseline_and_kernels} & \DOVE, reference orders, C++ exact lister, VOLT setup and patch (Section~\ref{sec:relax}, Appendix~\ref{app:protocol})\\
\repopath{1_preregistered_eval}, \repopath{1_preregistered_eval_results} & Pre-registration and its SHA-256, evaluation pipeline, tests, per-graph results, run plans, environment record and run log for the 32 development graphs and the held-out graphs (Section~\ref{sec:experiments}, Appendices~\ref{app:protocol} and~\ref{app:prereg})\\
\repopath{2_secondary_analyses}, \repopath{2_secondary_analyses_results} & Secondary analyses, kernels and per-graph results (Section~\ref{sec:posthoc}, Appendix~\ref{app:posthoc})\\
\repopath{3_heldout_with_neigh} & Protocol, hashes, protocol history, driver and results of the Neigh-inclusive held-out analysis (Appendix~\ref{app:neighsupp})\\
\repopath{4_exact_small_instances} & Exact optima of 26 graphs with at most seven vertices by enumerating all orders, validation and geometry checks\\
\bottomrule
\end{tabular}
\end{table}

\paragraph{Graphs and timing.} Development synthetic graphs come from six families (Erd\H{o}s--R\'enyi, Barab\'asi--Albert, Watts--Strogatz, powerlaw-cluster, stochastic block, core--periphery; $n=145$--$270$); the prospective group (bipartite, random geometric, random regular) was generated after all tuning, and the small certified instances use seven families (Erd\H{o}s--R\'enyi, Barab\'asi--Albert, Watts--Strogatz, powerlaw-cluster, stochastic block, geometric, regular). The four public development graphs are Bitcoin-Alpha, P2P-Gnutella05, Loc-Brightkite and GitHub. Small-instance, ablation and random-graph experiments ran on a Linux machine (Intel Xeon 2.1\,GHz, 2 vCPUs; Python 3.11, NumPy 2.4, SciPy 1.17, Numba 0.67, HiGHS 1.15, GCC 13) with one thread per run. The 32-graph and held-out evaluations ran on a Windows machine (AMD Ryzen 5 3600; Python 3.11.9, NumPy 2.1, SciPy 1.15, Numba 0.61). In the small-instance, ablation and random-graph experiments, a candidate counts only if it is exported and checked before the deadline.

\paragraph{Budget accounting in the public-graph runs.} In the 32-graph and held-out runs, the \TILS clock starts before any array is built and stops when the final order is returned, so it covers the initial tangent phase, ILS and the final pass; verification and export are not charged to any method. ILS stops once the elapsed time reaches $B-t_{\mathrm f}$, where $t_{\mathrm f}$ is the time reserved for the final pass (the first protocol amendment, Appendix~\ref{app:prereg}); the last batch of proposals is not interrupted. The final pass then runs until no strict decrease (at most 100 rounds), and any time beyond $B$ is reported as overrun. The main (returned-order) analysis therefore permits these reported overruns, whereas the capped analyses below rescore the saved orders so that no run uses time beyond $B$; the two estimate different quantities, and neither is a rerun with an enforced deadline. Averaged over runs, \TILS used 0.5\% less than $B$ on the held-out graphs and 0.1\% less on the 32 development graphs. Of the 123 runs (41 graphs, three seeds), 44 ended after $B$: 38 by less than 1\% of $B$, five by 1.1--1.6\% on the 32 development graphs with budgets below 0.3\,s, and one by 3.7\% (Wiki-Talk, where the final pass hit 100 rounds). On sub-second budgets the last ILS batch can itself end a few milliseconds after $B$ (23 runs, at most 6\,ms). The final pass lowers work in 108 of 123 runs, by 0.10 points of reduction on average (0.05 on the 32 development graphs, 0.28 on the held-out graphs, at most 1.02). Of the total, the initial tangent phase accounts for 4.43 (32 development graphs) and 7.71 (held-out) points and ILS for 1.22 and 6.84; without the prefix, ILS alone reaches the same level (Section~\ref{sec:posthoc}). In the 44 over-budget runs, the returned order did not yet exist at $B$. Our runs do not log timestamped incumbents, so the strict rule is a lower bound for a hard-capped implementation, not an estimate of it; logging the incumbent at $B$ would settle this. Two capped scores bound the effect of overruns. Dropping the final pass in every over-budget run gives means of 5.69\% (32 development graphs) and 14.78\% (held-out), with \TILS against \DOVE at 31/1/0 and 9/0/0. The strict rule, the conservative analysis of Section~\ref{sec:experiments}, counts only saved orders known to exist at $B$ (the ILS incumbent if ILS ended by $B$, otherwise the order after the initial tangent phase), and may therefore miss a better incumbent that existed at $B$ but was not saved. It gives 5.57\% and 14.17\% (95\% bootstrap intervals \mbox{[3.79, 7.52]\%} and \mbox{[9.07, 20.36]\%}) and 27/1/4 and 8/0/1; the losses come from runs whose ILS phase ended at most 6\,ms late on budgets below 2\,s and fall back to the tangent-phase order.

\section{Pre-Registered Held-Out Evaluation}\label{app:prereg}
The protocol was fixed and hashed before any held-out graph was downloaded. On Wiki-Talk, the candidate orders of the reference span a factor of 4.6 in work: ascending degree gives 57,386,025 candidate tests (the reference), the input order 257,969,291, and the reversed smallest-last (degeneracy) order 264,800,723.
\paragraph{Graphs.} Email-Eu-core, CollegeMsg, LastFM-Asia, Deezer-Europe, Facebook page--page, Amazon0302, Wiki-Talk, com-Youtube, web-Google and roadNet-CA (communication, social, product, web and road networks), all from SNAP \citep{leskovec2014snap}. None of these files was used before the protocol was fixed; two share sources with the 32 development graphs (Amazon0302 and Com-Amazon are different crawls of Amazon co-purchasing data; Email-Eu-core and Email-EuAll are different datasets of email at a European research institution). The first two columns of each file are read, comments and headers are skipped, and the common preprocessing is applied: remove self-loops, symmetrize, deduplicate and remap non-isolated vertices densely. A graph that fails to download or parse is reported as missing and never replaced. In the executed run, Facebook page--page was missing (HTTP 404 on all five attempts); no VOLT binary was available, so Neigh did not enter the held-out reference; there was no other deviation. Two amendments were made before any held-out file was downloaded: one fixed the time reserved for the final tangent pass inside the budget, and one added the bootstrap interval and sign test to the reported statistics.
\paragraph{Methods and metrics.} The reference, T, \DOVE and \TILS are as above; \TILS receives \DOVE's measured time on each graph and uses seeds 1--3. Neigh enters the reference only if the VOLT binary is available on the evaluation machine, and this is reported. We report per-graph reduction relative to $R^\star$, the mean, median, improved/tied/regressed counts, the geometric-mean work ratio, and paired wins/ties/losses of \TILS against \DOVE and T. Exact work and triangle counts are recomputed for every saved order. Graphs with more than 20M edges would be skipped for optimization; none of the listed graphs is that large.
\paragraph{Record and its limits.} The repository contains the protocol (\repopath{1_preregistered_eval/PREREGISTRATION.md}, which records its writing time as 2026-09-28 05:58 UTC), its SHA-256, the evaluation code, which refuses to start the held-out part if the hash differs, and the run log, which records the hash check at 07:44 UTC, before the first download. These records were produced by our own scripts and made public after the run, and the protocol was not deposited with an independent registry. A checksum identifies a file but does not by itself show when the file existed, so these records document the protocol and its execution without independently timestamping them. The Neigh-inclusive analysis (Appendix~\ref{app:neighsupp}) has its own hashed protocol and protocol history, the reuse study (Appendix~\ref{app:reuse}) was frozen with a checksum before its full run, and all analyses labeled secondary or post hoc are exploratory.

\section{Fixed-Temperature Landscape Diagnostic}\label{app:fixedtau}
Theorem~\ref{thm:landscape} concerns unnormalized scores at a fixed temperature, whereas \DOVE anneals $\tau$ and standardizes the scores after every step. To separate the two effects, we ran gradient descent and Adam on $J$ with $\tau=0.1$ and $\epsilon=0.01$ fixed for 20{,}000 steps, from ascending degree, on an Erd\H{o}s--R\'enyi graph $G(2000,8/1999)$ (largest component: $n=1998$, $m=7938$) and a Barab\'asi--Albert graph ($n=2000$, four edges per new vertex, $m=7984$); Figure~\ref{fig:fixedtau} shows the ER graph. Without standardization, Adam saturates as predicted on both graphs: $\|\nabla J\|$ falls by more than eight orders of magnitude, the score range $\max_v s_v-\min_v s_v$ grows from $2$ to $14.1$ and $17.4$, 99\% of the edges have $p_{uv}(1-p_{uv})<10^{-6}$, and the Rayleigh quotient of the numerically computed Fiedler vector is below $10^{-20}$. This value is a floating-point estimate, not a certified bound: saturated weights can underflow (at a score gap of $40\tau$, $\sigma$ rounds to 1 in double precision although the exact weight $e^{-40}/(1+e^{-40})^2\approx4.2\times10^{-18}$ is positive), and a rigorous upper bound would evaluate $w(z)=e^{-|z|}/(1+e^{-|z|})^2$ directly, project the test vector onto the zero-sum subspace of each component, and control the rounding error. Theorem~\ref{thm:landscape} itself is exact; gradient descent behaves the same way on the ER graph (five orders of magnitude) but had not converged on the BA graph. With per-step standardization the score range stays bounded and $\|\nabla J\|$ levels off at 0.3--3.9, so the iterates are not stationary points of $J$ and part~(ii) does not apply to them: gradient descent approaches stationary points of the norm-constrained problem, while Adam keeps moving. In every variant, including the annealed schedule, the order exported from the scores ends 14--20\% above the work that iterated local search reaches on the same graph, and a final tangent pass recovers at most 0.6\%.

\begin{figure}[H]
\centering
\includegraphics[width=0.55\textwidth]{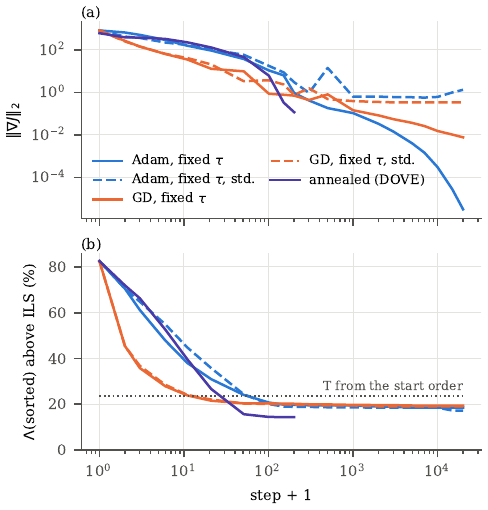}
\caption{Erd\H{o}s--R\'enyi graph ($n=1998$ after taking the largest component, $m=7938$), fixed $\tau=0.1$, $\epsilon=0.01$ unless annealed. (a) $\|\nabla J\|_2$ per step. (b) Exact work of the sorted scores, in \% above the best order found by iterated local search on the same graph; dotted: tangent refinement from the start order. Solid: unnormalized scores; dashed: scores standardized after each step.}
\label{fig:fixedtau}
\end{figure}

\section{Secondary Analyses on Public Graphs (Not Pre-Registered)}\label{app:posthoc}
These comparisons were run after the pre-registered analysis, from the saved results and cached graphs, without modifying them. Coloring orders are greedy colorings with vertices processed in reverse smallest-last order (SL) or by descending degree (DEG), with color classes as blocks in ascending or descending color order. A descending order induces the same orientation as the reversal of the ascending one, so the eight orders have at most four distinct work values. ILS-T is ILS from $R^\star$ without the tangent prefix, followed by the final tangent pass, with the same budget, seeds and parameters as \TILS. C+ILS-T is ILS-T started from a coloring order, with its construction charged. Kernel times are medians of seven alternating repetitions of a Numba port of the adaptive-forward lister; they are not the C++ bitmap kernel of our baseline code (\repopath{0_dove_baseline_and_kernels} in the repository). The intervals in Section~\ref{sec:posthoc} are percentile bootstraps over graphs of the per-graph ratio of median times (seed-1 orders); per-graph intervals paired over the seven repetitions lie entirely below 1 on 15 of 41 graphs (7 of 9 held-out) and above 1 on 9 (1 held-out). Tables~\ref{tab:ph-colour}--\ref{tab:ph-timing} list the held-out graphs; the corresponding tables for the 32 development graphs are in the repository (\repopath{2_secondary_analyses_results}).

\begin{table}[H]
\centering\small
\caption{Held-out graphs: coloring orders, reduction of work relative to $R^\star$ (\%; negative = more work), the best variant (rev.: reversed), and the construction time of DEG-desc (s).}
\label{tab:ph-colour}
\begin{tabular}{lrrrrrlr}
\toprule
Graph & $\Lam_{R^\star}$ & SL-asc & SL-desc & DEG-asc & DEG-desc & Best & $t$\\
\midrule
Email-Eu-core & 230.1k & $-$9.84 & $-$26.70 & $-$5.80 & $-$28.83 & DEG-asc & 0.00 \\
CollegeMsg & 69.0k & $-$82.40 & $-$81.51 & $-$47.35 & $-$68.41 & DEG-asc & 0.00 \\
LastFM-Asia & 96.0k & $-$15.83 & $-$33.35 & $-$13.95 & $-$46.58 & DEG-asc & 0.01 \\
Deezer-Europe & 225.5k & $-$15.85 & $-$23.98 & $-$18.59 & $-$32.83 & SL-asc & 0.03 \\
Amazon0302 & 1.376M & $-$22.95 & $-$17.36 & $-$22.13 & $-$24.58 & SL-asc, rev. & 0.35 \\
Wiki-Talk & 57.386M & $-$1.53 & $-$316.27 & $-$1.84 & $-$324.29 & SL-asc & 1.81 \\
com-Youtube & 21.995M & $-$16.54 & $-$76.42 & 2.11 & $-$89.04 & DEG-asc & 1.29 \\
web-Google & 16.783M & $-$24.38 & $-$27.60 & $-$9.65 & $-$30.41 & DEG-asc & 3.39 \\
roadNet-CA & 861.6k & 9.76 & 9.71 & $-$11.11 & $-$11.03 & SL-asc & 1.01 \\
 
\bottomrule
\end{tabular}
\end{table}

\begin{table}[H]
\centering\small
\caption{Held-out graphs: reduction (\%) relative to $R^\star$ of \TILS, ILS without the tangent prefix (ILS-T), and ILS-T started from a coloring order (C+ILS-T); mean$\pm$SD over seeds 1--3.}
\label{tab:ph-ils}
\begin{tabular}{lrrrr}
\toprule
Graph & \TILS & ILS-T & C+ILS-T (DEG-desc) & C+ILS-T (best)\\
\midrule
Email-Eu-core & 9.52 & 9.61$\pm$0.07 & 9.56$\pm$0.05 & 9.76$\pm$0.01 \\
CollegeMsg & 21.29 & 21.29$\pm$0.00 & 21.18$\pm$0.08 & 21.26$\pm$0.02 \\
LastFM-Asia & 8.45 & 8.46$\pm$0.03 & 8.72$\pm$0.02 & 8.38$\pm$0.31 \\
Deezer-Europe & 9.34 & 9.39$\pm$0.15 & 9.33$\pm$0.16 & 9.05$\pm$0.17 \\
Amazon0302 & 8.73 & 8.82$\pm$0.04 & 8.50$\pm$0.02 & 8.73$\pm$0.05 \\
Wiki-Talk & 13.04 & 12.26$\pm$0.13 & 11.09$\pm$0.02 & 11.86$\pm$0.09 \\
com-Youtube & 24.43 & 24.63$\pm$0.02 & 26.21$\pm$0.37 & 28.21$\pm$0.00 \\
web-Google & 5.16 & 5.18$\pm$0.00 & 3.67$\pm$0.06 & 4.46$\pm$0.00 \\
roadNet-CA & 33.50 & 33.72$\pm$0.01 & 33.10$\pm$0.00 & 35.85$\pm$0.06 \\
 
\bottomrule
\end{tabular}
\end{table}

\begin{table}[H]
\centering\small
\caption{Held-out graphs: median kernel time of $R^\star$ (ms), kernel-time ratios to $R^\star$, and ordering times (s).}
\label{tab:ph-timing}
\begin{tabular}{lrrrrrrr}
\toprule
Graph & $R^\star$ ms & T & \DOVE & \TILS & $t_{\rm T}$ & $t_{\rm DOVE}$ & $t_{\rm T+ILS}$\\
\midrule
Email-Eu-core & 0.47 & 1.017 & 1.019 & 1.007 & 0.002 & 0.3 & 0.3 \\
CollegeMsg & 0.30 & 0.982 & 0.935 & 0.929 & 0.001 & 0.2 & 0.2 \\
LastFM-Asia & 0.68 & 1.011 & 0.998 & 0.992 & 0.003 & 0.4 & 0.4 \\
Deezer-Europe & 2.36 & 1.014 & 0.993 & 0.990 & 0.011 & 1.9 & 1.9 \\
Amazon0302 & 21.13 & 0.995 & 0.988 & 0.985 & 0.172 & 34.1 & 34.0 \\
Wiki-Talk & 305.94 & 0.984 & 1.006 & 0.974 & 28.890 & 228.8 & 230.5 \\
com-Youtube & 136.12 & 0.941 & 0.927 & 0.912 & 4.567 & 124.1 & 121.6 \\
web-Google & 260.61 & 0.994 & 0.938 & 0.961 & 5.231 & 168.6 & 165.0 \\
roadNet-CA & 46.34 & 0.969 & 0.959 & 0.909 & 0.489 & 130.4 & 130.2 \\
 
\bottomrule
\end{tabular}
\end{table}

\section{Held-Out Reference with Neigh (Supplementary)}\label{app:neighsupp}
\paragraph{Protocol.} This run is not part of the pre-registration. Its analyses were fixed in a protocol file (\repopath{3_heldout_with_neigh/SUPP_PROTOCOL.md}, hashed in \repopath{SHA256SUMS.txt}) before any held-out Neigh order existed, and it only reads the original results. An amendment, also recorded before any held-out Neigh order existed, moved its execution from the Windows evaluation machine to the Linux machine of the synthetic experiments (2 vCPUs, one thread per run). There, VOLT \citep{lecuyer2023tailored} is built by the route of our baseline code (\repopath{scripts/setup_volt.sh}: a pinned VOLT commit and a reproducibility patch); this build reproduces all 20 stored Neigh orders of the four public development graphs exactly. The graphs are the caches written by the original run and are checked against it (dense-edge SHA-256 and the work of every saved order). A Windows build route prepared for the evaluation machine is documented in the repository but was not used. Neigh uses threshold 0.01 and seeds 1--5 (\NSNeighSeeds{} runs succeeded), each order in both directions. The strengthened reference $R^\star_N$ is the lowest-work order among the original $R^\star$ and these ten orders (ties keep $R^\star$).

\paragraph{Reruns.} If $R^\star_N\ne R^\star$, T, \DOVE and \TILS are rerun from $R^\star_N$ with the pre-registered code, seeds and budget rule: \TILS receives the new \DOVE time, measured on the Linux machine. For the five largest changed graphs this time is 21--46\% shorter than on the evaluation machine, which likely explains why the reruns from $R^\star_N$ end below the original \TILS orders on Deezer-Europe, Amazon0302 and roadNet-CA. If $R^\star_N=R^\star$, the original runs are reused (\NSReReused{} graph, Wiki-Talk, where every Neigh order is worse). Every new order is verified as in the main run (\NSCheck).

\paragraph{Results.} Against $R^\star_N$, the original orders reach \NSFrozenTMean\% (T), \NSFrozenDoveMean\% (\DOVE) and \NSFrozenTILSMean\% (\TILS; 95\% bootstrap interval \mbox{\NSFrozenTILSCI}, improved/tied/regressed \NSFrozenTILSITR). The reruns give \NSReTMean\% (T), \NSReDoveMean\% (\DOVE, which returns its start order on Amazon0302, web-Google and roadNet-CA) and \NSReTILSMean\% (\TILS; \mbox{\NSReTILSCI}; improved/tied/regressed \NSReTILSITR). Wins/ties/losses of \TILS against \DOVE and against T are \NSReTILSvsDove{} and \NSReTILSvsT{} (sign tests, $p=\NSReTILSvsDoveP$ each); Table~\ref{tab:neighsupp} lists the graphs. Of the \NSReRuns{} \TILS runs scored against $R^\star_N$ (reruns and reused original runs), \NSReOverRuns{} ended after $B$ (largest overrun \NSReMaxOver, Wiki-Talk). Rescoring them conservatively by a saved order known to exist at $B$ gives \NSReTILSCappedMean\% (\mbox{\NSReTILSCappedCI}), and \TILS still beats \DOVE and T on all nine graphs (9/0/0 each, two-sided $p=0.0039$); dropping only their final pass gives \NSReTILSNoFinalMean\%. As in Appendix~\ref{app:protocol}, these are rescorings of saved orders, not reruns with an enforced deadline.

\paragraph{Interpretation.} Rescoring the original orders isolates the effect of the stronger reference ($14.83\%\to\NSFrozenTILSMean\%$). The reruns additionally change the start order, the machine and the budgets, so the difference between \NSFrozenTILSMean\% and \NSReTILSMean\% does not estimate the benefit of restarting from Neigh, and these comparisons do not establish hardware-independent superiority of a time-limited optimizer. A per-run table in the repository (\repopath{3_heldout_with_neigh/results_linux/supp_runs.csv}) marks each record as original-order rescoring, Linux rerun or reused evaluation-machine run; each \TILS run scored against $R^\star_N$ also lists its budget, times and strict-cutoff incumbent.

\paragraph{Kernel time.} Kernel times were measured on the Linux machine (21 alternating repetitions; geometric mean over graphs of the per-graph paired ratios). \TILS relative to $R^\star_N$: \NSTimeReGeo{} (\mbox{\NSTimeReCI}; work ratio \NSTimeReWorkGeo). The original \TILS orders relative to the original $R^\star$: \NSTimeFrozenGeo{} (\mbox{\NSTimeFrozenCI}; work ratio \NSTimeFrozenWorkGeo), whereas the evaluation machine measured \PHHeldTimeRatioTILS{} (\mbox{[0.938, 0.984]}) for the same orders (Section~\ref{sec:posthoc}). The held-out kernel-time gain was therefore not reproduced on a second machine, and there is none against $R^\star_N$.

\IfFileExists{supp_heldout_rows.tex}{\def\supprowsfile{supp_heldout_rows.tex}}{\def\supprowsfile{supp_heldout_rows_default.tex}}
\begin{table}[H]
\centering\small\setlength{\tabcolsep}{3pt}
\caption{Supplementary held-out run with Neigh (not pre-registered). $R^\star_N$: winning reference candidate; Gap: reduction of $R^\star_N$ relative to the original $R^\star$ (\%); Original: original \TILS orders relative to $R^\star_N$; Rerun: methods rerun from $R^\star_N$ (\%, relative to $R^\star_N$; \TILS mean$\pm$SD over three seeds). Reruns on the Linux machine; $^\dagger$: $R^\star_N=R^\star$, original evaluation-machine runs reused.}
\label{tab:neighsupp}
\begin{tabular}{llrrrrr}
\toprule
& & & Original & \multicolumn{3}{c}{Rerun from $R^\star_N$}\\
Graph & $R^\star_N$ & Gap & \TILS & T & \DOVE & \TILS\\
\midrule
\multicolumn{7}{c}{\pend{} (supplementary run pending)} \\
 
\bottomrule
\end{tabular}
\end{table}

\section{Budget Curve}\label{app:budget}
To check that the advantage of \TILS does not depend on the single 200-step schedule, we ran, on the four public development graphs, \DOVE with 25, 50, 100, 200 and 400 steps (the temperature is annealed over the steps) from the same $R^\star$, and gave \TILS the measured \DOVE time at each step count (\DOVE: one run, seed 7; \TILS: three seeds; Linux machine, one thread; Table~\ref{tab:budget}). \TILS is ahead at every budget on every graph and seed, and \TILS with the 25-step budget (one sixth to one eighth of the default time) is ahead of \DOVE with 400 steps on all four graphs. \DOVE is not monotone in its schedule (GitHub: 4.15\% at 200 steps, 2.38\% at 400).

\begin{table}[H]
\centering\small\setlength{\tabcolsep}{3pt}
\caption{Budget curve: reduction (\%) relative to $R^\star$ with \DOVE at 25--400 steps and \TILS at the matched measured time $B$ (mean of three seeds).}
\label{tab:budget}
\begin{tabular}{llrrrrr}
\toprule
Graph & & 25 & 50 & 100 & 200 & 400\\
\midrule
Bitcoin-Alpha & \DOVE & 6.02 & 8.30 & 8.86 & 9.07 & 9.10 \\
 & \TILS & 9.30 & 9.40 & 9.45 & 9.49 & 9.51 \\
 & $B$ (s) & 0.04 & 0.08 & 0.15 & 0.29 & 0.60 \\
P2P-Gnutella05 & \DOVE & 2.31 & 2.77 & 2.86 & 2.83 & 2.87 \\
 & \TILS & 3.74 & 4.27 & 4.60 & 4.85 & 4.85 \\
 & $B$ (s) & 0.08 & 0.16 & 0.32 & 0.64 & 1.23 \\
Loc-Brightkite & \DOVE & 5.93 & 7.75 & 8.39 & 8.62 & 8.66 \\
 & \TILS & 9.05 & 9.13 & 9.33 & 9.43 & 9.52 \\
 & $B$ (s) & 0.62 & 0.94 & 1.91 & 3.68 & 7.28 \\
GitHub & \DOVE & 0.00 & 3.53 & 4.01 & 4.15 & 2.38 \\
 & \TILS & 6.36 & 6.37 & 6.39 & 6.40 & 6.42 \\
 & $B$ (s) & 0.62 & 1.31 & 2.46 & 4.84 & 9.59 \\
\bottomrule
\end{tabular}
\end{table}

\section{A Second Concave Objective: Minimum-Entropy Orientation}\label{app:entropy}
Minimum-entropy orientation \citep{cardinal2008entropy} assigns each edge to one endpoint and minimizes the entropy (in bits) of the load distribution $x_v/m$. With $x$ the forward-degree vector of an order, the objective is $F(x)=\sum_v f(x_v)$ with $f(t)=-(t/m)\log_2(t/m)$ and $f(0)=0$. It is continuous and concave on $\ZG$, so some order attains the minimum. Its derivative $f'(t)=-(\log_2(t/m)+1/\ln2)/m$ is finite for $t>0$ but $f'(0^+)=+\infty$, and every order has vertices of forward degree zero (its sinks), so Theorem~\ref{thm:rounding}, which needs a finite supergradient, does not apply as stated. The following boundary extension restricts the tangent step to the face on which zero loads stay zero.

\begin{proposition}[Tangent step with zero loads]\label{prop:entropyboundary}
Let $F(x)=\sum_vf(x_v)$ with $f$ continuous and concave on $[0,\infty)$ and differentiable on $(0,\infty)$. Let $x=\dout_\pi$ for an order $\pi$, $Z=\{v:x_v=0\}$, $U=V\setminus Z$ and $a_v=f'(x_v)$ for $v\in U$. Let $\pi'$ list $U$ by ascending $(a_v,\mathrm{rank}_\pi(v))$, followed by $Z$ in $\pi$-order, and let
\[
\delta=\sum_{v\in U}a_vx_v-\!\!\sum_{\substack{uv\in E\\u,v\in U}}\!\!\min(a_u,a_v)-\!\!\sum_{\substack{uv\in E\\u\in U,\,v\in Z}}\!\!a_u .
\]
Then $\delta\ge0$ and $F(\dout_{\pi'})\le F(x)-\delta$.
\end{proposition}
\begin{proof}
$Z$ is independent, because the earlier endpoint of any edge has positive forward degree. In $\pi'$ every vertex of $Z$ comes after all its neighbors, so $y=\dout_{\pi'}$ has $y_Z=0$, each edge between $U$ and $Z$ is charged to its endpoint in $U$, and each edge inside $U$ is charged to its endpoint with smaller $a$. Hence $\sum_{v\in U}a_vy_v$ equals the two sums subtracted in $\delta$. This is the minimum of $\sum_{v\in U}a_vz_v$ over the face $\{z\in\ZG:z_Z=0\}$, on which every edge between $U$ and $Z$ puts its whole unit on its endpoint in $U$ and every edge inside $U$ contributes at least $\min(a_u,a_v)$; $x$ lies on this face, so $\delta=\sum_{v\in U}a_v(x_v-y_v)\ge0$. For $v\in U$, concavity of $f$ on $[0,\infty)$ and differentiability at $x_v>0$ give $f(t)\le f(x_v)+a_v(t-x_v)$ for all $t\ge0$. Summing over $U$ and using $y_Z=x_Z=0$ gives $F(y)\le F(x)+\sum_{v\in U}a_v(y_v-x_v)=F(x)-\delta$.
\end{proof}
No infinite quantity enters $\delta$. Our implementation realizes $\pi'$ by giving zero-load vertices the key $+\infty$ and breaking ties by the current rank, and it accepts a step only on a strict decrease of $F$. For the relaxation, after isolated vertices are removed, every finite score vector gives positive soft degrees, where $F$ is differentiable, so Theorem~\ref{thm:landscape}(ii)--(iii) apply on this open domain; (ii) yields a global maximizer over $\ZG$ by continuity. The greedy algorithm analyzed by \citet{cardinal2008entropy} (additive error at most $\log_2e\approx1.44$ bits, which is tight; they also give a different linear-time algorithm within one bit) repeatedly takes the vertex covering most uncovered edges, which is itself an order. We compare random orders, tangent refinement from them, greedy, tangent refinement from greedy, and Adam on the unconstrained pairwise-logistic relaxation at fixed $\tau=0.1$ (3000 steps, learning rate 0.05, initial scores i.i.d.\ $N(0,1)$), rounded by sorting and then refined; three seeds.

As Theorem~\ref{thm:landscape} predicts, the relaxation only saturates. With $N(0,1)$ initial scores and $\tau=0.1$ it starts nearly saturated (on Bitcoin-Alpha, seed 1, numerical estimates of $\lambda_*$ are $2.5\times10^{-12}$ at step 0 and $1.3\times10^{-13}$ at the end; see Appendix~\ref{app:fixedtau} for their limits); on all four graphs the gradient norm falls by more than three orders of magnitude while the score range keeps growing, and the rounded orders end 0.05--0.14 bits above greedy, which tangent refinement barely changes. From near-uniform scores ($N(0,0.01^2)$, Bitcoin-Alpha, three seeds) the dynamics are visible: the estimate of $\lambda_*$ falls from 0.18 to $3.5\times10^{-9}$ within 100 steps and below numerical resolution afterwards, the gradient norm falls from 0.92 to $9.5\times10^{-5}$, and the rounded order (8.356 bits, 8.355 after refinement) is as good as greedy (8.357). Saturation thus says nothing about the quality of the rounded order, in either direction, which is the content of Theorem~\ref{thm:landscape}(iii). Tangent refinement alone brings random orders to within 0.06--0.19 bits of greedy in 9--19 sorts, and the greedy order is a fixed point of this step on every graph: $\pi'=\pi$ and $\delta=0$ exactly (a zero gap is not a certificate of global optimality; Appendix~\ref{app:proofs}). These comparisons concern the interpretation of gradients and certificates; they do not establish a better entropy optimizer than greedy.

\begin{table}[H]
\centering\small\setlength{\tabcolsep}{2.5pt}
\caption{Minimum-entropy orientation: entropy of the load distribution (bits; lower is better), mean of three seeds where random. +T: after tangent refinement of the order to its left; Relax.: relaxation from $N(0,1)$ initial scores, rounded by sorting.}
\label{tab:entropy}
\begin{tabular}{lrrrrrrrr}
\toprule
Graph & $n$ & $m$ & Random & +T & Greedy & +T & Relax. & +T\\
\midrule
Bitcoin-Alpha & 3,783 & 14,124 & 9.87 & 8.46 & 8.357 & 8.357 & 8.423 & 8.419 \\
P2P-Gnutella05 & 8,846 & 31,839 & 12.14 & 11.50 & 11.311 & 11.311 & 11.455 & 11.453 \\
GitHub & 37,700 & 289,003 & 13.15 & 11.33 & 11.270 & 11.270 & 11.319 & 11.316 \\
Loc-Brightkite & 58,228 & 214,078 & 13.89 & 12.88 & 12.766 & 12.766 & 12.845 & 12.839 \\
\bottomrule
\end{tabular}
\end{table}

\section{Degree-Preserving Null Ensembles}\label{app:nullens}
Network-motif analysis compares the number of triangles and other small subgraphs in a network with an ensemble of randomizations that preserve every vertex degree \citep{milo2002motifs,alon2007motifs}. For each of the four public development graphs we drew five such randomizations by double-edge swaps ($10m$ successful swaps, at most $100m$ attempts, NetworkX \citep{hagberg2008networkx}; degree sequences checked) and measured the work of ascending and descending degree, Split (better direction), smallest-last, and \TILS with a 1\,s budget from the best of these orders (seed 1; Linux machine). Table~\ref{tab:nullens} also reports the limit of Theorem~\ref{thm:mf} for ascending and descending degree, evaluated for the empirical degree distribution (4 quadrature nodes per class), although these distributions violate its bounded-degree assumption. This analysis is exploratory and not pre-registered; ten successful swaps per edge and five randomizations need not give equilibrated, independent samples from the uniform null, and the graphs are not biological networks.

The limit matches the sampled work to within 0.3\% on P2P-Gnutella05 and 4.3\% on Loc-Brightkite, but overestimates it by 18--26\% on Bitcoin-Alpha and 29--97\% on GitHub. The error grows with the ratio of the maximum degree to the structural cutoff $\sqrt{2m}$ \citep{boguna2004cutoffs}, beyond which a simple graph cannot keep the degrees of adjacent vertices independent, as the configuration model assumes. The V-shaped Split order beats ascending degree by 9.8\% on P2P-Gnutella05, as the limit predicts (V layout: 9.4\% below ascending), but by only 0.5\% on Bitcoin-Alpha and 1.6\% on Loc-Brightkite, and it is 1.9\% worse on GitHub. On every ensemble, \TILS removes 16.5--31.4\% of the work of the best degree-based or degeneracy order, so the gain from looking beyond degree is not an artifact of the structure of the original networks. Here each randomization is optimized separately; Appendix~\ref{app:reuse} tests whether one optimized order can be reused along a chain of randomizations.

\begin{table}[H]
\centering\small\setlength{\tabcolsep}{2pt}
\caption{Degree-preserving null ensembles (five randomizations): mean work ($\times10^3$) of ascending (Asc) and descending (Desc) degree, sampled and limit (MF) values, Split (better direction), smallest-last (SL), and \TILS with 1\,s (reduction relative to the best of the standard orders, \%). $r=d_{\max}/\sqrt{2m}$.}
\label{tab:nullens}
\begin{tabular}{lrrrrrrrr}
\toprule
Graph & $r$ & Asc & MF & Desc & MF & Split & SL & \TILS\\
\midrule
Bitcoin-Alpha & 3.04 & 87.8 & 104.0 & 152.4 & 191.8 & 87.4 & 95.4 & 68.1 (21.9) \\
P2P-Gnutella05 & 0.35 & 98.5 & 98.7 & 120.9 & 121.2 & 88.9 & 128.9 & 61.0 (31.4) \\
Loc-Brightkite & 1.73 & 1,407.5 & 1,456.4 & 2,466.2 & 2,571.9 & 1,384.8 & 1,491.3 & 1,040.8 (24.8) \\
GitHub & 12.44 & 3,563.8 & 4,606.4 & 6,409.5 & 12,655.8 & 3,632.3 & 3,707.2 & 2,977.6 (16.5) \\
\bottomrule
\end{tabular}
\end{table}

\section{Reusing Orders along a Null-Model Chain}\label{app:reuse}
\paragraph{Design (pre-registered).} Motif analysis lists triangles on many samples of a degree-preserving null model \citep{milo2002motifs}; consecutive samples of a double-edge-swap chain are related graphs with fixed vertex identities. The protocol (hypotheses, claim rules, policies and analysis) was frozen with a checksum before any full run; an earlier pilot on three of the graphs is disclosed in it. For each of eight graphs (Bitcoin-Alpha, P2P-Gnutella05, Loc-Brightkite, GitHub, and the held-out Email-Eu-core, CollegeMsg, LastFM-Asia, Deezer-Europe, i.e., the four held-out graphs with at most 100k edges; only their degree sequences are used, and the pre-registered held-out analysis is unaffected) and two chain seeds, we run $10m$ successful swaps of burn-in and then 20 frames (samples) separated by $s\cdot m$ successful swaps, $s\in\{0.001,0.01,0.1,1,10\}$, which change on average \RChangedLowS, \RChangedSmallS, \RChangedMidS, \RChangedOneS\ and \RChangedHighS\% of the edges per frame (80 runs, 1,600 frames). Frame~0 is optimized once by \TILS (100\,$\mu$s per edge, 2--60\,s), which removes \RFrameZeroRedMin--\RFrameZeroRedMax\% of the work of the fresh cheap reference $R^\star_{\rm cheap}$ (the best of Natural, Random, Degree, Core, smallest-last, Split and Check, each in both directions; all of its construction and evaluation time is charged). On frames $t\ge1$ we compare: ascending degree (Asc); $R^\star_{\rm cheap}$ (Cheap); Cheap followed by tangent refinement (Cheap+T); the frame-0 order unchanged (Frozen); the previous frame's order repaired by tangent refinement (Reuse+T); the same repaired by ILS without tangent passes for exactly the time Reuse+T took, at least 1\,ms (Reuse+ILS); fresh \TILS from $R^\star_{\rm cheap}$ (10\,$\mu$s per edge, at least 0.5\,s); and fresh Neigh, with and without T (frames 5, 10, 15, 20). As an incremental counting comparison, which maintains the number of triangles rather than listing them and so serves only count-based applications, we maintain the triangle \emph{count} along the chain, updating it for each swap from the common neighborhoods of the four affected edges (Inc) \citep{makkar2017dynamic}. The metric is $r=100(\Lam_{R^\star_{\rm cheap}}-\Lam)/\Lam_{R^\star_{\rm cheap}}$, averaged over 20 frames and both seeds. Every order is checked by listing: the number of candidate tests equals $\Lam$ and the triangle counts agree across all orders and with Inc in all 1,600 frames. Ordering times of the deterministic policies, build and listing times (the marker kernel, a Numba port of the adaptive lister that holds one forward set in a marker array, median of five alternating repetitions; and a kernel whose candidate test is an open-addressing hash probe, median of three) come from a separate single-process timing pass on the saved orders.

\paragraph{Work.} Table~\ref{tab:reuse}, Figure~\ref{fig:reuse}. (H1, passed) Tangent refinement of a fresh cheap order reduces work in all \RConfigs\ configurations (\RCheapTMin--\RCheapTMax\%, mean \RCheapTMean\%). (H2, failed) Reuse beats it only when very few edges change: at $s=0.001$ Reuse+T is better on \RHTwoWinsLowS\ graphs (one-sided sign test, as the protocol specifies, $p=\RHTwoPLowS$; two-sided $p=0.0078$; \RReuseTLowS\% vs.\ \RCheapTLowS\%), but at $s=0.01$ only on \RHTwoWinsSmallS, and at $s\in\{1,10\}$ it is more than half a point worse on every graph (\RReuseTHighS\% vs.\ \RCheapTHighS\% at $s=10$): repairing a stale order by tangent refinement stops at a worse fixed point than refining a fresh cheap order, and the gap accumulates along the chain. The frozen order decays quickly and is worse than $R^\star_{\rm cheap}$ itself from $s=0.1$ on. (H3, passed) At $s=0.001$ Reuse+T is within \RHThreeGapMax\ points of fresh \TILS on all eight graphs at $\ge\RHThreeTimeFactorMin\times$ lower ordering time. (H4, passed) Over all frames and spacings, the per-graph median ordering time of Reuse+T is \RHFourRatioMin--\RHFourRatioMax\ times that of computing $R^\star_{\rm cheap}$. (P2) At matched time, tangent repair beats ILS repair in all 40 configurations (wins/ties/losses \RPTwo); in the time a tangent repair takes, ILS gets only 1--143 chunks of 256 proposals (median 7). Fresh Neigh gains only \RNeighHighS\% over $R^\star_{\rm cheap}$ on these samples (\RNeighTHighS\% after T).

\paragraph{Time.} Table~\ref{tab:reusetime}. (H5) No configuration shows a wall-clock gain: per frame, ordering plus building plus listing with Reuse+T takes \RHFiveRatioLowSMin--\RHFiveRatioLowSMax\ times as long as with ascending degree at $s=0.001$ and up to \RHFiveRatioMax\ times at $s=10$, so the frame-0 optimization never pays off. With the marker kernel the listing time barely depends on $\Lam$ (within-frame slope at most \RGammaMarkerMax\,ns per candidate test; within $R^2\le0.01$ on Bitcoin-Alpha, Loc-Brightkite and GitHub and $\le0.19$ on Email-Eu-core and CollegeMsg); with a hash-probe test it is \RGammaHashMin--\RGammaHashMax\,ns over individual graph--spacing configurations (Table~\ref{tab:reusetime} gives per-graph medians). Ordering beyond a degree sort would pay off only if a candidate test cost more than the break-even $\gamma^\ast$ of \RGammaStarLowSMin--\RGammaStarLowSMax\,ns (Reuse+T, $s=0.001$; up to \RGammaStarMax\,ns at $s=10$; \RGammaStarCheapTMin--\RGammaStarCheapTMax\,ns for Cheap+T). (H6) Updating the count incrementally costs \RIncListLowS\ (median) of one listing at $s=0.001$ and \RIncListSmallS\ at $s=0.01$, but \RIncListMidS, \RIncListOneS\ and \RIncListHighS\ times a listing at $s=0.1$, $1$ and $10$: when few edges change, incremental counting is cheaper than relisting. Peak memory per run was at most 459\,MB.

\paragraph{Stability.} For the carried order, the change of $\Lam$ between consecutive frames never exceeded the bounds of Lemma~\ref{lem:stab} (\RStabViol\ violations in \RStabChecks\ transitions), but the bounds are loose (median ratio of bound to actual change \RStabLocSlack\ for the localized and \RStabUniSlack\ for the uniform bound).

\paragraph{Post hoc (not pre-registered).} Keeping on each frame the better of Reuse+T and Cheap+T and carrying it forward costs 2--14\% more time than Cheap+T and gains \RPHGainLowS, \RPHGainSmallS\ and \RPHGainMidS\ points over Cheap+T at $s=0.001$, $0.01$ and $0.1$ (none at $s\ge1$, where the reused order is never better).

\begin{table}[H]
\centering\small\setlength{\tabcolsep}{3pt}
\caption{Summary of the order-reuse study. Changed edges: mean fraction of edges changed per frame. $r$: reduction of $\Lam$ (\%) relative to $R^\star_{\rm cheap}$, mean over the eight graphs of per-graph means. Fresh: fresh \TILS. Reuse+T/Cheap time: median over the eight graphs of the per-graph ratio of the median ordering time of Reuse+T to that of computing $R^\star_{\rm cheap}$. Inc/list: median over the eight graphs of the per-graph median time of an incremental triangle-count update relative to one listing. Time per frame with Reuse+T is never lower than with ascending degree (Table~\ref{tab:reusetime}).}
\label{tab:reusemain}
\begin{tabular}{@{}rrrrrrrr@{}}
\toprule
$s$ & changed & Cheap+T & Frozen & Reuse+T & Fresh & Reuse+T/Cheap & Inc/list\\
 & edges (\%) & \multicolumn{4}{c}{$r$ (\%)} & time & \\
\midrule
0.001 & 0.2 & 8.4 & 10.0 & 10.1 & 10.4 & 0.02 & 0.09\\
0.01 & 2.0 & 8.4 & 5.6 & 8.2 & 10.3 & 0.06 & 0.71\\
0.1 & 18.0 & 8.3 & $-$8.1 & 5.3 & 10.3 & 0.11 & 4.28\\
1 & 83.1 & 8.4 & $-$12.7 & 4.1 & 10.3 & 0.14 & 38.72\\
10 & 94.8 & 8.4 & $-$12.7 & 4.0 & 10.4 & 0.15 & 382.93\\
\bottomrule
\end{tabular}

\end{table}
\begin{table}[H]
\centering\small\setlength{\tabcolsep}{2pt}
\caption{Order reuse along degree-preserving swap chains: reduction $r$ (\%) of $\Lam$ relative to the fresh cheap reference $R^\star_{\rm cheap}$, mean over 20 frames and two chain seeds; median ordering time per frame (ms); and the median time of the incremental count update relative to one listing (Inc/list). $s$: swaps per frame divided by $m$ (0.2\%, 2.0\% and 94.8\% of edges changed per frame on average). Neigh: frames 5, 10, 15, 20 only.}
\label{tab:reuse}
\begin{tabular}{@{}llrrrrrrrrr@{}}
\toprule
& & \multicolumn{5}{c}{$r$ (\%)} & \multicolumn{3}{c}{ordering (ms)} & \\
\cmidrule(lr){3-7}\cmidrule(lr){8-10}
Graph & $s$ & Cheap+T & Frozen & Reuse+T & Fresh & Neigh & Cheap & Reuse+T & Fresh & Inc/list\\
\midrule
Bitcoin-Alpha & 0.001 & 11.4 & 12.5 & 12.7 & 12.5 & 0.9 & 39.3 & 0.56 & 537 & 0.11\\
Gnutella05 & 0.001 & 7.0 & 11.5 & 11.7 & 12.4 & 5.7 & 92.4 & 1.73 & 593 & 0.07\\
Brightkite & 0.001 & 12.0 & 12.6 & 12.7 & 12.9 & $-$0.4 & 774.9 & 20.56 & 2862 & 0.08\\
GitHub & 0.001 & 12.0 & 12.2 & 12.3 & 12.1 & $-$0.5 & 903.2 & 45.61 & 3726 & 0.12\\
Email-Eu & 0.001 & 4.5 & 5.2 & 5.2 & 5.3 & 1.3 & 27.6 & 0.82 & 527 & 0.16\\
CollegeMsg & 0.001 & 7.3 & 8.0 & 8.1 & 8.3 & 1.3 & 28.8 & 0.60 & 528 & 0.12\\
LastFM-Asia & 0.001 & 7.0 & 9.0 & 9.3 & 9.7 & 2.1 & 83.3 & 1.68 & 578 & 0.07\\
Deezer & 0.001 & 6.0 & 8.9 & 9.2 & 9.9 & 3.1 & 331.7 & 6.90 & 1244 & 0.05\\
\midrule
Bitcoin-Alpha & 0.01 & 11.1 & 8.8 & 10.8 & 12.2 & 0.1 & 38.6 & 1.36 & 538 & 0.77\\
Gnutella05 & 0.01 & 6.8 & 4.4 & 8.0 & 12.1 & 6.0 & 92.2 & 4.06 & 594 & 0.54\\
Brightkite & 0.01 & 12.0 & 9.3 & 11.6 & 12.8 & $-$0.4 & 773.7 & 48.54 & 2872 & 0.64\\
GitHub & 0.01 & 12.0 & 10.2 & 11.5 & 12.2 & $-$0.4 & 892.8 & 102.87 & 3760 & 1.03\\
Email-Eu & 0.01 & 4.6 & 2.9 & 4.5 & 5.4 & 1.2 & 28.1 & 2.72 & 526 & 1.03\\
CollegeMsg & 0.01 & 7.1 & 5.1 & 7.0 & 8.1 & 1.4 & 30.5 & 1.84 & 528 & 0.85\\
LastFM-Asia & 0.01 & 7.2 & 2.9 & 6.3 & 9.6 & 1.5 & 80.0 & 3.52 & 583 & 0.57\\
Deezer & 0.01 & 6.0 & 1.4 & 5.5 & 9.8 & 3.3 & 323.7 & 16.23 & 1240 & 0.46\\
\midrule
Bitcoin-Alpha & 10 & 11.6 & $-$1.7 & 9.2 & 12.9 & 0.9 & 38.4 & 3.61 & 537 & 412.69\\
Gnutella05 & 10 & 6.8 & $-$30.1 & $-$1.8 & 12.0 & 6.1 & 93.2 & 9.70 & 591 & 232.53\\
Brightkite & 10 & 11.9 & $-$5.9 & 8.8 & 12.7 & $-$0.3 & 780.2 & 141.45 & 2886 & 353.18\\
GitHub & 10 & 12.0 & 3.5 & 10.8 & 12.1 & $-$0.5 & 882.1 & 248.15 & 3711 & 673.69\\
Email-Eu & 10 & 4.5 & $-$6.1 & 2.5 & 5.4 & 1.4 & 28.8 & 6.66 & 526 & 715.64\\
CollegeMsg & 10 & 6.9 & $-$6.8 & 4.4 & 8.0 & 1.3 & 30.0 & 5.25 & 529 & 460.66\\
LastFM-Asia & 10 & 7.1 & $-$22.5 & 1.1 & 9.7 & 1.8 & 79.8 & 10.01 & 581 & 243.14\\
Deezer & 10 & 6.1 & $-$32.4 & $-$2.8 & 9.9 & 2.7 & 332.6 & 38.86 & 1244 & 226.31\\
\bottomrule
\end{tabular}

\end{table}

\begin{table}[H]
\centering\small\setlength{\tabcolsep}{2pt}
\caption{Listing time and break-even probe cost. List: median listing time of the Asc and Reuse+T orders (marker kernel). $\gamma$: within-frame slope of listing time on $\Lam$ (ns per candidate test; within $R^2$ in parentheses), median over the five spacings, for the marker and the hash-probe kernel. $\gamma^\ast$: median break-even cost per candidate test at which Reuse+T's extra ordering time equals its saved tests. Last columns: mean per-frame time (ordering, building and listing) of Reuse+T relative to Asc, with a 95\% bootstrap interval, as the protocol specifies, from resampling the 40 frames (20 per chain, two chains) individually; consecutive frames of a chain are correlated, so these intervals can understate the uncertainty.}
\label{tab:reusetime}
\begin{tabular}{@{}lrrrrrrrr@{}}
\toprule
& \multicolumn{2}{c}{list (ms)} & \multicolumn{2}{c}{$\gamma$ (ns)} & \multicolumn{2}{c}{$\gamma^\ast_{\rm Reuse+T}$ (ns)} & \multicolumn{2}{c}{time Reuse+T / Asc}\\
\cmidrule(lr){2-3}\cmidrule(lr){4-5}\cmidrule(lr){6-7}\cmidrule(lr){8-9}
Graph & Asc & Reuse+T & marker & hash & $s{=}0.001$ & $s{=}10$ & $s{=}0.001$ & $s{=}10$\\
\midrule
Bitcoin-Alpha & 0.42 & 0.43 & $-$0.0 (0.01) & 6.7 (0.62) & 24 & 216 & 1.60 [1.52, 1.70] & 4.69 [4.40, 4.98]\\
Gnutella05 & 1.10 & 0.93 & 5.3 (0.84) & 7.3 (0.74) & 40 & 338 & 1.50 [1.43, 1.58] & 4.52 [4.34, 4.73]\\
Brightkite & 8.10 & 8.29 & $-$0.1 (0.00) & 4.9 (0.60) & 50 & 416 & 1.98 [1.88, 2.10] & 7.99 [7.52, 8.49]\\
GitHub & 11.83 & 12.13 & $-$0.2 (0.01) & 8.4 (0.70) & 69 & 422 & 2.55 [2.40, 2.70] & 9.75 [9.30, 10.21]\\
Email-Eu & 0.60 & 0.58 & 0.4 (0.14) & 7.3 (0.76) & 13 & 122 & 1.77 [1.67, 1.88] & 6.73 [6.24, 7.30]\\
CollegeMsg & 0.48 & 0.47 & 0.6 (0.19) & 6.1 (0.68) & 16 & 174 & 1.60 [1.51, 1.71] & 5.99 [5.56, 6.46]\\
LastFM-Asia & 1.00 & 0.88 & 3.5 (0.79) & 6.1 (0.69) & 37 & 273 & 1.53 [1.47, 1.59] & 5.24 [4.96, 5.57]\\
Deezer & 3.60 & 3.03 & 5.7 (0.86) & 6.9 (0.69) & 52 & 391 & 1.60 [1.53, 1.67] & 4.99 [4.80, 5.19]\\
\bottomrule
\end{tabular}

\end{table}

\begin{figure}[H]
\centering
\includegraphics[width=\textwidth]{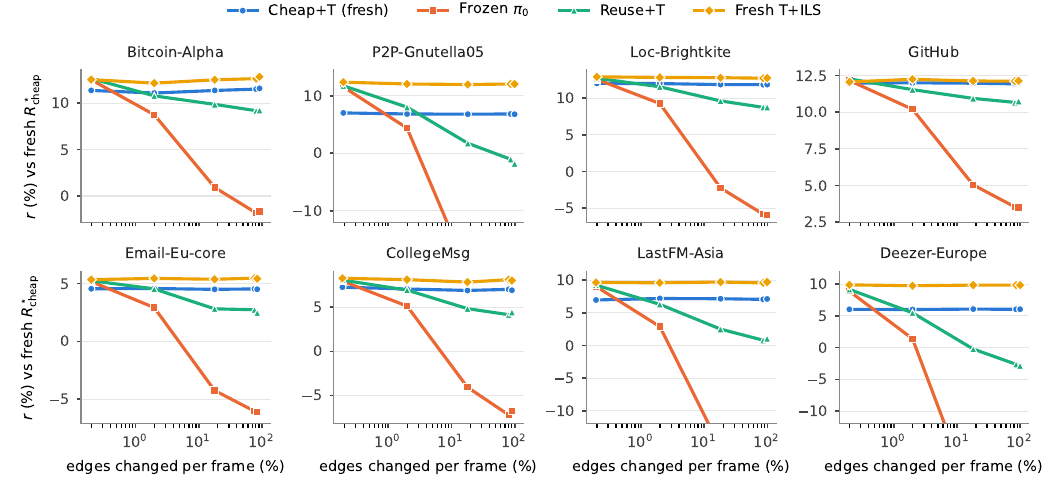}
\caption{Reduction $r$ relative to the fresh cheap reference versus the fraction of edges changed per frame (log scale), for the eight swap chains. Frozen values below $-12\%$ are clipped.}
\label{fig:reuse}
\end{figure}

\section{Additional Tables}\label{app:extra}
\begin{table}[H]
\centering\small\setlength{\tabcolsep}{4pt}
\caption{What is prior or classical and what is new.}
\label{tab:novelty}
\begin{tabular}{@{}>{\raggedright\arraybackslash}p{0.62\columnwidth}>{\raggedright\arraybackslash}p{0.35\columnwidth}@{}}
\toprule
Result & Source\\
\midrule
Adaptive lister; its candidate-test count is $\Lam$, Eq.~\eqref{eq:work} (time $\Theta(n+m+\Lam)$ with constant-time membership tests) & \citet{yu2020aot}\\
Non-adaptive costs: NP-hardness, orders & \citet{lecuyer2023tailored}\\
$\Cpm$: Min UnCut connection, UGC inapproximability; $\Cpp$ over orientations & \citet[\S 5.3.3]{lecuyer2023thesis}\\
Polytope vertices are the orders & \citet{stanley1980decompositions}\\
Linear minimization by sorting & \citet{edmonds1970submodular}\\
Tangent step as linearization & \citet{mangasarian1997concave}\\
Logistic map onto interior (Thm.~\ref{thm:landscape}i) & \citet{wainwright2008graphical}\\
\midrule
$\Lam$ concave in forward degrees; certified rounding (Thm.~\ref{thm:rounding}) & new\\
$\Lopt{S_2(H)}=\beta(H)$; hardness and inapproximability for $\Lam$ (Thm.~\ref{thm:hard}) & new\\
Program without acyclicity (Prop.~\ref{prop:bounds}b) & new\\
Stationary points are maxima (Thm.~\ref{thm:landscape}ii) & new\\
Limit of degree-based orders (Thm.~\ref{thm:mf}) & new\\
Degree loses a constant factor (Prop.~\ref{prop:cubic}) & new\\
V-shapes; repair (Sec.~\ref{sec:random}, App.~\ref{app:reuse}) & new, empirical\\
\bottomrule
\end{tabular}
\end{table}

\paragraph{Further random-graph results.} On the random graphs of Table~\ref{tab:scaling} at $n=10^5$, \DOVE ends 5.4--19.9 percentage points of reduction short of \TILS; there \TILS has about five times \DOVE's time, so this shows attainable improvement, not equal-compute dominance. From $n=10^3$ to $10^4$ the \TILS reductions fall by 4--5 points on the power laws (Table~\ref{tab:scaling-full}). Greedy-coloring orders never match \TILS; starting ILS from the pre-declared DEG-desc coloring adds 0.1--1.1 points at $n=10^5$ (Table~\ref{tab:scaling-full}) but not on public graphs (Section~\ref{sec:posthoc}).

\paragraph{Tables.} Table~\ref{tab:novelty} separates prior ingredients from new results. Table~\ref{tab:exact} gives the excess work over the certified optimum (Section~\ref{sec:exact}) and Table~\ref{tab:scaling} the random-graph comparison of Section~\ref{sec:random}. As a correctness check, the program of Proposition~\ref{prop:bounds}(b) matches brute force on all 1{,}245 unlabeled graphs with at least one edge on two to seven vertices. Table~\ref{tab:dev} lists the 32 development graphs, Table~\ref{tab:suite} summarizes them by group (development, validation and audit), and Table~\ref{tab:ablation} gives the ablation of Section~\ref{sec:ablation}. Tables~\ref{tab:scaling-full} and~\ref{tab:scaling-colour} give all columns of the random-graph comparison, and Table~\ref{tab:ablation-public} the per-seed dispersion of the public-graph ablation: without the tangent prefix, ILS differs by $+0.10$ points, with per-graph seed SDs of up to 0.33; the differences between insertion variants (0.1--0.5 points) are carried mostly by P2P-Gnutella05.

\begin{table}[H]
\centering\small\setlength{\tabcolsep}{2.5pt}
\caption{Excess work over the solver-certified optimum on 27 small instances (0.2\,s per heuristic run, three seeds). Mean and maximum excess are taken over instances of the seed-mean excess. \# optimal: instances on which the seed-mean work equals the optimum, that is, every seed reached it.}
\label{tab:exact}
\begin{tabular}{lrrr}
\toprule
Method & Mean excess & Max excess & \# optimal\\
\midrule
Reference $R^\star$ & 26.0\% & 120\% & 1\\
Tangent refinement (T) & 21.1\% & 120\% & 2\\
Strict insertion & 10.1\% & 80\% & 11\\
T+strict insertion & 10.5\% & 80\% & 10\\
Orientation-space search & 2.1\% & 25\% & 21\\
Iterated local search (ILS) & 0.0\% & 0\% & 27\\
\bottomrule
\end{tabular}
\end{table}

\begin{table}[H]
\centering\small\setlength{\tabcolsep}{2.5pt}
\caption{Reduction of $\Lam$ (\%) relative to the best-found degree-based order on sampled graphs with $n=10^5$ (ER: $G(n,8/(n-1))$; PL$\gamma$: erased configuration model with $p_k\propto k^{-\gamma}$ on $[2,100]$), mean$\pm$SD over five graph seeds; all methods start from that order. \TILS uses 100\,$\mu$s per edge (about $5\times$ the time of \DOVE). Col.: best of four greedy-coloring orders, without search. Table~\ref{tab:scaling-full} gives $n=10^3,10^4$ and further variants.}
\label{tab:scaling}
\begin{tabular}{lrrrr}
\toprule
& T & \DOVE & \TILS & Col.\\
\midrule
ER8 & 23.1{\tiny$\pm$0.0} & 27.0{\tiny$\pm$0.1} & 36.1{\tiny$\pm$0.1} & 14.5{\tiny$\pm$0.2}\\
PL2.1 & 19.8{\tiny$\pm$0.3} & 22.4{\tiny$\pm$0.4} & 27.8{\tiny$\pm$0.1} & 4.8{\tiny$\pm$0.2}\\
PL2.5 & 25.0{\tiny$\pm$0.4} & 29.8{\tiny$\pm$0.6} & 41.0{\tiny$\pm$0.6} & 17.6{\tiny$\pm$1.2}\\
PL3.0 & 28.8{\tiny$\pm$0.3} & 38.0{\tiny$\pm$0.2} & 57.9{\tiny$\pm$0.4} & 42.5{\tiny$\pm$0.3}\\
\bottomrule
\end{tabular}
\end{table}

\begin{table}[H]
\centering\small\setlength{\tabcolsep}{3.2pt}
\caption{Mean work reduction (\%) relative to $R^\star$ on the 32 development graphs, by group (unweighted means over graphs); \TILS receives \DOVE's measured time on each graph. I/T/R: improved/tied/regressed for \TILS. All three groups informed the method; none is held out.}
\label{tab:suite}
\begin{tabular}{lrrrrr}
\toprule
Group & Graphs & \DOVE & T & \TILS & I/T/R\\
\midrule
Development & 13 & 1.67 & 1.63 & 2.20 & 13/0/0 \\
Validation & 11 & 5.09 & 5.22 & 6.11 & 11/0/0 \\
Audit & 8 & 7.72 & 7.88 & 10.83 & 8/0/0 \\
\midrule
All & 32 & 4.36 & 4.43 & 5.70 & 32/0/0 \\
 
\bottomrule
\end{tabular}
\end{table}

\clearpage
{\small\setlength{\tabcolsep}{3pt}
\begin{longtable}{@{}llrrrlrrr@{}}
\caption{The 32 development graphs, all from SNAP \citep{leskovec2014snap}; Ego-$k$ is the ego network of user $k$ in ego-Facebook. Names follow the repository. Group: stage of development at which the graph was used (all groups informed the method). $\Lam_{R^\star}$: reference work; $R^\star$: winning candidate (Neigh-s$i$: stored Neigh order of seed $i$; rev.: reversed). T, \DOVE, \TILS: reduction (\%) relative to $R^\star$; \TILS: mean$\pm$SD over three seeds. Unweighted means: \SuiteTMean\ (T), \SuiteDoveMean\ (\DOVE) and \SuiteTIKMean\ (\TILS).}\label{tab:dev}\\
\toprule
Graph & Group & $n$ & $m$ & $\Lam_{R^\star}$ & $R^\star$ & T & \DOVE & \TILS\\
\midrule\endfirsthead
\caption[]{The 32 development graphs (continued).}\\
\toprule
Graph & Group & $n$ & $m$ & $\Lam_{R^\star}$ & $R^\star$ & T & \DOVE & \TILS\\
\midrule\endhead
\bottomrule\endfoot
CA-CondMat & Development & 23,133 & 93,439 & 316,262 & Neigh-s4, rev. & 3.07 & 3.16 & 4.23{\tiny$\pm$0.00}\\
CA-GrQc & Development & 5,241 & 14,484 & 54,867 & Neigh-s2 & 0.91 & 0.85 & 1.45{\tiny$\pm$0.01}\\
CA-HepPh & Development & 12,006 & 118,489 & 3,821,951 & Neigh-s2 & 2.29 & 2.18 & 2.53{\tiny$\pm$0.00}\\
Ego-0 & Development & 333 & 2,519 & 16,307 & Neigh-s5, rev. & 2.89 & 2.96 & 3.99{\tiny$\pm$0.00}\\
Ego-107 & Development & 1,034 & 26,749 & 624,726 & Neigh-s4, rev. & 1.58 & 1.68 & 1.82{\tiny$\pm$0.01}\\
Ego-1684 & Development & 786 & 14,024 & 201,862 & Neigh-s4 & 1.88 & 2.04 & 2.38{\tiny$\pm$0.01}\\
Ego-1912 & Development & 747 & 30,025 & 1,095,493 & Neigh-s2, rev. & 1.11 & 1.24 & 1.32{\tiny$\pm$0.00}\\
Ego-3437 & Development & 534 & 4,813 & 31,251 & Neigh-s5 & 1.58 & 1.82 & 2.35{\tiny$\pm$0.01}\\
Ego-348 & Development & 224 & 3,192 & 34,870 & Neigh-s4 & 1.83 & 1.87 & 1.98{\tiny$\pm$0.00}\\
Ego-3980 & Development & 52 & 146 & 221 & Neigh-s2 & 1.36 & 1.36 & 1.36{\tiny$\pm$0.00}\\
Ego-414 & Development & 150 & 1,693 & 12,871 & Neigh-s5 & 0.64 & 0.66 & 0.99{\tiny$\pm$0.00}\\
Ego-686 & Development & 168 & 1,656 & 11,657 & Neigh-s5, rev. & 1.66 & 1.96 & 3.17{\tiny$\pm$0.01}\\
Ego-698 & Development & 61 & 270 & 748 & Neigh-s1 & 0.40 & 0.00 & 1.02{\tiny$\pm$0.08}\\
CA-AstroPh & Validation & 18,771 & 198,050 & 2,502,046 & Neigh-s4 & 3.90 & 3.99 & 4.41{\tiny$\pm$0.01}\\
CA-HepTh & Validation & 9,875 & 25,973 & 53,365 & Neigh-s2 & 2.47 & 2.64 & 3.87{\tiny$\pm$0.06}\\
Cit-HepTh & Validation & 27,769 & 352,285 & 4,212,943 & Neigh-s1, rev. & 2.53 & 1.72 & 2.91{\tiny$\pm$0.00}\\
Com-Amazon & Validation & 334,863 & 925,872 & 1,084,032 & Neigh-s5, rev. & 0.57 & 0.00 & 2.44{\tiny$\pm$0.00}\\
Com-DBLP & Validation & 317,080 & 1,049,866 & 3,556,614 & Neigh-s4, rev. & 3.71 & 3.52 & 4.80{\tiny$\pm$0.00}\\
Email-Enron & Validation & 36,692 & 183,831 & 2,158,516 & Neigh-s1 & 10.39 & 10.38 & 10.79{\tiny$\pm$0.00}\\
P2P-Gnutella31 & Validation & 62,586 & 147,892 & 203,725 & Neigh-s2 & 4.23 & 4.75 & 7.62{\tiny$\pm$0.00}\\
Soc-Epinions1 & Validation & 75,879 & 405,740 & 7,985,079 & Neigh-s4 & 12.21 & 11.64 & 12.33{\tiny$\pm$0.00}\\
Soc-Slashdot0811 & Validation & 77,360 & 469,180 & 4,906,888 & Neigh-s4 & 11.21 & 11.32 & 11.55{\tiny$\pm$0.00}\\
Web-NotreDame & Validation & 325,729 & 1,090,108 & 9,563,275 & Neigh-s1, rev. & 0.30 & 0.00 & 0.49{\tiny$\pm$0.00}\\
Wiki-Vote & Validation & 7,115 & 100,762 & 1,891,565 & Neigh-s4 & 5.96 & 6.00 & 6.04{\tiny$\pm$0.00}\\
AS-Caida & Audit & 26,475 & 53,381 & 131,589 & Neigh-s2 & 16.72 & 16.18 & 17.53{\tiny$\pm$0.01}\\
Bitcoin-OTC & Audit & 5,881 & 21,492 & 115,800 & Neigh-s1 & 10.08 & 10.40 & 10.79{\tiny$\pm$0.00}\\
Cit-HepPh & Audit & 34,546 & 420,877 & 3,972,842 & Neigh-s1, rev. & 2.85 & 2.94 & 3.43{\tiny$\pm$0.03}\\
Email-EuAll & Audit & 265,009 & 364,481 & 1,276,904 & Degree & 19.76 & 19.28 & 21.07{\tiny$\pm$0.00}\\
Loc-Gowalla & Audit & 196,591 & 950,327 & 8,631,304 & Neigh-s5 & 10.10 & 9.55 & 10.62{\tiny$\pm$0.00}\\
P2P-Gnutella04 & Audit & 10,876 & 39,994 & 73,280 & Neigh-s2 & 2.98 & 3.38 & 5.09{\tiny$\pm$0.05}\\
RoadNet-PA & Audit & 1,088,092 & 1,541,898 & 395,391 & Neigh-s4, rev. & 0.05 & 0.00 & 17.48{\tiny$\pm$0.12}\\
Web-Stanford & Audit & 281,903 & 1,992,636 & 12,284,557 & Neigh-s1, rev. & 0.54 & 0.00 & 0.63{\tiny$\pm$0.00}\\
 
\end{longtable}}

\begin{table}[H]
\centering\small\setlength{\tabcolsep}{3pt}
\caption{Ablation on the four public development graphs and the synthetic graphs: mean reduction (\%) relative to $R^\star$ after a final tangent pass; five seeds. Public: 4 graphs, 10\,s; synthetic development: 6 graphs, 1\,s; synthetic prospective: 9 graphs generated after tuning (3 bipartite with optimum 0), 1\,s. The six tuning graphs are separate instances of the development families.}
\label{tab:ablation}
\begin{tabular}{lrrr}
\toprule
Method & Public & Synth.\ dev. & Synth.\ prosp.\\
\midrule
T+strict insertion & 7.48 & 10.82 & 9.68\\
Strict insertion & 7.37 & 11.06 & 9.52\\
Orientation-space search & 7.15 & 12.32 & 11.02\\
ILS & 7.85 & 16.72 & 11.92\\
\TILS & 7.75 & 16.51 & 11.75\\
\bottomrule
\end{tabular}
\end{table}

\begin{table}[H]
\centering\small\setlength{\tabcolsep}{4pt}
\caption{Sampled graphs: reduction (\%) of the search methods relative to the best-found degree-based order, mean$\pm$SD over graph seeds (ER: $G(n,8/(n-1))$). ILS-T: ILS without the tangent prefix. C+ILS-T: ILS-T (Numba kernel, 100\,$\mu$s per edge, one ILS seed per graph) from the DEG-desc coloring order (pre-declared) or from the per-graph best coloring order of Table~\ref{tab:scaling-colour} (selection cost charged); construction charged to the budget.}
\label{tab:scaling-full}
\begin{tabular}{@{}llrrrr@{}}
\toprule
& & & & \multicolumn{2}{c}{C+ILS-T from}\\
\cmidrule(l){5-6}
& $n$ & \TILS & ILS-T & DEG-desc & best coloring\\
\midrule
ER8 & $10^3$ & 36.8$\pm$0.8 & 36.8$\pm$0.4 & 36.9$\pm$0.5 & 36.8$\pm$0.4\\
 & $10^4$ & 36.3$\pm$0.3 & 36.4$\pm$0.7 & 36.3$\pm$0.5 & 36.5$\pm$0.6\\
 & $10^5$ & 36.1$\pm$0.1 & 36.1$\pm$0.1 & 36.2$\pm$0.1 & 36.4$\pm$0.1\\
\addlinespace
PL2.1 & $10^3$ & 31.5$\pm$3.9 & 31.4$\pm$3.3 & 32.0$\pm$3.8 & 31.1$\pm$3.0\\
 & $10^4$ & 26.9$\pm$1.1 & 27.3$\pm$1.2 & 27.4$\pm$1.3 & 27.3$\pm$1.2\\
 & $10^5$ & 27.8$\pm$0.1 & 28.5$\pm$0.3 & 28.7$\pm$0.2 & 28.3$\pm$0.2\\
\addlinespace
PL2.5 & $10^3$ & 45.3$\pm$3.5 & 45.6$\pm$3.9 & 44.6$\pm$3.3 & 45.9$\pm$3.5\\
 & $10^4$ & 40.2$\pm$2.2 & 40.8$\pm$2.6 & 40.5$\pm$2.4 & 41.0$\pm$2.0\\
 & $10^5$ & 41.0$\pm$0.6 & 41.5$\pm$0.5 & 41.7$\pm$0.7 & 41.8$\pm$0.6\\
\addlinespace
PL3.0 & $10^3$ & 61.2$\pm$4.8 & 60.4$\pm$4.8 & 61.1$\pm$4.6 & 62.3$\pm$4.2\\
 & $10^4$ & 57.0$\pm$2.5 & 56.8$\pm$2.1 & 58.5$\pm$1.1 & 58.8$\pm$1.3\\
 & $10^5$ & 57.9$\pm$0.4 & 58.2$\pm$0.3 & 59.0$\pm$0.3 & 59.8$\pm$0.5\\
\bottomrule
\end{tabular}
\end{table}

\begin{table}[H]
\centering\small\setlength{\tabcolsep}{4pt}
\caption{Sampled graphs: reduction (\%) of greedy-coloring orders relative to the best-found degree-based order, without search; mean$\pm$SD over graph seeds. Greedy coloring with vertices in reverse smallest-last (SL) or descending-degree (DEG) order, color classes ascending or descending (each has the same work as the reversal of the other direction). Best: best of the four by exact work.}
\label{tab:scaling-colour}
\begin{tabular}{@{}llrrrrr@{}}
\toprule
& $n$ & SL-asc & SL-desc & DEG-asc & DEG-desc & Best\\
\midrule
ER8 & $10^3$ & 11.6$\pm$1.0 & 15.2$\pm$0.5 & 9.5$\pm$0.9 & 10.6$\pm$1.0 & 15.2$\pm$0.5\\
 & $10^4$ & 11.0$\pm$0.7 & 14.4$\pm$0.7 & 10.0$\pm$0.5 & 11.0$\pm$0.5 & 14.4$\pm$0.7\\
 & $10^5$ & 10.9$\pm$0.2 & 14.5$\pm$0.2 & 9.8$\pm$0.2 & 10.9$\pm$0.1 & 14.5$\pm$0.2\\
\addlinespace
PL2.1 & $10^3$ & 11.0$\pm$3.3 & $-$1.5$\pm$6.0 & 7.5$\pm$4.2 & $-$10.7$\pm$5.2 & 11.0$\pm$3.3\\
 & $10^4$ & 4.5$\pm$0.6 & $-$14.3$\pm$3.6 & 2.8$\pm$0.7 & $-$18.4$\pm$4.0 & 4.5$\pm$0.6\\
 & $10^5$ & 4.8$\pm$0.2 & $-$12.6$\pm$0.6 & 4.0$\pm$0.2 & $-$16.6$\pm$0.6 & 4.8$\pm$0.2\\
\addlinespace
PL2.5 & $10^3$ & 17.1$\pm$1.7 & 21.6$\pm$7.2 & 10.6$\pm$3.3 & 16.6$\pm$4.9 & 22.4$\pm$6.4\\
 & $10^4$ & 13.8$\pm$2.2 & 15.4$\pm$4.5 & 9.5$\pm$1.3 & 9.4$\pm$4.8 & 15.7$\pm$4.1\\
 & $10^5$ & 14.4$\pm$0.7 & 17.6$\pm$1.2 & 10.0$\pm$0.4 & 12.4$\pm$1.1 & 17.6$\pm$1.2\\
\addlinespace
PL3.0 & $10^3$ & 29.1$\pm$9.9 & 43.3$\pm$7.1 & 20.7$\pm$11.4 & 31.8$\pm$7.1 & 43.3$\pm$7.1\\
 & $10^4$ & 25.9$\pm$3.4 & 40.6$\pm$1.4 & 18.3$\pm$2.6 & 33.3$\pm$1.4 & 40.6$\pm$1.4\\
 & $10^5$ & 27.8$\pm$0.3 & 42.5$\pm$0.3 & 18.8$\pm$0.3 & 33.9$\pm$0.5 & 42.5$\pm$0.3\\
\bottomrule
\end{tabular}
\end{table}

\begin{table}[H]
\centering\small\setlength{\tabcolsep}{4pt}
\caption{Public graphs at 10\,s: reduction (\%) relative to $R^\star$ after a final tangent pass, mean$\pm$SD over five seeds. SI: strict insertion; T+SI: tangent prefix then strict insertion; Orient.: orientation-space search. $\Delta$: ILS minus \TILS, paired by seed (mean$\pm$SD of the five per-seed differences, percentage points).}
\label{tab:ablation-public}
\begin{tabular}{lrrrrrr}
\toprule
Graph & T+SI & SI & Orient. & ILS & \TILS & $\Delta$\\
\midrule
Bitcoin-Alpha & 9.60$\pm$0.09 & 9.59$\pm$0.14 & 9.58$\pm$0.02 & 9.80$\pm$0.25 & 9.61$\pm$0.03 & 0.18$\pm$0.25\\
P2P-Gnutella05 & 4.62$\pm$0.18 & 4.12$\pm$0.06 & 3.68$\pm$0.02 & 5.84$\pm$0.27 & 5.57$\pm$0.12 & 0.27$\pm$0.33\\
Loc-Brightkite & 9.29$\pm$0.05 & 9.38$\pm$0.05 & 8.95$\pm$0.01 & 9.40$\pm$0.04 & 9.41$\pm$0.01 & $-$0.01$\pm$0.05\\
GitHub & 6.41$\pm$0.00 & 6.38$\pm$0.00 & 6.37$\pm$0.00 & 6.38$\pm$0.00 & 6.41$\pm$0.00 & $-$0.04$\pm$0.00\\
\midrule
Mean & 7.48 & 7.37 & 7.15 & 7.85 & 7.75 & 0.10\\
\bottomrule
\end{tabular}
\end{table}

{\small\setlength{\tabcolsep}{4pt}
\begin{longtable}{@{}lrrrrrrrr@{}}
\caption{The 22 small instances without a certified optimum (HiGHS, no warm start; 60\,s for $n\le32$, 120\,s otherwise). LB: solver lower bound; MIP UB: solver incumbent objective; MIP order: exact work of the order obtained from the incumbent by tangent rounding; best heuristic: lowest work of any heuristic run (0.2\,s, seeds 1--3); ILS: seed mean. Gap: relative bound width $(\mathrm{UB}_{\rm best}-\lceil\mathrm{LB}\rceil)/\mathrm{UB}_{\rm best}$, with $\mathrm{UB}_{\rm best}$ the lower of the MIP order and the best heuristic run. It upper-bounds the shortfall from optimality as a fraction of the best order's work, $(\mathrm{UB}_{\rm best}-\Lopt{G})/\mathrm{UB}_{\rm best}$; it does not estimate the actual shortfall. Rounding LB up is valid for an exact lower bound because the optimum is an integer; for the solver's floating-point bound it assumes that the numerical error does not cross an integer.}\label{tab:uncertified}\\
\toprule
Instance & $n$ & $m$ & LB & MIP UB & MIP order & Best heur. & ILS mean & Gap\\
\midrule\endfirsthead
\caption[]{The small instances without a certified optimum (continued).}\\
\toprule
Instance & $n$ & $m$ & LB & MIP UB & MIP order & Best heur. & ILS mean & Gap\\
\midrule\endhead
\bottomrule\endfoot
ws-24-2024 & 24 & 72 & 74.7 & 90 & 90 & 89 & 89.0 & 15.7\% \\
ws-32-2032 & 32 & 96 & 81.0 & 113 & 111 & 110 & 110.0 & 26.4\% \\
sbm-32-2032 & 32 & 100 & 83.3 & 114 & 114 & 114 & 114.0 & 26.3\% \\
geometric-32-2032 & 32 & 84 & 92.0 & 104 & 104 & 104 & 104.0 & 11.5\% \\
ws-48-1048 & 48 & 144 & 109.0 & 159 & 159 & 158 & 158.0 & 31.0\% \\
sbm-48-1048 & 48 & 253 & 288.0 & 662 & 609 & 562 & 565.7 & 48.8\% \\
geometric-48-1048 & 48 & 139 & 157.0 & 199 & 199 & 199 & 199.0 & 21.1\% \\
er-48-2048 & 48 & 151 & 120.7 & 169 & 169 & 168 & 168.3 & 28.0\% \\
ws-48-2048 & 48 & 144 & 116.3 & 157 & 156 & 154 & 154.0 & 24.0\% \\
sbm-48-2048 & 48 & 260 & 259.0 & 698 & 641 & 592 & 594.0 & 56.2\% \\
geometric-48-2048 & 48 & 127 & 142.0 & 157 & 155 & 155 & 155.0 & 8.4\% \\
er-64-1064 & 64 & 224 & 126.5 & 332 & 302 & 273 & 274.0 & 53.5\% \\
ws-64-1064 & 64 & 192 & 160.7 & 224 & 219 & 203 & 203.0 & 20.7\% \\
powerlaw-cluster-64-1064 & 64 & 181 & 97.0 & 185 & 185 & 184 & 184.3 & 47.3\% \\
sbm-64-1064 & 64 & 457 & 666.0 & 1892 & 1713 & 1450 & 1450.0 & 54.1\% \\
geometric-64-1064 & 63 & 206 & 273.0 & 335 & 330 & 317 & 317.0 & 13.9\% \\
regular-64-1064 & 64 & 128 & 47.5 & 54 & 54 & 54 & 54.0 & 11.1\% \\
er-64-2064 & 64 & 210 & 132.5 & 258 & 254 & 221 & 221.0 & 39.8\% \\
ws-64-2064 & 64 & 192 & 163.6 & 234 & 225 & 209 & 209.0 & 21.5\% \\
powerlaw-cluster-64-2064 & 64 & 183 & 135.7 & 186 & 186 & 186 & 186.0 & 26.9\% \\
sbm-64-2064 & 64 & 454 & 707.0 & 1929 & 1675 & 1505 & 1505.7 & 53.0\% \\
geometric-64-2064 & 64 & 185 & 225.0 & 278 & 274 & 267 & 267.0 & 15.7\% \\
\end{longtable}}

\end{document}